\documentclass[11pt,a4paper]{article}
\usepackage[T1]{fontenc}
\usepackage[utf8]{inputenc}
\usepackage[margin=27mm,headheight=14pt]{geometry}
\usepackage{amsmath,amssymb,amsthm,mathtools}
\usepackage{newtxtext,newtxmath}
\usepackage{microtype}
\usepackage{graphicx}
\graphicspath{{figures/}{./}}
\usepackage[backend=biber,style=authoryear,natbib=true,maxcitenames=2,maxbibnames=99]{biblatex}
\AtBeginBibliography{\small}
\usepackage{enumitem}
\usepackage{booktabs,longtable,array,placeins,needspace}
\usepackage{fancyhdr}
\usepackage[hidelinks]{hyperref}
\hypersetup{pdftitle={Kriging Beyond Gaussianity: Fixed-Domain Asymptotic Optimality},
 pdfsubject={Fixed-domain comparison of linear and optimal nonlinear prediction},
 pdfauthor={Moreno Bevilacqua and Hector Araya}}
\setlist[enumerate]{itemsep=3pt,topsep=5pt}
\numberwithin{equation}{section}
\newtheorem{theorem}{Theorem}
\newtheorem{proposition}{Proposition}
\newtheorem{lemma}{Lemma}
\newtheorem{corollary}{Corollary}
\theoremstyle{definition}
\newtheorem{assumption}{Assumption}
\theoremstyle{remark}
\newtheorem{remark}{Remark}
\newcommand{\R}{\mathbb R}
\newcommand{\N}{\mathbb N}
\newcommand{\Nzero}{\mathbb N_0}
\newcommand{\E}{\mathbb E}
\newcommand{\PP}{\mathbb P}
\newcommand{\gq}{\gamma_q}
\newcommand{\bk}{\boldsymbol k}
\newcommand{\bz}{\boldsymbol z}
\newcommand{\bG}{\boldsymbol G}
\newcommand{\bZ}{\boldsymbol Z}
\newcommand{\dd}{\,\mathrm d}
\newcommand{\norm}[1]{\left\lVert#1\right\rVert}
\newcommand{\ind}{\mathbf 1}
\newcommand{\MSPEL}{\operatorname{MSPE}_n^{\mathrm L}}
\newcommand{\MSPEN}{\operatorname{MSPE}_n^{\mathrm{NL}}}
\newcommand{\MSPEO}{\operatorname{MSPE}_n^{\mathrm O}}
\newcommand{\MSPER}{\operatorname{MSPE}_n^{\rho}}
\newcommand{\YL}{\widehat Y_n^{\mathrm L}}
\newcommand{\YN}{\widehat Y_n^{\mathrm{NL}}}
\newcommand{\YO}{\widehat Y_n^{\mathrm O}}
\newcommand{\YR}{\widehat Y_n^{\rho}}
\DeclareMathOperator{\Var}{Var}
\DeclareMathOperator{\Cov}{Cov}
\DeclareMathOperator{\dist}{dist}
\DeclareMathOperator{\asinh}{asinh}

\newcommand{\numericalpanel}[2]{%
  \includegraphics[width=\linewidth]{#1}%
}
\allowdisplaybreaks[1]
\title{Kriging Beyond Gaussianity: Fixed-Domain Asymptotic Optimality}
\author{Moreno Bevilacqua \qquad H\'ector Araya\\
\small Universidad Adolfo Ib\'a\~nez, Chile}
\date{}
\begin{document}
\maketitle
\vspace{-1.4em}

\begin{abstract}
Under squared-error loss, the optimal predictor of a non-Gaussian random
field is its conditional mean, which may be difficult to compute. We consider
a broad class of non-Gaussian random fields obtained by transforming finitely
many independent Gaussian fields sharing a common correlation function, and
study when kriging entails no first-order loss in mean squared prediction error
(MSPE) under fixed-domain asymptotics. We compare the observed-data
conditional mean, kriging with the observed covariance, kriging with the
latent correlation, and a latent-information oracle. These predictors satisfy
a finite-sample MSPE ordering. Under conditions linking transformation
regularity to the high-frequency decay of the latent spectral density, all
four MSPEs are asymptotically equivalent. For radially nonincreasing spectral
densities with polynomial tails, the required regularity order is independent
of both the spatial dimension and the number of latent fields. Exponential
correlation also yields finite-sample bounds. Examples include Weibull,
skew-Gaussian, and copula-based fields, while numerical comparisons illustrate
finite-sample differences. Counterexamples show why covariance structure and
first-order transformation regularity alone are insufficient. Finally,
standardized prediction errors converge to a Gaussian scale mixture, which is
standard normal exactly when the transformation's gradient norm is almost
surely constant.
\end{abstract}


\section{Introduction}\label{sec:introduction}

Kriging---or Gaussian process regression---is a standard tool for prediction
in the analysis of random fields, computer experiments, and machine learning
\citep{Cressie:1993,williams2006gaussian,cp2003,wendland2004scattered}.
With a known mean and covariance, simple kriging is the best linear predictor,
whether or not the underlying field is Gaussian. Under Gaussianity, it also
coincides with the conditional mean and therefore minimizes mean squared
prediction error (MSPE) among all predictors based on the observations.

For a non-Gaussian field, the optimal predictor under squared-error loss
remains the conditional mean. In general, however, this predictor depends on
the full joint distribution of the target and the observations and may be
unavailable in closed form or computationally difficult to evaluate. Kriging
provides a simpler alternative because it requires only first- and second-order
structure. A second simplification is possible in latent Gaussian models: even
when the covariance of the observed non-Gaussian field is difficult to derive
or evaluate, one may compute kriging weights directly from the latent Gaussian
correlation and apply them to the observed data. We study the predictive cost
of both simplifications under fixed-domain asymptotics.

Specifically, let
\begin{equation}\label{eq:intro-general-model}
Y(s)=F\{Z_1(s),\ldots,Z_q(s)\},
\qquad s\in D,
\end{equation}
where $Z_1,\ldots,Z_q$ are independent copies of a standardized stationary
Gaussian random field with common correlation function $\rho$, and
$F:\mathbb{R}^q\to\mathbb{R}$ is measurable. The transformation need not be
invertible or monotone. Hence, particularly when $q>1$, observing $Y$ need not
reveal the underlying latent Gaussian vector.

An important motivation for studying the class
\eqref{eq:intro-general-model} is that it encompasses many non-Gaussian
random-field constructions proposed in the literature. Examples include
log-Gaussian fields \citep{DeOliveira2006}, skew-Gaussian fields
\citep{ZhangElShaarawi2010}, Tukey $g$-and-$h$ fields
\citep{XuGenton2017}, sinh--arcsinh fields
\citep{JonesPewsey2009,BlasiEtAl2022}, Gamma and Weibull fields
\citep{BevilacquaEtAl2020}, Bernoulli threshold fields
\citep{HeagertyLele1998}, Gaussian-copula random fields
\citep{MasarottoVarin2012}, and more complex copula-based constructions
\citep{BevilacquaEtAl2024}, to mention only a few.

A useful feature of the Gaussian-functional representation  \eqref{eq:intro-general-model} is that,
under suitable conditions on $F$, regularity properties of the latent
Gaussian field are inherited by the transformed field. In particular,
mean-square continuity and mean-square differentiability are preserved under
the conditions stated in Proposition~\ref{prop:ms-regularity}; related
sample-path consequences follow from general second-order theory
\citep{Scheuerer2010}. This inheritance makes it possible to control the
regularity of the non-Gaussian field  $Y$ through  flexible correlation families with explicit smoothness parameters, such
as the Mat\'ern \citep{ppmm} and generalized Wendland families \citep{BevilacquaEtAl2019}.

We compare three prediction procedures based on the observed non-Gaussian
data: the optimal conditional mean, simple kriging based on the true
covariance of the observed field, and a simpler kriging predictor whose
weights are computed directly from the latent Gaussian correlation. To
quantify the possible value of the unobserved latent variables, we also
introduce the conditional mean given the latent Gaussian observations as a
theoretical oracle. For every finite design, this oracle has no larger MSPE
than the optimal predictor based on the observed non-Gaussian data, which in
turn has no larger MSPE than true-covariance kriging; the latter has no larger
MSPE than latent-correlation kriging.

Our main result shows that, under explicit conditions linking the
high-frequency behaviour of the latent spectral density to the regularity of the
transformation, these finite-sample differences disappear to first order.
Writing $\operatorname{MSPE}_n^{\mathrm{NL}}$ for the MSPE of the optimal
conditional mean based on the observed data,
$\operatorname{MSPE}_n^{\mathrm L}$ for that of true-covariance kriging, and
$\operatorname{MSPE}_n^{\rho}$ for that of latent-correlation kriging, we show
that
\begin{equation}\label{eq:intro-main-equivalence}
\operatorname{MSPE}_n^{\rho}
\sim
\operatorname{MSPE}_n^{\mathrm L}
\sim
\operatorname{MSPE}_n^{\mathrm{NL}}
\end{equation}
under fixed-domain asymptotics. Thus both kriging procedures are
asymptotically equivalent in MSPE to the optimal conditional mean based on
the observed non-Gaussian data. The same first-order MSPE is also attained by
the theoretical predictor based on the latent Gaussian observations. Hence
the additional information contained in the latent variables becomes
asymptotically negligible, even for multilatent and noninvertible
constructions in which those variables cannot in general be recovered from
the observed field. The result is therefore a relative-efficiency statement,
stronger than the mere fact that all MSPEs vanish as observations become dense
in a fixed region.

The closest precedent is the scalar transformed-Gaussian setting studied by
\citet{HiranoYajima2013}. For a Mat\'ern latent field, they establish, under
weighted Hermite conditions, asymptotic MSPE equivalence between linear
prediction based on the covariance of the transformed field and a
conditional-mean benchmark based on the latent Gaussian observations. 

The present development extends this prediction comparison in several
directions. Most importantly, it allows finitely many latent Gaussian fields
and noninvertible transformations, so that the latent vector need not be
recoverable from the observed field. It also replaces the Mat\'ern-specific
scalar argument by explicit conditions on the latent spectral density and on
the Gaussian-Sobolev regularity of the transformation. For radially
nonincreasing spectral densities with polynomial tails, the required
regularity order does not increase with either the spatial dimension or the
number of latent components. In the rough-latent regime, a finite expected
squared weak gradient is sufficient. A broader class of spectral densities is
also covered under a stronger transformation regularity condition.

The underlying mechanism combines two ingredients. First, Gaussian
conditioning gives an exact expression for the MSPE of the latent-information
oracle and identifies its leading behaviour. Second, the regularity conditions
imply that the observed and latent spectral densities have proportional
high-frequency behaviour. A comparison of the corresponding kriging
residuals then yields the same leading MSPE for the simpler
latent-correlation predictor, while the intermediate predictors are controlled
by the finite-sample MSPE ordering. For exponential latent correlation, the
argument can be strengthened to give finite-sample MSPE bounds. In that case,
it is sufficient that observations accumulate at the prediction target;
density throughout the entire domain is not required.

Asymptotic equivalence in MSPE does not by itself determine the distribution
of the prediction error. We therefore also study the limiting standardized
prediction error. In general, the limit is a Gaussian scale mixture whose
random variance depends on the local behaviour of the transformation. It is
standard normal precisely when the squared norm of the transformation
gradient is almost surely constant. Hence asymptotically efficient point
prediction and asymptotically Gaussian uncertainty quantification are distinct
properties. The log-Gaussian and skew-Gaussian constructions provide simple
contrasting examples of this distinction.

The regularity assumptions underlying the positive MSPE-equivalence results
are also substantive. We show that an invertible and globally Lipschitz
transformation can retain a non-vanishing advantage for nonlinear prediction
when the latent Gaussian field is sufficiently smooth. Thus first-order
regularity alone is not enough at arbitrary latent smoothness. We also
construct a continuous non-Gaussian process with Mat\'ern covariance for which
nonlinear prediction has a strictly smaller asymptotic MSPE order than any
linear predictor. This second example shows that covariance structure alone
cannot determine whether kriging is asymptotically equivalent to the optimal
nonlinear predictor. The Gaussian-functional representation and its
regularity conditions therefore play an essential role in the positive
results.

The examples in Section~\ref{sec:examples} illustrate the scope of the theory
for both scalar and genuinely multilatent constructions, including
noninvertible models for which the latent Gaussian variables cannot be
recovered from the observations. 
Bernoulli threshold fields provide a particularly informative boundary case.
They belong to the Gaussian-functional class and arise naturally in spatial
binary models \citep{HeagertyLele1998}, but their discontinuous threshold
transformation falls outside the Gaussian-Sobolev regularity conditions of
our positive results. We therefore study them separately in
Section~\ref{sec:probit}. In this case the MSPE of the predictor based on the
latent Gaussian observations remains available explicitly, but our theory
does not imply asymptotic MSPE equivalence with the optimal predictor based
only on the observed binary field. This example further illustrates the
difference between belonging to the Gaussian-functional class and satisfying
the conditions required for asymptotic kriging efficiency.

Section~\ref{sec:setup} introduces the model, the predictors, and their
finite-sample comparison. Section~\ref{sec:main} gives the asymptotic MSPE
results and the finite-sample exponential bound. Section~\ref{sec:examples}
presents representative non-Gaussian models. Section~\ref{sec:limits} studies
limiting prediction errors, Gaussian calibration, and counterexamples, and
Section~\ref{sec:numerical} gives deterministic numerical comparisons for
log-Gaussian and sinh--arcsinh fields. Section~\ref{sec:discussion}
concludes. Technical proofs and additional examples are collected in the
appendices.

\section{Model and finite-sample prediction comparisons}\label{sec:setup}

This section establishes the statistical structure used throughout the paper.
We first define the Gaussian-functional model and the covariance map that links
latent and observed dependence. We then introduce the optimal conditional mean,
true-covariance kriging, and latent-correlation kriging. Their MSPEs admit an
exact finite-sample ordering and decomposition, which separates the gain from
nonlinear prediction from the additional gain obtained by using the true
observed covariance. These identities provide the starting point for the
fixed-domain asymptotic analysis in Section~\ref{sec:main}.

Throughout, $\N=\{1,2,\ldots\}$, $\Nzero=\{0,1,\ldots\}$,
$\|\cdot\|$ denotes the Euclidean norm, $I_q$ is the $q\times q$ identity
matrix, and $\boldsymbol1_n$ is the vector of $n$ ones. We write $\ind_A$
for the indicator of an event or set $A$, and
$\dist(s,S)=\min_{t\in S}\|s-t\|$ for the distance to a nonempty finite
set $S$. For positive sequences or functions,
$u\sim v$ means that $u/v$ tends to one, and $u\asymp v$ means that this
ratio is eventually bounded above and below by positive constants.

Let $D\subset\R^d$, $d\in\N$, be compact and fix $s_\star\in D$. Let
$Z_1,\ldots,Z_q$, with $1\le q<\infty$, be independent copies of a real,
centered stationary Gaussian field on $\R^d$ with unit variance and
correlation function $\rho$; observations are restricted to $D$. Independence refers to different latent fields, not to
different sites. Their common spectral density $f$ uses the convention
\begin{equation}\label{eq:spectral-convention}
 \rho(h)=\int_{\R^d}e^{\mathrm i\lambda^{\mathsf T}h}
 f(\lambda)\dd\lambda,
 \qquad
 \int_{\R^d}f(\lambda)\dd\lambda=1.
\end{equation}
The observed field is
\begin{equation}\label{eq:model}
 Y(s)=F\{Z_1(s),\ldots,Z_q(s)\},
 \qquad s\in D,
\end{equation}
where $F:\R^q\to\R$ is Borel measurable. Write $\gq$ for standard Gaussian
measure on $\R^q$, and assume
\begin{equation}\label{eq:L2-F}
 F\in L^2(\gq),
 \qquad
 0<\Var\{F(\bG)\}<\infty,
 \qquad
 \bG\sim N_q(0,I_q).
\end{equation}

The true field law is fixed. Its mean
$m_Y=\E F(\bG)$ is known throughout, there is no observation error, and all
covariance parameters are treated as fixed. The true covariance of the
observed field is regarded as known for the benchmark true-covariance kriging
predictor introduced below, whereas the simpler predictor based on the latent
correlation does not require it to be evaluated. The function $F$ need not be
monotone, invertible, or everywhere classically differentiable. A choice of
representative on Gaussian-null sets does not change the law at deterministic
finite sets of sites. The number of latent fields $q$ is distinct from the
spatial dimension $d$.

Let $\bG,\bG'$ be independent $N_q(0,I_q)$ vectors. For $-1\le u\le1$, put
\[
 \bG_u=u\bG+\sqrt{1-u^2}\,\bG'
\]
and define
\begin{equation}\label{eq:KF-definition}
 K_F(u)=\Cov\{F(\bG),F(\bG_u)\}.
\end{equation}
The paired vectors have cross-covariance $uI_q$. The argument $u$ is a
scalar correlation, not a spatial lag, and $K_F$ is not normalized by the
marginal variance:
\[
 K_F(0)=0,
 \qquad
 K_F(1)=\Var\{F(\bG)\}.
\]
The covariance of the observed field is
\begin{equation}\label{eq:observed-covariance}
 C_Y(h)=K_F\{\rho(h)\}.
\end{equation}
Thus $K_F$ separates the effect of the pointwise transformation from the
spatial correlation. It depends on the Gaussian representation of $F$, not
only on the univariate marginal distribution of $Y$.

For distinct sites
$S_n=\{s_1,\ldots,s_n\}\subset D\setminus\{s_\star\}$, put
\[
 \boldsymbol Y_n
 =
 \bigl(Y(s_1),\ldots,Y(s_n)\bigr)^{\mathsf T},
 \qquad
 \mathcal G_n=\sigma(\boldsymbol Y_n),
\]
and write $Z_j(S_n)=(Z_j(s_1),\ldots,Z_j(s_n))^{\mathsf T}$ for the
corresponding latent observation vector. More generally, $W(S_n)$ denotes
the vector of values of a scalar field $W$ at these ordered sites. Define
\begin{align*}
 C_{Y,n}
 &=
 \{C_Y(s_i-s_k)\}_{i,k=1}^n,
 &
 \boldsymbol c_{Y,n}
 &=
 \{C_Y(s_\star-s_i)\}_{i=1}^n,
 \\
 R_n
 &=
 \{\rho(s_i-s_k)\}_{i,k=1}^n,
 &
 \boldsymbol r_n
 &=
 \{\rho(s_\star-s_i)\}_{i=1}^n.
\end{align*}
The finite-sample comparisons below assume that the covariance matrices are
positive definite; the spectral assumptions introduced in
Section~\ref{sec:main} ensure this property. Write
\[
 \boldsymbol w_{Y,n}
 =
 C_{Y,n}^{-1}\boldsymbol c_{Y,n},
 \qquad
 \boldsymbol w_{\rho,n}
 =
 R_n^{-1}\boldsymbol r_n.
\]

Kriging with the observed covariance and with the latent correlation gives,
respectively,
\begin{align}
 \YL
 &=
 m_Y+
 \boldsymbol w_{Y,n}^{\mathsf T}
 (\boldsymbol Y_n-m_Y\boldsymbol1_n),
 \label{eq:linear-predictor}
 \\
 \YR
 &=
 m_Y+
 \boldsymbol w_{\rho,n}^{\mathsf T}
 (\boldsymbol Y_n-m_Y\boldsymbol1_n).
 \label{eq:latent-covariance-predictor}
\end{align}
Both predictors act on the same observed non-Gaussian data and predict
$Y(s_\star)$ on its original scale. In
\eqref{eq:latent-covariance-predictor}, the latent Gaussian correlation is
used only to determine the kriging weights. Thus $\YR$ is neither a predictor
constructed from observed $Z_j$ values nor a transformation of a latent
Gaussian kriging predictor.
The optimal predictor based on the observed data is
\begin{equation}\label{eq:nonlinear-predictor}
 \YN
 =
 \E\{Y(s_\star)\mid\mathcal G_n\}.
\end{equation}
For $A\in\{\mathrm L,\rho,\mathrm{NL}\}$, define
\[
 \operatorname{MSPE}_n^A
 =
 \E\!\left[
 \{Y(s_\star)-\widehat Y_n^A\}^2
 \right].
\]
These are unconditional deterministic MSPEs under the true non-Gaussian law,
not variances reported by a working Gaussian model. The term linear prediction
includes the known-mean correction, or equivalently linear prediction of the
centered field.

\begin{proposition}\label{prop:finite-gains}
For every finite design with the covariance matrices above positive definite,
define
\begin{equation}\label{eq:gain-definitions}
 G_{\mathrm N,n}
 :=
 \E[(\YL-\YN)^2],
 \qquad
 G_{\mathrm C,n}
 :=
 \E[(\YR-\YL)^2].
\end{equation}
Then
\begin{equation}\label{eq:finite-hierarchy}
 \MSPEN
 \le
 \MSPEL
 \le
 \MSPER,
\end{equation}
with the exact identities
\begin{align}
 \MSPEL-\MSPEN
 &=
 G_{\mathrm N,n},
 \label{eq:projection-risk}
 \\
 \MSPER-\MSPEL
 &=
 G_{\mathrm C,n}
 \notag\\
 &=
 (\boldsymbol w_{\rho,n}-\boldsymbol w_{Y,n})^{\mathsf T}
 C_{Y,n}
 (\boldsymbol w_{\rho,n}-\boldsymbol w_{Y,n}).
 \label{eq:covariance-gain}
\end{align}
Consequently,
\begin{equation}\label{eq:two-gain-decomposition}
 \MSPER
 =
 \MSPEN 
 +
 G_{\mathrm N,n}
 +
 G_{\mathrm C,n}.
\end{equation}
The second inequality in~\eqref{eq:finite-hierarchy} is strict exactly when
$\boldsymbol w_{\rho,n}\ne\boldsymbol w_{Y,n}$.
\end{proposition}

The proof is given in Appendix~\ref{app:finite-gains}.
Proposition~\ref{prop:finite-gains} separates two distinct finite-sample
gains. $G_{\mathrm N,n}$ measures the improvement obtained by allowing
nonlinear prediction beyond the linear class, whereas $G_{\mathrm C,n}$
measures the additional improvement obtained, within the linear class, by
using the true covariance of the observed field instead of the latent
Gaussian correlation. No dense-sampling, spectral-tail, or
transformation-regularity assumption is needed for these identities.
Section~\ref{sec:main} shows that, under the conditions of the main theorem,
both gains are asymptotically negligible relative to the optimal
observed-data MSPE.

 \section{Asymptotic optimality with observed or latent covariance}
\label{sec:main}\label{sec:exponential}

This section turns the finite-sample MSPE ordering of
Proposition~\ref{prop:finite-gains} into an asymptotic equivalence result.
The argument has three main ingredients. First, we introduce a theoretical
predictor that observes the latent Gaussian variables and show that its MSPE
is available exactly. Second, we identify the regularity conditions under
which the observed non-Gaussian field inherits the high-frequency behaviour
of the latent field. Finally, a spectral comparison for a single kriging
residual connects these two facts and yields the main result.

We begin with assumptions on the latent spectral density and on the design
sequence.

\begin{assumption}\label{ass:spectra}
For a fixed $\nu>0$, the common latent spectral density $f$ is continuous,
even and strictly positive on $\R^d$, and
\begin{equation}\label{eq:latent-tail}
 f(\lambda)
 =
 c\norm{\lambda}^{-d-2\nu}\{1+o(1)\},
 \qquad
 \norm{\lambda}\longrightarrow\infty,
 \qquad
 c>0.
\end{equation}
\end{assumption}

The limit is taken over all frequency directions; isotropy at finite
frequencies is not required. Writing $p=d+2\nu>d$, continuity, positivity,
and~\eqref{eq:latent-tail} imply the global two-sided bound
\begin{equation}\label{eq:latent-envelope}
 b(1+\norm{\lambda})^{-p}
 \le
 f(\lambda)
 \le
 B(1+\norm{\lambda})^{-p},
 \qquad \lambda\in\R^d,
\end{equation}
for some constants $0<b\le B<\infty$. Corollary~\ref{cor:latent-Matern}
verifies these conditions explicitly for the Mat\'ern family.

\begin{assumption}\label{ass:design}
The observation sites form a sequence of distinct deterministic points
$\{s_i:i\ge1\}\subset D\setminus\{s_\star\}$ such that
$
 \overline{\{s_i:i\ge1\}}=D.
$
At stage $n$, the field is observed at
$S_n=\{s_1,\ldots,s_n\}$.
\end{assumption}

Only $Y$ is observed. The same sites are used below to define a theoretical
latent-information benchmark. The nested formulation is convenient for the
main statement; Remark~\ref{rem:nonnested-designs} gives the corresponding
extension to nonnested designs.

Suppose, only for comparison, that the latent Gaussian values at the
observation sites were also available. Define
\[
 \mathcal H_n
 =
 \sigma\{Z_j(s_i):1\le j\le q,\ 1\le i\le n\}.
\]
Since the latent fields are independent copies with the same correlation,
their conditional variances at the prediction site coincide. Using
$R_n$ and $\boldsymbol r_n$ from Section~\ref{sec:setup}, write
\begin{equation}\label{eq:latent-variance}
 v_n
 =
 1-\boldsymbol r_n^{\mathsf T}R_n^{-1}\boldsymbol r_n.
\end{equation}
Thus, for every $j$,
$
 v_n
 =
 \Var\{Z_j(s_\star)\mid Z_j(S_n)\}.
$
Under Assumptions~\ref{ass:spectra} and~\ref{ass:design},
$0<v_n\le1$ and $v_n\to0$.

The latent-information predictor is
\begin{equation}\label{eq:oracle}
 \YO
 =
 \E\{Y(s_\star)\mid\mathcal H_n\},
 \qquad
 \MSPEO
 =
 \E\!\left[\{Y(s_\star)-\YO\}^2\right].
\end{equation}
It is a theoretical comparator rather than a proposed prediction procedure.
Because
$\mathcal G_n\subseteq\mathcal H_n$, the finite-sample hierarchy becomes
\begin{equation}\label{eq:sandwich}
 \MSPEO
 \le
 \MSPEN
 \le
 \MSPEL
 \le
 \MSPER.
\end{equation}

For subsequent use, define the covariance decrement
\begin{equation}\label{eq:Delta-definition}
 \Delta_F(v)=K_F(1)-K_F(1-v),\qquad 0\le v\le2.
\end{equation}
The oracle MSPE uses $0\le v\le1$; the larger domain also accommodates
covariance increments when the latent correlation is negative. The main
advantage of the enlarged information set is that its MSPE has a simple
exact form.

\begin{proposition}\label{prop:exact}
For independent standardized Gaussian replicates with common correlation
$\rho$ and $F$ satisfying~\eqref{eq:L2-F},
\begin{equation}\label{eq:oracle-risk}
 \MSPEO
 =
 K_F(1)-K_F(1-v_n).
\end{equation}
This identity holds for every finite design, without invertibility or
differentiability of $F$. If $K_F'(1^-)<\infty$ and $v_n\to0$, then
\begin{equation}\label{eq:oracle-asymptotic-mspe}
 \frac{\MSPEO}{v_n}
 \longrightarrow
 K_F'(1^-).
\end{equation}
If $q=1$ and $F$ has a measurable inverse on its image, then
$\mathcal G_n=\mathcal H_n$ and therefore
$\MSPEN=\MSPEO$ for every $n$.
\end{proposition}

The identity follows from Gaussian conditioning. Conditionally on
$\mathcal H_n$, two independent copies of the latent target have common
conditional mean and independent residuals with covariance $v_nI_q$.
Unconditionally, their cross-covariance is therefore $(1-v_n)I_q$, so the
covariance map gives~\eqref{eq:oracle-risk}. The full argument is given in
Appendix~\ref{app:oracle-proof}. For scalar transformations, the corresponding
Hermite-series expression appears in
\citet[Section~3.1]{HiranoYajima2013}.

Proposition~\ref{prop:exact} already identifies the quantity that will govern
the leading MSPE. It remains to determine when the observed field has the
corresponding high-frequency covariance structure.

Let $H_k$ denote the probabilists' Hermite polynomial,
\[
 H_k(z)=(-1)^k e^{z^2/2}\frac{\mathrm d^k}{\mathrm dz^k}e^{-z^2/2},
 \qquad h_k(z)=\frac{H_k(z)}{\sqrt{k!}}.
\]
For a multi-index $\bk=(k_1,\ldots,k_q)\in\Nzero^q$, put
$|\bk|=\sum_{j=1}^qk_j$ and
$\Psi_{\bk}(\bz)=\prod_{j=1}^q h_{k_j}(z_j)$. These functions form an
orthonormal basis of $L^2(\gq)$, so
\begin{equation}\label{eq:Hermite}
 F=\sum_{\bk\in\Nzero^q}b_{\bk}\Psi_{\bk}
 \quad\text{in }L^2(\gq),\qquad
 b_{\bk}=\E\{F(\bG)\Psi_{\bk}(\bG)\}.
\end{equation}
In particular, $b_{\boldsymbol0}=m_Y$. Mehler's formula gives
\begin{equation}\label{eq:KF-series}
 K_F(u)
 =
 \sum_{\ell\ge1}a_\ell u^\ell,
 \qquad
 a_\ell
 =
 \sum_{|\bk|=\ell}b_{\bk}^2
 \ge0,
 \qquad
 \sum_{\ell\ge1}a_\ell
 =
 K_F(1),
\end{equation}
The derivation is recalled in Appendix~\ref{app:Hermite}. Since the
coefficients are nonnegative, the left derivative at $u=1$ exists in the
extended sense. 
The left derivative of \(K_F\) at \(1\) is
\[
K_F'(1^-)
=
\lim_{u\uparrow1}
\frac{K_F(u)-K_F(1)}{u-1}
=
\lim_{v\downarrow0}
\frac{K_F(1)-K_F(1-v)}{v}.
\]
We define
\begin{equation}\label{eq:energies}
 \mathcal D(F)
 :=
 K_F'(1^-)
 =
 \sum_{\ell\ge1}\ell a_\ell.
\end{equation}

When this quantity is finite, $F$ has a Gaussian-square-integrable weak
gradient $\nabla F=(\partial_1F,\ldots,\partial_qF)^{\mathsf T}$, where
weak derivatives are distributional derivatives on $\R^q$. In that case,
\begin{equation}\label{eq:gradient-energies}
 \mathcal D(F)
 =
 \E\norm{\nabla F(\bG)}^2
 =
 \sum_{j=1}^q\mathcal D_j(F),
 \qquad
 \mathcal D_j(F)
 =
 \E[\{\partial_jF(\bG)\}^2].
\end{equation}
Thus $\mathcal D(F)$ measures the squared gradient of the transformation
averaged under standard Gaussian measure. In parts of the Gaussian-analysis
literature, this quantity is referred to as the Gaussian Dirichlet energy.
It can be calculated either from the weak gradient or from the endpoint slope
of the covariance map; individual Hermite coefficients need not be evaluated.

For $s\ge0$, we use the Gaussian Sobolev scale
\begin{equation}\label{eq:gaussian-sobolev}
 F\in W^{s,2}(\gq)
 \quad\Longleftrightarrow\quad
 |\E F(\bG)|^2
 +
 \sum_{\ell\ge1}(1+\ell)^s a_\ell
 <\infty.
\end{equation}
Here $L^2(\gq)$ specifies the integrability weight. For integer $s$, the scale
is characterized by Gaussian-square-integrable weak derivatives; the
criteria needed below are given in
Appendix~\ref{app:Hermite-Sobolev}. At noninteger $s$ it is the Gaussian
Bessel-potential scale associated with the Ornstein--Uhlenbeck operator; see
\citet{Janson1997,LunardiPallara2020}.

The covariance-map representation also gives a simple preservation principle
for regularity. By mean-square differentiability of order $m$ we mean the
existence, as $L^2$ limits, of all spatial partial derivatives of total
order at most $m$.

\begin{proposition}
\label{prop:ms-regularity}
If $F\in L^2(\gq)$ and the latent Gaussian field is mean-square continuous,
then $Y$ is mean-square continuous. More generally, for any integer $m\ge1$,
if the latent field is mean-square differentiable of order $m$ and
$F\in W^{m,2}(\gq)$, then $Y$ is mean-square differentiable of order $m$.
\end{proposition}

For measurable second-order fields, mean-square differentiability also has
sample-path consequences: under the corresponding stationary conditions,
the paths belong almost surely to the appropriate local Sobolev space.
Stronger fractional conditions can yield classical pathwise smoothness through
Sobolev embeddings; see \citet{Scheuerer2010}. The proof and these
sample-path consequences are given in Appendix~\ref{app:regularity-proof}.

The regularity condition required for the sharper spectral result is the
following.

\begin{assumption}\label{ass:Hermite}
Let $\nu>0$ be the latent smoothness parameter in
Assumption~\ref{ass:spectra}. With the coefficients in~\eqref{eq:KF-series},
\begin{equation}\label{eq:Hermite-condition}
 \sum_{\ell\ge1}a_\ell w_\nu(\ell)<\infty,
 \qquad
 w_\nu(\ell)=
 \begin{cases}
  \ell, & 0<\nu<1,\\
  \ell(1+\log\ell), & \nu=1,\\
  \ell^\nu, & \nu>1.
 \end{cases}
\end{equation}
\end{assumption}

Here $\log$ denotes the natural logarithm. The condition has a direct
Gaussian-Sobolev interpretation. Under
\eqref{eq:L2-F}, it is equivalent to
$F\in W^{1,2}(\gq)$ for $0<\nu<1$ and to
$F\in W^{\nu,2}(\gq)$ for $\nu>1$. At $\nu=1$ the logarithmic factor gives
a slightly stronger condition than $W^{1,2}$; any
$W^{1+\epsilon,2}(\gq)$ condition with $\epsilon>0$ is sufficient.
More generally, square integrability of the weak derivatives through an
integer order $M>\nu$ provides a convenient sufficient check.

Two features are worth emphasizing. First, the required regularity order
depends on the latent smoothness $\nu$, but not on the spatial dimension $d$
or on the number $q$ of latent fields. Second, the condition is sufficient,
not claimed to be necessary. In particular, the additional logarithmic factor
at $\nu=1$ is used to obtain a uniform convolution bound.

Assumption~\ref{ass:Hermite} and nonconstancy of $F$ imply
$
 0<\mathcal D(F)=K_F'(1^-)<\infty.
$
We write
\begin{equation}\label{eq:observed-coefficient-definition}
 c_Y
 =
 cK_F'(1^-)
 =
 c\mathcal D(F).
\end{equation}
Proposition~\ref{prop:exact} shows that finiteness of the endpoint slope alone
determines the leading oracle MSPE. The additional regularity in
Assumption~\ref{ass:Hermite}, when $\nu\ge1$, is needed to transfer the
latent spectral tail to the observed field.

For the sharper transfer result we use one further property of the latent
spectral density. Its role is to obtain a convolution bound that is uniform
in Hermite degree:
\begin{equation}\label{eq:radial-monotone}
 f(\lambda)
 =
 \widetilde f(\|\lambda\|),
 \qquad
 \widetilde f
 \text{ is nonincreasing on }[0,\infty).
\end{equation}
This is a condition on the spectral density, not on monotonicity of the
covariance function or directly on sample-path behaviour.
Theorem~\ref{thm:general} below removes it at the cost of stronger
transformation regularity.

The exact oracle MSPE identifies $\mathcal D(F)$ as the local cost of the
transformation. The next proposition shows that exactly the same quantity
appears in the high-frequency behaviour of the observed covariance. This is
the key link between nonlinear prediction and kriging.

\begin{proposition}\label{prop:tail}
Under Assumptions~\ref{ass:spectra} and~\ref{ass:Hermite}, together with
\eqref{eq:radial-monotone}, the observed spectral density $f_Y$ is continuous
and strictly positive and satisfies
\begin{equation}\label{eq:observed-tail}
 f_Y(\lambda)
 =
 cK_F'(1^-)
 \norm{\lambda}^{-d-2\nu}
 \{1+o(1)\},
 \qquad
 \norm{\lambda}\to\infty.
\end{equation}
Moreover, $f_Y$ is bounded above and below by positive multiples of
$(1+\norm{\lambda})^{-d-2\nu}$. In particular,
\begin{equation}\label{eq:observed-latent-tail-ratio}
 \frac{f_Y(\lambda)}{f(\lambda)}
 \longrightarrow
 K_F'(1^-)
 =
 \mathcal D(F).
\end{equation}
\end{proposition}

The proof uses
$
 f_Y(\lambda)
 =
 \sum_{\ell\ge1}a_\ell f^{*\ell}(\lambda),
$
where $f^{*\ell}$ is the $\ell$-fold convolution of $f$ on $\R^d$,
with $f^{*1}=f$. The representation holds because $\rho^\ell$ has
spectral density $f^{*\ell}$.
For each fixed $\ell$, the polynomial tail coefficient of $f^{*\ell}$ is
$\ell c$. The convolution bound in
Appendix~\ref{app:convolutions} controls these terms uniformly in $\ell$,
and Assumption~\ref{ass:Hermite} permits the high-frequency limit to pass
through the Hermite series.

Propositions~\ref{prop:exact} and~\ref{prop:tail} identify the same
coefficient $\mathcal D(F)$ in two apparently different objects: the MSPE
of the latent-information predictor and the high-frequency ratio of the
observed and latent spectral densities. The main theorem combines them.

\begin{theorem}\label{thm:main}
Suppose~\eqref{eq:L2-F}, Assumptions~\ref{ass:spectra},
\ref{ass:design}, and~\ref{ass:Hermite}, together with
\eqref{eq:radial-monotone}, hold. Then
\begin{equation}\label{eq:main-risk-scale}
 \MSPEO
 \sim
 \MSPEN
 \sim
 \MSPEL
 \sim
 \MSPER
 \sim
 \mathcal D(F)v_n.
\end{equation}
Consequently,
\begin{equation}\label{eq:main-ratio}
 \frac{\MSPEL}{\MSPEN}
 \longrightarrow1,
 \qquad
 \frac{\MSPER}{\MSPEN}
 \longrightarrow1,
\end{equation}
and
\begin{equation}\label{eq:oracle-information-gap}
 \frac{\E[(\YN-\YO)^2]}{\MSPEO}
 \longrightarrow0.
\end{equation}
All four MSPEs are strictly positive for every finite design. The conclusions
hold for every fixed $d\ge1$, finite $q\ge1$, and $\nu>0$, subject to the
corresponding transformation condition in~\eqref{eq:Hermite-condition}.
\end{theorem}

The proof is a sandwich argument. Proposition~\ref{prop:exact} gives
$
 \MSPEO
 \sim
 \mathcal D(F)v_n.
$
On the other hand, Proposition~\ref{prop:tail} and the residual comparison in
Lemma~\ref{lem:residual-energy} imply
$
 \MSPER
 \sim
 \mathcal D(F)v_n.$
Since
\[
 \MSPEO
 \le
 \MSPEN
 \le
 \MSPEL
 \le
 \MSPER,
\]
the two intermediate MSPEs must have the same first-order behaviour. No
separate asymptotic comparison of the two sets of kriging weights is needed.
The formal proof is given in Appendix~\ref{app:asymptotic-proofs}.

The theorem does not eliminate the finite-sample advantages identified in
Proposition~\ref{prop:finite-gains}; it shows that they are lower order
relative to the optimal observed-data MSPE. Indeed,
\begin{equation}\label{eq:vanishing-gains}
 \frac{G_{\mathrm N,n}}{\MSPEN}
 \longrightarrow0,
 \qquad
 \frac{G_{\mathrm C,n}}{\MSPEN}
 \longrightarrow0.
\end{equation}
Thus neither the restriction to linear prediction nor the replacement of the
true observed covariance by the latent correlation produces a first-order
MSPE loss under the stated conditions. No monotonicity of these gains or
universal rate in $n$ is asserted.

\begin{remark}\label{rem:observed-scale}
Both $\YL$ and $\YR$ act directly on the observed non-Gaussian data.
Neither is obtained by predicting the latent Gaussian field and then applying
$F$. Constructing $\YL$ requires the observed covariance map $K_F$, whereas
constructing $\YR$ requires only $\rho$ and the known mean $m_Y$.
The transformation still enters the assumptions that justify asymptotic
efficiency, and $\mathcal D(F)$ determines the leading MSPE even though it is
not needed to calculate the latent-correlation kriging weights.
Unknown or misspecified means, observation error, and estimated covariance
parameters are not covered here. Finally, MSPE efficiency alone does not imply
Gaussian uncertainty calibration; Section~\ref{sec:error-calibration}
studies the limiting prediction error.
\end{remark}

The Mat\'ern family gives the canonical specialization of the general result.

\begin{corollary}\label{cor:latent-Matern}
Suppose~\eqref{eq:L2-F} and Assumptions~\ref{ass:design} and
\ref{ass:Hermite} hold. If the common latent correlation is Mat\'ern with
range $\alpha>0$ and smoothness $\nu>0$, then all conclusions of
Theorem~\ref{thm:main} hold.
\end{corollary}

Our range parametrization is fixed by the Mat\'ern spectral density
\begin{equation}\label{eq:matern-spectral-density}
 f(\lambda)
 =C_{\nu,d}\alpha^{-2\nu}
  (\alpha^{-2}+\|\lambda\|^2)^{-\nu-d/2},
 \qquad
 C_{\nu,d}=\frac{\Gamma(\nu+d/2)}{\pi^{d/2}\Gamma(\nu)},
\end{equation}
where $\Gamma$ is the gamma function. The verification of the spectral
conditions is given in Appendix~\ref{app:matern-corollary}.
The theorem is not restricted to Matérn correlations. For the
generalized Wendland family ${\cal GW}_{\kappa,\mu,a}$ \citep{BevilacquaEtAl2019}, the spectral
conditions of Theorem~\ref{thm:main} hold whenever
\[
\kappa>-\frac12,\qquad
\mu\ge \kappa+\frac{d+3}{2},
\]
with latent smoothness $\nu=\kappa+1/2$.

When the latent field is rough, the transformation condition simplifies
considerably.

\begin{corollary}\label{cor:rough}
Suppose~\eqref{eq:L2-F}, Assumptions~\ref{ass:spectra} and
\ref{ass:design}, and~\eqref{eq:radial-monotone} hold, with $0<\nu<1$.
If
\begin{equation}\label{eq:W12}
 \E\|\nabla F(\bG)\|^2<\infty,
 \qquad
 \bG\sim N_q(0,I_q),
\end{equation}
where the gradient is weak, then all conclusions of
Theorem~\ref{thm:main} hold.
\end{corollary}

Equivalently, in this regime it is enough that
$K_F'(1^-)<\infty$. In particular, for a Mat\'ern latent field with
$0<\nu<1$, a finite expected squared weak gradient is sufficient for
asymptotic MSPE equivalence. The short proof is given in
Appendix~\ref{app:rough-corollary}.

The exponential model, corresponding to Mat\'ern smoothness $\nu=1/2$,
admits a stronger conclusion. Let
\begin{equation}\label{eq:exponential-correlations}
 \rho(h)
 =
 \exp(-\|h\|/\alpha),
 \qquad
 \alpha>0.
\end{equation}
In this case the observed spectral density can be dominated globally by the
latent spectral density, yielding a finite-sample version of the asymptotic
sandwich.

\begin{proposition}\label{prop:exponential-sandwich}
Suppose~\eqref{eq:L2-F}, let $F\in W^{1,2}(\gq)$, and let the common latent
correlation be~\eqref{eq:exponential-correlations}. With $\mathcal D(F)$
and $\Delta_F$ defined in~\eqref{eq:energies} and
\eqref{eq:Delta-definition}, for every finite design of distinct sites
excluding $s_\star$,
\begin{equation}\label{eq:exponential-finite-sandwich}
 0
 <
 \Delta_F(v_n)
 =
 \MSPEO
 \le
 \MSPEN
 \le
 \MSPEL
 \le
 \MSPER
 \le
 \mathcal D(F)v_n.
\end{equation}
Consequently,
\begin{equation}\label{eq:exponential-finite-ratios}
 1
 \le
 \frac{\MSPEL}{\MSPEN}
 \le
 \frac{\MSPER}{\MSPEN}
 \le
 \frac{\mathcal D(F)v_n}{\Delta_F(v_n)}.
\end{equation}
If $v_n\to0$ along a sequence of finite designs, then all conclusions of
Theorem~\ref{thm:main} hold along that sequence. In particular,
$\operatorname{dist}(s_\star,S_n)\to0$ is sufficient; neither nesting nor
density throughout $D$ is required.

If, additionally, $K_F''(1^-)<\infty$, then
\begin{equation}\label{eq:exponential-finite-gap}
 0
 \le
 \MSPER-\MSPEO
 \le
 \mathcal D(F)v_n-\Delta_F(v_n)
 \le
 \tfrac12 K_F''(1^-)v_n^2.
\end{equation}
\end{proposition}

The upper bound follows from a simple scaling property. Since
$\rho(h)^\ell=\exp(-\ell\|h\|/\alpha)$, its spectral density corresponds to
range $\alpha/\ell$ and is bounded above by $\ell$ times the original
exponential spectral density. The nonnegative expansion
\eqref{eq:KF-series} therefore yields
$
 f_Y(\lambda)
 \le
 \mathcal D(F) f_\alpha(\lambda)
$
at every frequency, where $f_\alpha$ is the spectral density of
\eqref{eq:exponential-correlations}. Applying this covariance domination
to the prediction
residual generated by the latent-correlation kriging weights gives
$\MSPER\le\mathcal D(F)v_n$. The exact oracle MSPE supplies the lower bound.
Appendix~\ref{app:exponential} gives the details and verifies
\begin{equation}\label{eq:exponential-cY}
 c
 =
 K_d/\alpha,
 \qquad
 c_Y
 =
 (K_d/\alpha)K_F'(1^-),
 \qquad
 K_d
 =
 \frac{\Gamma\{(d+1)/2\}}
 {\pi^{(d+1)/2}}.
\end{equation}
The same exponential spectral calculation underlies the scalar condition in
\citet[Remark~1]{HiranoYajima2013}; here it yields finite-sample control of
the entire multilatent MSPE hierarchy.

Radial monotonicity is useful because it yields the sharper,
dimension-independent transformation condition in
Assumption~\ref{ass:Hermite}, but it is not essential to the MSPE-equivalence
mechanism. It can be removed at the cost of stronger transformation
regularity.

\begin{theorem}\label{thm:general}
Suppose~\eqref{eq:L2-F} and Assumptions~\ref{ass:spectra} and
\ref{ass:design} hold. Without assuming~\eqref{eq:radial-monotone}, replace
Assumption~\ref{ass:Hermite} by
\begin{equation}\label{eq:conservative-Hermite}
 \sum_{\ell\ge1}
 a_\ell(1+\ell)^{d+2\nu+1}
 <\infty.
\end{equation}
Then all conclusions of Theorem~\ref{thm:main} hold.
\end{theorem}

The stronger exponent follows from the elementary convolution bound in
Lemma~\ref{lem:conv-bound}, which replaces the refined radial argument.
A convenient coefficient-free sufficient condition is Gaussian-square
integrability of the weak derivatives of $F$ through an integer order
$M\ge d+2\nu+1$. The proof is given in
Appendix~\ref{app:general-proof}.

\section{Examples}\label{sec:examples}\label{sec:clayton}

We illustrate the conditions of Section~\ref{sec:main} with four
representative non-Gaussian constructions: a smooth invertible scalar
transformation, two noninvertible multilatent fields, and a copula-based
field. The distinction matters because, for multilatent models, the
observations generally do not determine the latent Gaussian values.

Throughout, $Z_1,Z_2,\ldots$ are independent standardized Gaussian fields
with common correlation $\rho$; in scalar examples we write $Z$.
The relevant quantity is the Gaussian Dirichlet energy
$\mathcal D(F)=\E\|\nabla F(\bG)\|^2$. The checks of weak and fractional
regularity, as well as further examples, are in Appendix~\ref{app:examples}.
We write $\Phi$ and $\phi$ for the standard normal distribution function
and density, $z_p=\Phi^{-1}(p)$ for $0<p<1$, and $Q$ for a marginal quantile
function on $(0,1)$.

\subsection{Scalar transformations}

For $q=1$, consider the log-Gaussian field
\begin{equation}\label{eq:lognormal-example}
Y(s)=\exp\{\mu+\omega Z(s)\},\qquad
\mu\in\R,\quad\omega>0,\qquad
\mathcal D(F)=\omega^2e^{2\mu+2\omega^2}.
\end{equation}
The transformation is smooth and invertible and satisfies the conditions
of Theorem~\ref{thm:main} for every finite latent smoothness $\nu>0$.
Invertibility also makes the nonlinear MSPE available exactly; see
\eqref{eq:lognormal-exact-risk}. The covariance map is given in
Appendix~\ref{app:lognormal-regularity}.

The sinh--arcsinh and Tukey $g$-and-$h$ fields provide other smooth scalar
examples under their moment restrictions
\citep{JonesPewsey2009,BlasiEtAl2022,XuGenton2017}.
Likewise, Gaussian-copula fields $Y(s)=Q\{\Phi(Z(s))\}$ are covered when
$Q\circ\Phi$ has the required Gaussian-Sobolev regularity; see
Appendix~\ref{app:additional-scalars}.

\subsection{Noninvertible multilatent transformations}

The skew-Gaussian field of \citet{ZhangElShaarawi2010} is
\begin{equation}\label{eq:skew-example}
Y(s)=\mu+\delta|Z_1(s)|+\tau Z_2(s),\qquad
\mu,\delta\in\R,\quad\tau>0.
\end{equation}
Its weak gradient gives $\mathcal D(F)=\delta^2+\tau^2$.
The absolute-value component has Gaussian-Sobolev regularity of every
order $s<3/2$, but not of order $3/2$. Consequently,
Theorem~\ref{thm:main} applies for $0<\nu<3/2$ under
\eqref{eq:radial-monotone}. The covariance map and regularity calculation
are given in Appendix~\ref{app:fractional-examples}.

Another noninvertible example is the Weibull field of
\citet{BevilacquaEtAl2020},
\begin{equation}\label{eq:Weibull-example}
Y(s)=\tau\left\{\frac{Z_1(s)^2+Z_2(s)^2}{2}\right\}^{1/k},
\qquad \tau,k>0.
\end{equation}
Its margins are Weibull with shape $k$ and scale $\tau$, and
\begin{equation}\label{eq:Weibull-energy}
\mathcal D(F)=\frac{2\tau^2}{k^2}\Gamma(2/k)<\infty,
\qquad k>0.
\end{equation}
Thus Corollary~\ref{cor:rough} applies for every $k>0$ throughout
$0<\nu<1$, including cases where the transformation is not classically
differentiable at the origin. The weak-gradient argument is given in
Appendix~\ref{app:weibull-regularity}.

\subsection{Copula-based random fields}

A genuinely multilatent example is the Clayton-like spatial copula of
\citet{BevilacquaEtAl2024}. For $q_0\in\N$, let
$A_{q_0}(s)=\sum_{j=1}^{q_0}Z_j(s)^2$ and
$B_2(s)=Z_{q_0+1}(s)^2+Z_{q_0+2}(s)^2$.
Define
\begin{equation}\label{eq:clayton-uniform-field}
U_{q_0}(s)=\left\{\frac{A_{q_0}(s)}{A_{q_0}(s)+B_2(s)}\right\}^{q_0/2},
\qquad Y(s)=U_{q_0}(s).
\end{equation}
The corresponding pointwise map is
\begin{equation}\label{eq:clayton-map}
F_{q_0}(\bz)=
\left(\frac{\sum_{j=1}^{q_0}z_j^2}{\sum_{j=1}^{q_0+2}z_j^2}\right)^{q_0/2},
\qquad \bz\ne\boldsymbol0,
\end{equation}
with any value assigned at $\boldsymbol0$. Since the ratio in
\eqref{eq:clayton-uniform-field} has a
$\operatorname{Beta}(q_0/2,1)$ distribution before the power is applied,
$U_{q_0}(s)$ is uniform on $(0,1)$. Despite the noninvertibility,
\begin{equation}\label{eq:clayton-energy}
\mathcal D(F_{q_0})=K_{F_{q_0}}'(1^-)
=\frac{2q_0}{3(3q_0-2)},\qquad q_0\ge1.
\end{equation}

Prescribed marginal distributions can be obtained by setting
\begin{equation}\label{eq:clayton-quantile-field}
Y(s)=Q\{U_{q_0}(s)\}.
\end{equation}
For a strictly increasing, locally absolutely continuous quantile $Q$,
the corresponding gradient energy is determined by
\begin{equation}\label{eq:clayton-quantile-energy}
I_{q_0}(Q)=q_0\int_0^1[Q'(t)]^2
t^{\,2-2/q_0}(1-t^{2/q_0})\dd t.
\end{equation}

\begin{corollary}\label{cor:clayton}
Suppose Assumptions~\ref{ass:spectra} and~\ref{ass:design} hold with
$0<\nu<1$, and assume~\eqref{eq:radial-monotone}.
For the uniform Clayton-like field,
\begin{equation}\label{eq:clayton-uniform-risk}
\MSPEO\sim\MSPEN\sim\MSPEL\sim\MSPER
\sim\frac{2q_0}{3(3q_0-2)}v_n,\qquad q_0\ge1.
\end{equation}
For the quantile field~\eqref{eq:clayton-quantile-field}, suppose that
$Q$ is strictly increasing and locally absolutely continuous and that
\begin{equation}\label{eq:clayton-quantile-conditions}
\int_0^1Q(t)^2\dd t<\infty,\qquad I_{q_0}(Q)<\infty.
\end{equation}
Then $F_{q_0,Q}=Q\circ F_{q_0}$ belongs to
$W^{1,2}(\gamma_{q_0+2})$ and
$\mathcal D(F_{q_0,Q})=I_{q_0}(Q)$, so
\begin{equation}\label{eq:clayton-quantile-risk}
\MSPEO\sim\MSPEN\sim\MSPEL\sim\MSPER
\sim I_{q_0}(Q)v_n.
\end{equation}
\end{corollary}

The calculations are given in Appendix~\ref{app:clayton}. In particular,
for $\operatorname{Beta}(a,b)$ margins the quantile conditions hold for
all $a,b>0$ if $q_0\ge2$; if $q_0=1$, they hold exactly when $a<2$.

Table~\ref{tab:models} summarizes these and three additional constructions
whose details have been moved to Appendix~\ref{app:examples}.
\begin{table}[htbp]
\centering\small
\setlength{\tabcolsep}{4pt}
\renewcommand{\arraystretch}{1.22}
\caption{Representative Gaussian-functional fields. Entries give
$\mathcal D(F)$ and sufficient conditions for asymptotic MSPE equivalence,
under the corresponding latent spectral and design assumptions.
Here $q$ denotes the number of latent Gaussian fields.}
\label{tab:models}
\begin{tabular}{@{}>{\raggedright\arraybackslash}p{30mm}
>{\raggedright\arraybackslash}p{10mm}
>{\raggedright\arraybackslash}p{44mm}
>{\raggedright\arraybackslash}p{58mm}@{}}
\toprule
Field & $q$ & $\mathcal D(F)$ & Sufficient conditions\\
\midrule
Log-Gaussian & $1$ & $\omega^2e^{2\mu+2\omega^2}$ & All finite $\nu>0$\\
Skew-Gaussian & $2$ & $\delta^2+\tau^2$ & $0<\nu<3/2$\\
Weibull & $2$ & $2\tau^2\Gamma(2/k)/k^2$ & $0<\nu<1$, every $k>0$\\
Clayton-like, uniform & $q_0+2$ & $2q_0/\{3(3q_0-2)\}$ & $0<\nu<1$, $q_0\ge1$\\
Clayton-like, $Q$ & $q_0+2$ & $I_{q_0}(Q)$ & $0<\nu<1$, $\int_0^1Q^2<\infty$, $I_{q_0}(Q)<\infty$\\
Spatial-scale modulation & $2$ & $(1+\sigma^2)e^{2\sigma^2}$ & All finite $\nu>0$\\
Product & $2$ & $2$ & All finite $\nu>0$\\
Gamma, Gaussian squares & $q$ & $4q\tau^2$ & All finite $\nu>0$\\
\bottomrule
\end{tabular}
\end{table}

The conditions in Table~\ref{tab:models} are sufficient, not necessary,
and refer to the particular Gaussian-functional representations specified
here or in Appendix~\ref{app:examples}. The spatial-probit field is treated
separately in Section~\ref{sec:probit}, since its threshold map falls
outside the positive Sobolev-regularity results.

\FloatBarrier
\section{Scope and limitations}
\label{sec:limits}

The preceding results establish asymptotic MSPE equivalence under explicit
distributional and regularity conditions. We now examine two complementary
questions: what this equivalence implies for the distribution of the
prediction error, and how it can fail when the assumptions are violated.

We first characterize the limiting standardized prediction error and identify
the additional condition required for asymptotically Gaussian calibration.
We then show, through two counterexamples, that both transformation regularity
and the Gaussian-functional structure are substantive: an invertible
transformation can retain a nonvanishing nonlinear-prediction advantage when
it is insufficiently regular relative to the latent smoothness, while a
continuous non-Gaussian process with Mat\'ern covariance can exhibit a
different prediction order altogether. Finally, the spatial-probit model
illustrates a discontinuous transformation that belongs to the
Gaussian-functional representation but falls outside the regularity conditions
of the positive results.

\subsection{Limiting prediction errors and Gaussian calibration}
\label{sec:error-calibration}

MSPE equivalence controls average squared prediction error, but does not
determine the limiting distribution of that error. The same Gaussian
conditioning argument used for the latent-information MSPE shows that the
expected squared gradient of the transformation determines the asymptotic
MSPE, whereas variation in the gradient norm determines whether the
standardized error is Gaussian.

\begin{proposition}
\label{prop:prediction-error-limit}
Suppose the assumptions of Theorem~\ref{thm:main} or
Theorem~\ref{thm:general} hold, or those of
Proposition~\ref{prop:exponential-sandwich} hold along a sequence with
$v_n\to0$. Let $\boldsymbol G,\boldsymbol U$ be independent
$N_q(0,I_q)$ vectors and define
\begin{equation}\label{eq:general-error-limit}
 W_F
 :=
 \frac{\nabla F(\boldsymbol G)^{\mathsf T}\boldsymbol U}
      {\sqrt{\mathcal D(F)}}.
\end{equation}
Then, for each $A\in\{\mathrm O,\mathrm{NL},\mathrm L,\rho\}$,
\begin{equation}\label{eq:standardized-error-convergence}
 \frac{Y(s_\star)-\widehat Y_n^A}
      {\sqrt{\operatorname{MSPE}_n^A}}
 \xrightarrow{d}
 W_F.
\end{equation}
The convergence holds jointly with
$
 \boldsymbol Z_\star
 =
 \bigl(Z_1(s_\star),\ldots,Z_q(s_\star)\bigr)^{\mathsf T},
$
with joint limit $(\boldsymbol G,W_F)$. Moreover,
\begin{equation}\label{eq:error-mixture-variance}
 W_F\mid\boldsymbol G
 \sim
 N\!\left(
 0,
 \frac{\|\nabla F(\boldsymbol G)\|^2}{\mathcal D(F)}
 \right),
\end{equation}
where $N(0,0)$ denotes point mass at zero if the gradient vanishes.
Thus $W_F$ has mean zero and variance one. In particular,
\begin{equation}\label{eq:normal-error-criterion}
 W_F\sim N(0,1)
 \quad\Longleftrightarrow\quad
 \|\nabla F(\boldsymbol G)\|^2
 =
 \mathcal D(F)
 \quad\text{almost surely}.
\end{equation}
For the latent-information predictor alone,
\eqref{eq:standardized-error-convergence} requires only
$F\in W^{1,2}(\gamma_q)$,
$0<\mathcal D(F)<\infty$, and $v_n\to0$.
\end{proposition}

The proof in Appendix~\ref{app:prediction-error} uses Gaussian
conditioning and an exact $L^2$ identity for $W^{1,2}(\gamma_q)$
transformations, obtained by Gaussian integration by parts.
The asymptotic MSPE comparisons transfer the oracle limit to the
three observed-data predictors. No invertibility assumption is required.

The proposition separates two notions that are easily conflated. The
coefficient
\[
 \mathcal D(F)
 =
 \E\|\nabla F(\boldsymbol G)\|^2
\]
determines the first-order MSPE, whereas the random quantity
$\|\nabla F(\boldsymbol G)\|^2$ determines the shape of the limiting
standardized error distribution. Hence MSPE efficiency does not by itself
imply Gaussian calibration.

When~\eqref{eq:normal-error-criterion} holds, for every fixed
$0<\eta<1$,
\[
 \widehat Y_n^A
 \pm
 z_{1-\eta/2}
 \sqrt{\operatorname{MSPE}_n^A}
\]
has marginal coverage tending to $1-\eta$. If the gradient norm is not
almost surely constant, the limiting error is instead a nonnormal Gaussian
scale mixture, so standard normal quantiles are not generally calibrated.

The log-Gaussian and skew-Gaussian fields give contrasting examples. For
the log-Gaussian transformation in~\eqref{eq:lognormal-example},
\[
 F'(z)=\omega e^{\mu+\omega z},
\]
and therefore
\begin{equation}\label{eq:lognormal-error-limit}
 W_F
 =
 e^{\omega G-\omega^2}U,
 \qquad
 G,U\ \text{independent }N(0,1).
\end{equation}
Its fourth moment is $3e^{4\omega^2}>3$, so the limiting standardized error
is non-Gaussian despite asymptotic MSPE efficiency. By contrast, the
skew-Gaussian transformation~\eqref{eq:skew-example} has weak gradient
$
 (\delta\operatorname{sign}z_1,\tau)
$
where $\operatorname{sign}(z)$ is $1$ for $z>0$, $-1$ for $z<0$, and
$0$ at zero. The gradient has constant squared norm
$\delta^2+\tau^2=\mathcal D(F)$ almost surely. Its standardized prediction errors are
therefore asymptotically standard normal whenever the corresponding
MSPE-equivalence conditions hold.

For the log-Gaussian model, the failure of Gaussianity under global MSPE
scaling can nevertheless be removed by a model-specific local normalization.
The joint convergence in Proposition~\ref{prop:prediction-error-limit} gives,
for each $A\in\{\mathrm{NL},\mathrm L,\rho\}$,
\begin{equation}\label{eq:lognormal-local-standardization}
 \frac{Y(s_\star)-\widehat Y_n^A}
      {\omega|\widehat Y_n^A|\sqrt{v_n}}
 \xrightarrow{d}
 N(0,1).
\end{equation}
The ratio may be assigned any value on the event $\widehat Y_n^A=0$;
its probability tends to zero, as shown in
Appendix~\ref{app:error-limit-examples}. Intervals based on this local
scale have asymptotically correct marginal Gaussian coverage. This normalization is specific to the log-Gaussian model:
neither it nor Proposition~\ref{prop:prediction-error-limit} establishes
convergence of the full predictive distribution conditional on the
observations.

\subsection{A regularity boundary within the latent-Gaussian class}
\label{sec:regularity-boundary}

The asymptotic criterion for rough latents cannot be used at arbitrary
smoothness, even for a monotone transformation.
\begin{proposition}\label{prop:kink}
Let $Z$ be a standardized stationary Gaussian process on a compact interval,
with Mat\'ern correlation
$\rho(h)=(1+|h|)e^{-|h|}$ and smoothness $\nu=3/2$. Set
\begin{equation}\label{eq:kink-field}
 Y(s)=Z(s)+\epsilon|Z(s)|,\qquad 0<\epsilon<1.
\end{equation}
Under Assumption~\ref{ass:design}, simple kriging with the true mean and
covariance satisfies
\begin{equation}\label{eq:kink-ratio}
 \frac{\MSPEL}{\MSPEN}
 \longrightarrow 1+\frac{2\epsilon^2}{\pi(1+\epsilon^2)}>1.
\end{equation}
The transformation is invertible and globally Lipschitz, belongs to
$W^{1,2}(\gamma_1)$, and has moments of every order, but is not in
$W^{3/2,2}(\gamma_1)$.
\end{proposition}

For the same transformation the positive theorem applies to Mat\'ern
smoothness $\nu<3/2$. Thus compatibility of latent smoothness and
transformation regularity is a genuine statistical issue, not merely a
technical feature of the sufficient convolution bound. The proof in
Appendix~\ref{app:kink} identifies the extra
high-frequency coefficient created by the kink and uses the exact observed-data
MSPE supplied by invertibility.
Moreover, the spectral densities in this example remain proportional at high
frequencies, but with a coefficient different from $\mathcal D(F)$.
Lemma~\ref{lem:spectral-comparison} therefore gives
$\MSPER/\MSPEL\to1$: using the latent correlation retains the same
nonunit ratio to the nonlinear MSPE. Thus covariance simplification can
remain efficient within the linear class while the gain from nonlinear
prediction persists. This example does not establish necessity of the full
Sobolev condition for every transformation.

\subsection{A continuous counterexample with Mat\'ern covariance}
\label{sec:covariance-counterexample}

The need for distributional assumptions is not just a technical feature of
Hermite expansions.

\begin{proposition}
\label{prop:counterexample}
There exists a stationary, continuous, non-Gaussian process on $\R$ with
finite moments of every order and covariance
\begin{equation}\label{eq:counter-covariance}
 C(h)=\frac{\lambda}{4a^3}(1+a|h|)e^{-a|h|},
                 \qquad a,\lambda>0,
\end{equation}
which is Mat\'ern with smoothness $3/2$, such that on the designs
\[
 S_N=\{j/N:-N\le j\le N,\ j\ne0\}\subset[-1,1],
                 \qquad s_\star=0,
\]
the best linear and optimal nonlinear MSPEs satisfy
\begin{equation}\label{eq:counter-rates}
 \operatorname{MSPE}_N^{\mathrm L}\asymp N^{-3},\qquad
 \operatorname{MSPE}_N^{\mathrm{NL}}=O(N^{-4}),\qquad
 \operatorname{MSPE}_N^{\mathrm{NL}}/\operatorname{MSPE}_N^{\mathrm L}\longrightarrow0.
\end{equation}
Here the design has $2N$ observations; $N$ denotes the reciprocal mesh.
The dyadic subsequence $N=2^m$ is nested and dense in $[-1,1]$.
\end{proposition}

The construction is the Poisson moving average
\[
 Y(t)=\int_{-\infty}^t(t-u)e^{-a(t-u)}\,\mathcal N(\dd u),
\]
where $\mathcal N$ is a homogeneous Poisson random measure with intensity
$\lambda\dd u$. It is therefore a
L\'evy-driven process. Non-Gaussian Mat\'ern fields driven by noise are
an established modelling alternative; see \citet{Bolin2014,WallinBolin2015}
for SPDE constructions, inference and prediction. The one-sided Poisson
filter here is not identified with those specific models: with non-Gaussian
noise, equal covariance spectral densities need not imply equal field laws.
After multiplication by a known smooth function, its paths are locally
piecewise linear. A rule based on four nearby observations exploits the
rarity of two local Poisson points, whereas the best linear MSPE is fixed
by the covariance. Appendix~\ref{app:counterexample} proves the comparison.

The positive theorem therefore does not say that every non-Gaussian field
with Mat\'ern covariance behaves as a Gaussian field for point prediction.
It says that, for the specified class of functions of finitely many Gaussian latent fields and its
regularity conditions, the additional distributional information does not
improve the leading MSPE. The counterexample has a different
local distributional mechanism despite its familiar covariance.

\subsection{Spatial probit: a discontinuous transformation}
\label{sec:probit}

Latent Gaussian threshold models for binary spatial data are studied by
\citet{HeagertyLele1998}. We consider a stationary, noise-free version with
no nugget and fix the latent variance to one:
\begin{equation}\label{eq:probit-field}
 Y(s)=\ind_{\{Z(s)>a\}},\qquad a\in\R,\qquad
 \E Z(s)=0,\qquad \Var Z(s)=\sigma^2=1.
\end{equation}
Here $Z$ has correlation $\rho$. Fixing the variance removes the latent
scale ambiguity: multiplying both $Z$ and its threshold by the same positive
constant leaves $Y$ unchanged. The threshold $a$ is fixed, and the known
mean is $p_a=\Phi(-a)$. Write $\Phi_2(x,y;u)$ and $\phi_2(x,y;u)$
for the bivariate standard normal distribution function and density with
correlation $u$.

\begin{proposition}
\label{prop:probit-boundary}
For~\eqref{eq:probit-field}, let $F_a(z)=\ind_{\{z>a\}}$ and
$K_a=K_{F_a}$. Then
\begin{equation}\label{eq:probit-covariance-map}
 K_a(u)=\Phi_2(-a,-a;u)-p_a^2,\qquad
 K_a(1)=p_a(1-p_a),\qquad C_Y(h)=K_a(\rho(h)).
\end{equation}
The endpoint values are understood by continuity. Although $F_a$ is bounded,
$F_a\notin W^{1,2}(\gamma_1)$ and $K_a'(1^-)=\infty$. For every finite
nonsingular latent design,
\begin{equation}\label{eq:probit-oracle}
 \MSPEO=K_a(1)-K_a(1-v_n).
\end{equation}
If the latent spectral and dense-design assumptions hold, then
\begin{equation}\label{eq:probit-oracle-asymptotic}
 \MSPEO\sim\frac{\phi(a)}{\sqrt{\pi}}\,v_n^{1/2}.
\end{equation}
At $a=0$, the covariance and oracle MSPE simplify to
\begin{equation}\label{eq:probit-symmetric}
 C_Y(h)=\frac{1}{2\pi}\arcsin\{\rho(h)\},\qquad
 \MSPEO=\frac{1}{2\pi}\arccos(1-v_n)
       \sim\frac{1}{\pi\sqrt{2}}\,v_n^{1/2}.
\end{equation}
\end{proposition}

The proof and the corresponding small-lag covariance behaviour are given in
Appendix~\ref{app:probit}. The construction already has the form
\eqref{eq:model} with $q=1$; it is transformation regularity, rather than
the representation, that fails. In particular, the ordinary derivative of
an indicator is zero away from its threshold, but this is not a
square-integrable weak derivative across the jump.

For this binary target, the optimal observed-data predictor is the
conditional success probability
$\YN=\PP\{Y(s_\star)=1\mid\mathcal G_n\}$, not a hard classification
rule. The oracle instead observes the latent Gaussian values and need not
be asymptotically equivalent in MSPE. The finite-sample hierarchy
$\MSPEO\le\MSPEN\le\MSPEL\le\MSPER$ remains valid. However, neither
$\MSPEL/\MSPEN\to1$ nor $\MSPER/\MSPEN\to1$ is established by this
example, and no further efficiency conclusion is asserted for it. The
square-root oracle scale alone does not determine the MSPE ratios.

\section{Numerical examples}\label{sec:numerical}

We illustrate the finite-sample prediction comparison with two scalar
invertible transformations: log-Gaussian and sinh--arcsinh. Their
invertibility is useful because observing the transformed field at a finite
set of locations is equivalent to observing the latent Gaussian field at the
same locations. Consequently, the optimal observed-data MSPE is available
exactly through Proposition~\ref{prop:exact}, and no Monte Carlo approximation
of the conditional mean is required.

The log-Gaussian model provides a fully analytic benchmark, whereas the
sinh--arcsinh model gives a second smooth transformation with separate
asymmetry and tail parameters. The examples are deliberately scalar so that
the nonlinear benchmark is exact; the multilatent scope of the theory is
illustrated analytically in Section~\ref{sec:examples}.

\subsection{Design and deterministic MSPE comparison}

The two experiments use different spatial dimensions. The log-Gaussian
experiment is carried out on the fixed domain $D=[0,1]$, whereas the
sinh--arcsinh experiment uses $D=[0,1]^2$. In both cases, the standardized
latent field has Mat\'ern smoothness $\nu=1/2$, hence exponential correlation
\begin{equation}\label{eq:numerical-exponential-correlation}
 \rho_\alpha(h)
 =
 \exp\{-\|h\|/\alpha\},
 \qquad
 \alpha\in\{0.10,0.30,0.60\}.
\end{equation}
Here $\|\cdot\|$ denotes absolute value in dimension one and Euclidean
distance in dimension two. There is no nugget, and all model parameters are
treated as known.

For the log-Gaussian experiment, the prediction target is $s_\star=1/3$ and
we use the nested dyadic designs
\begin{equation}\label{eq:numerical-design-log}
 S_{n_j}^{\mathrm{LG}}
 =
 \{i2^{-j}:i=0,\ldots,2^j\},
 \qquad
 n_j=2^j+1,
 \qquad
 j=4,\ldots,12.
\end{equation}
Thus
$
 n_j
 \in
 \{17,33,65,129,257,513,1025,2049,4097\},
$
and the target is never observed.

For the sinh--arcsinh experiment, the prediction target is
$s_\star=(1/3,2/5)$ and we use the nested tensor-product designs
\begin{equation}\label{eq:numerical-design-sas}
 S_{n_j}^{\mathrm{SAS}}
 =
 \left\{
 (i2^{-j},k2^{-j}):i,k=0,\ldots,2^j
 \right\},
 \qquad
 n_j=(2^j+1)^2,
 \qquad
 j=2,\ldots,6.
\end{equation}
Hence
$
 n_j
 \in
 \{25,81,289,1089,4225\},
$
and again the target is never observed. Both sequences correspond to
fixed-domain asymptotics.

For a fixed design and range parameter $\alpha$, let
$\boldsymbol w_{\rho,n}$ and $\boldsymbol w_{Y,n}$ denote the
latent-correlation and true-covariance kriging weights defined in
Section~\ref{sec:setup}. The latent conditional variance and the two linear
MSPEs are
\begin{align}
 v_n
 &=
 1-\boldsymbol r_n^{\mathsf T}\boldsymbol w_{\rho,n},
 \notag\\
 \MSPEL
 &=
 K_F(1)
 -
 \boldsymbol c_{Y,n}^{\mathsf T}\boldsymbol w_{Y,n},
 \label{eq:numerical-linear-risks}\\
 \MSPER
 &=
 K_F(1)
 -
 2\boldsymbol w_{\rho,n}^{\mathsf T}\boldsymbol c_{Y,n}
 +
 \boldsymbol w_{\rho,n}^{\mathsf T}
 C_{Y,n}\boldsymbol w_{\rho,n}.
 \notag
\end{align}
These quantities are obtained by deterministic finite-dimensional linear
algebra.

Because both transformations considered below are scalar and invertible,
$
 \sigma\{Y(S_n)\}
 =
 \sigma\{Z(S_n)\}.
$
Therefore
\begin{equation}\label{eq:numerical-nonlinear-risk}
 \MSPEN
 =
 \MSPEO
 =
 K_F(1)-K_F(1-v_n)
\end{equation}
holds exactly for every design and every $\alpha$. Thus all comparisons below
use the true optimal predictor based on the observed non-Gaussian data as
benchmark.

\subsection{Log-Gaussian and sinh--arcsinh models}

For the log-Gaussian field we take
\begin{equation}\label{eq:numerical-log}
 F(z)
 =
 \exp(\mu+\omega z),
 \qquad
 (\mu,\omega)=(0,2).
\end{equation}
The choice $\omega=2$ gives a strongly right-skewed marginal distribution,
making finite-sample differences between linear and nonlinear prediction
visible. Since the covariance map is available in closed form,
\eqref{eq:lognormal-exact-risk} gives
\begin{equation}\label{eq:numerical-log-risk}
 \MSPEN
 =
 e^{2\mu+2\omega^2}
 \left\{
 1-e^{-\omega^2v_n}
 \right\}.
\end{equation}
Thus all MSPEs in this experiment are available analytically up to the
finite-dimensional kriging calculations.

For the sinh--arcsinh field we use
\begin{equation}\label{eq:numerical-sas}
 F(z)
 =
 \mu+\sigma
 \sinh\left\{
 \frac{\asinh z+\varepsilon}{\kappa}
 \right\},
 \qquad
 (\mu,\sigma,\varepsilon,\kappa)
 =
 (0,1,0.75,0.80).
\end{equation}
This gives a smooth strictly increasing transformation with both asymmetry
and tail modification. Its marginal density is available analytically by
change of variables, whereas its covariance map is evaluated numerically.

Let $(z_a,w_a)_{a=1}^M$ denote Gauss--Hermite nodes and weights for
integration with respect to standard Gaussian measure, and put
\[
 \widehat m_M
 =
 \sum_{a=1}^M w_aF(z_a).
\]
For $-1<u<1$, we evaluate
\begin{equation}\label{eq:numerical-sas-quadrature}
 K_{F,M}(u)
 =
 \sum_{a,b=1}^M
 w_aw_bF(z_a)
 F\left(
 uz_a+\sqrt{1-u^2}\,z_b
 \right)
 -
 \widehat m_M^{\,2}.
\end{equation}
At $u=1$, the corresponding one-dimensional Gaussian expectation is used
directly.

We take $M=40$. As a numerical stability check, all configurations were
recomputed with $M=80$ and $M=120$. Relative to $M=40$, either RMSE ratio
defined below changed by less than $7\times10^{-7}$. Thus quadrature error is
negligible at the scale of the prediction comparisons.

\subsection{RMSE ratios and finite-sample gains}

For $A\in\{\mathrm{NL},\mathrm L,\rho\}$, define
\begin{equation}\label{eq:numerical-rmse}
 \operatorname{RMSE}_n^A
 :=
 \left(
 \operatorname{MSPE}_n^A
 \right)^{1/2}.
\end{equation}
We report the two ratios
\begin{equation}\label{eq:numerical-ratios}
 R_{\mathrm L}(n)
 :=
 \frac{\operatorname{RMSE}_n^{\mathrm L}}
      {\operatorname{RMSE}_n^{\mathrm{NL}}}
 =
 \left(
 \frac{\MSPEL}{\MSPEN}
 \right)^{1/2},
 \qquad
 R_{\rho}(n)
 :=
 \frac{\operatorname{RMSE}_n^{\rho}}
      {\operatorname{RMSE}_n^{\mathrm{NL}}}
 =
 \left(
 \frac{\MSPER}{\MSPEN}
 \right)^{1/2}.
\end{equation}
Proposition~\ref{prop:finite-gains} gives
\begin{equation}\label{eq:numerical-rmse-order}
 1
 \le
 R_{\mathrm L}(n)
 \le
 R_{\rho}(n)
\end{equation}
for every finite design.

The two departures in~\eqref{eq:numerical-rmse-order} have a direct
interpretation. Using the finite-sample gains from
Proposition~\ref{prop:finite-gains},
\begin{equation}\label{eq:numerical-rmse-gains}
 R_{\mathrm L}(n)^2-1
 =
 \frac{G_{\mathrm N,n}}{\MSPEN},
 \qquad
 R_{\rho}(n)^2-R_{\mathrm L}(n)^2
 =
 \frac{G_{\mathrm C,n}}{\MSPEN}.
\end{equation}
Thus the distance of $R_{\mathrm L}(n)$ from one measures the relative cost
of restricting prediction to the linear class, whereas the separation between
$R_{\rho}(n)$ and $R_{\mathrm L}(n)$ measures the additional cost of replacing
the true observed covariance by the latent Gaussian correlation.
Both examples also satisfy the finite-sample bound of
Proposition~\ref{prop:exponential-sandwich}:
\begin{equation}\label{eq:numerical-rmse-bound}
 1
 \le
 R_{\mathrm L}(n)
 \le
 R_{\rho}(n)
 \le
 \left\{
 \frac{\mathcal D(F)v_n}
 {K_F(1)-K_F(1-v_n)}
 \right\}^{1/2}.
\end{equation}
This bound is not used to construct the numerical curves; it provides an
analytic finite-sample reference for the same configurations.

Figure~\ref{fig:numerical-all} summarizes the results. The two experiments
illustrate different finite-sample decompositions of the same asymptotic
equivalence.
For the log-Gaussian field, both simplifications have a visible effect at
small sample sizes. With $\alpha=0.10$ and $n=17$,
\[
 R_{\mathrm L}(n)\approx1.0589,
 \qquad
 R_{\rho}(n)\approx1.1065.
\]
Thus both the restriction to linear prediction and the additional replacement
of the true covariance by the latent correlation contribute appreciably to
the finite-sample MSPE. By $n=4097$, both ratios are approximately $1.0011$.

For the sinh--arcsinh field, the situation is different. The two linear
predictors are almost indistinguishable: their RMSE ratios differ by less than
$4.5\times10^{-4}$ over all configurations considered. Hence essentially all
of the visible finite-sample discrepancy from the optimal predictor is due to
the restriction to linear prediction, while replacing the true observed
covariance by the latent correlation has almost no additional effect.

Convergence along the design sequence need not be monotone. For example,
when $\alpha=0.10$, the sinh--arcsinh true-covariance ratio increases from
approximately $1.0040$ at $n=25$ to $1.0100$ at $n=289$, and then decreases
to $1.0045$ at $n=4225$. This is consistent with
Theorem~\ref{thm:main}, which concerns the limiting ratios rather than their
finite-sample monotonicity.

The two experiments differ in both transformation and spatial dimension and
are therefore not intended to isolate a dimension effect.

\begin{figure}[htbp]
\centering

\begin{minipage}[t]{0.485\textwidth}
\centering
\numericalpanel{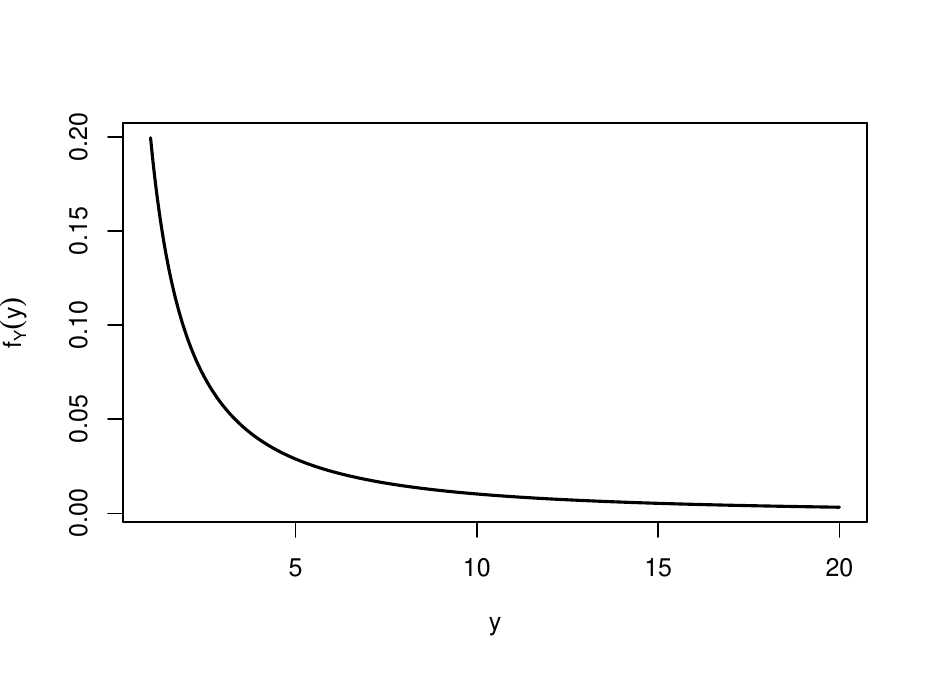}
{Log-Gaussian marginal density}
\par\smallskip
{\small (a) Log-Gaussian marginal density}
\end{minipage}\hfill
\begin{minipage}[t]{0.485\textwidth}
\centering
\numericalpanel{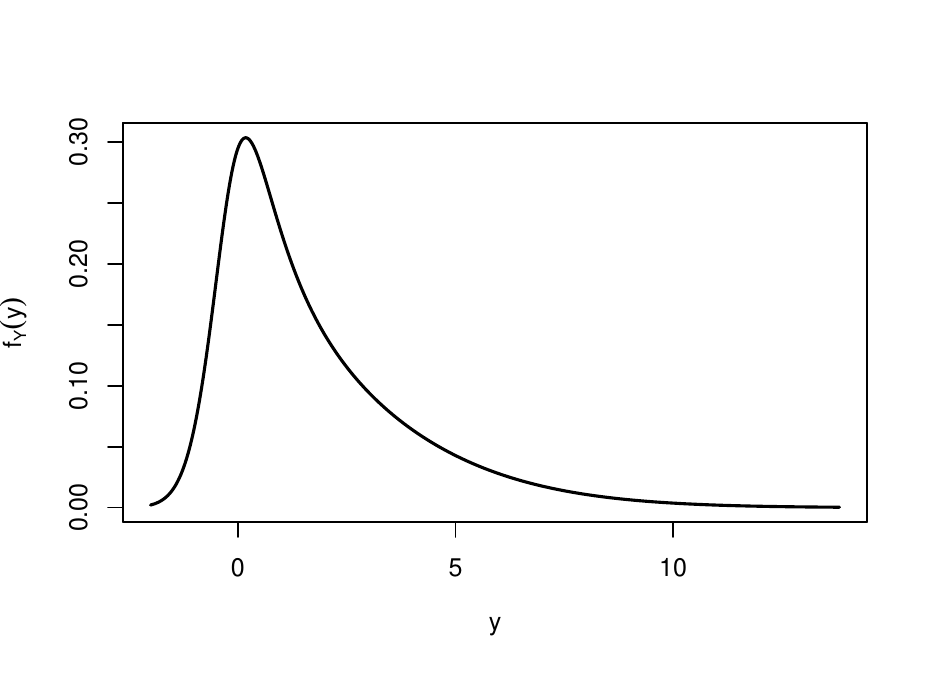}
{Sinh--arcsinh marginal density}
\par\smallskip
{\small (b) Sinh--arcsinh marginal density}
\end{minipage}

\medskip

\begin{minipage}[t]{0.485\textwidth}
\centering
\numericalpanel{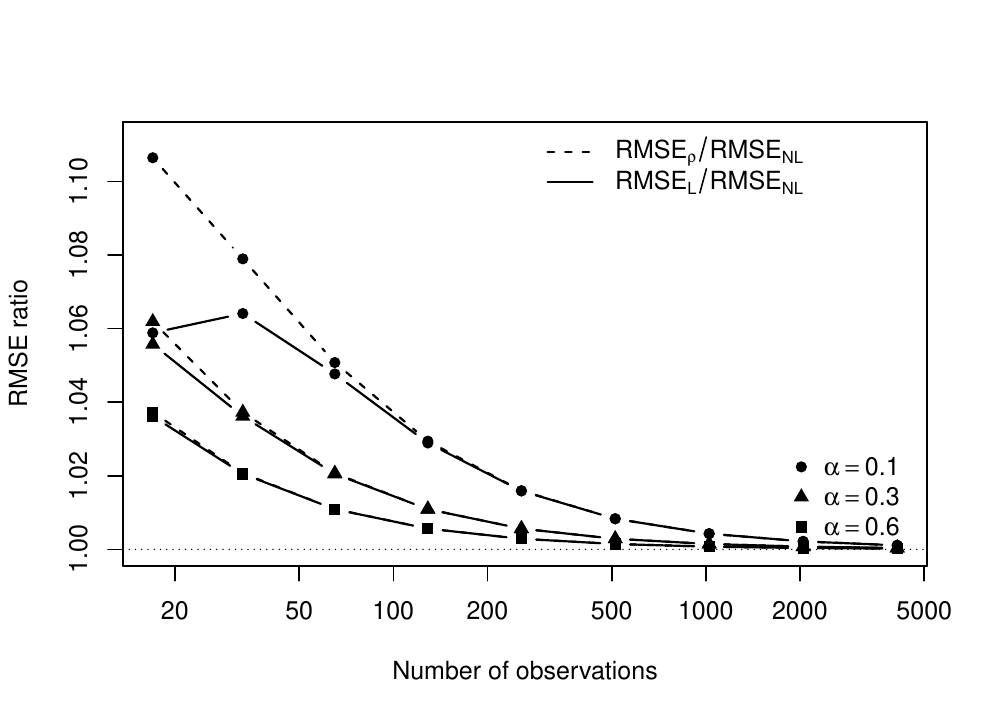}
{Log-Gaussian RMSE ratios}
\par\smallskip
{\small (c) Log-Gaussian, $D=[0,1]$}
\end{minipage}\hfill
\begin{minipage}[t]{0.485\textwidth}
\centering
\numericalpanel{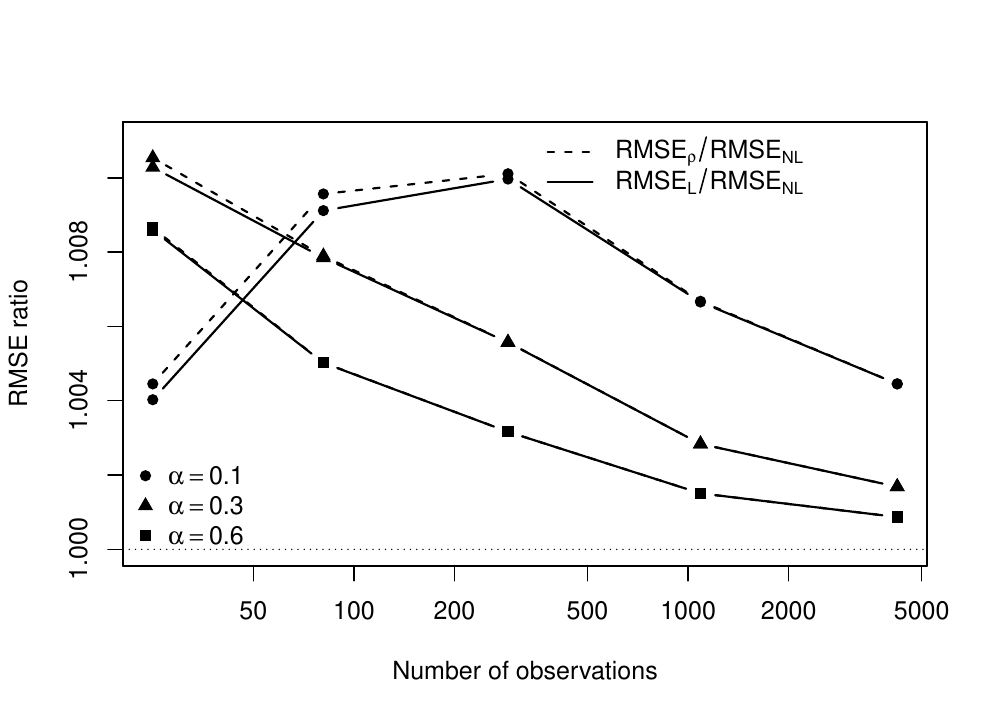}
{Sinh--arcsinh RMSE ratios}
\par\smallskip
{\small (d) Sinh--arcsinh, $D=[0,1]^2$}
\end{minipage}

\caption{Numerical illustrations for the log-Gaussian and sinh--arcsinh
fields
with underlying correlation
$
\rho_\alpha(h)
 =
 \exp\{-\|h\|/\alpha\}$. Panels~(a)--(b) show the marginal densities, and panels~(c)--(d)
show RMSE ratios relative to the optimal observed-data predictor.
Solid lines represent $R_{\mathrm L}(n)$ and dashed lines represent
$R_{\rho}(n)$. Circle, triangle, and square symbols correspond respectively
to $\alpha=0.10$, $0.30$, and $0.60$; the dotted horizontal line is one.}
\label{fig:numerical-all}
\end{figure}

\FloatBarrier
\section{Discussion}\label{sec:discussion}

We studied two simplifications in non-Gaussian spatial prediction: replacing
the optimal conditional mean by kriging, and replacing the true covariance of
the observed field by the latent Gaussian correlation. For the Gaussian-functional
class considered here, under compatible transformation regularity,
\[
 \MSPEO\sim\MSPEN\sim\MSPEL\sim\MSPER
 \sim\mathcal D(F)v_n.
\]
Thus latent-correlation kriging, applied directly to the observed non-Gaussian
data, attains the same first-order MSPE as the optimal observed-data
conditional mean. Correct covariance modelling and nonlinear prediction may
still improve finite-sample performance, but their gains are asymptotically
negligible under the stated conditions.

The quantity $\mathcal D(F)=\E\|\nabla F(\boldsymbol G)\|^2$ plays a central
role: it is both the endpoint slope of the covariance map and the coefficient
governing the observed high-frequency spectral behaviour and the leading
prediction MSPE. The Gaussian-Sobolev assumptions also ensure that the
transformed field inherits the corresponding mean-square regularity of the
latent field. For rough latent fields, a finite expected squared weak gradient
is sufficient under radial monotonicity of the spectral density; smoother
latent fields require higher transformation regularity. The required order is
independent of both the spatial dimension and the number of latent components,
although finite-sample constants need not be.

The prediction-error limit shows that MSPE efficiency and Gaussian uncertainty
calibration are distinct. Standardized errors converge in general to a
Gaussian scale mixture whose random variance depends on
$\|\nabla F(\boldsymbol G)\|^2$. A standard normal limit occurs precisely when
this gradient norm is almost surely constant. The log-Gaussian and
skew-Gaussian examples illustrate the distinction, while the former also
admits a model-specific local normalization restoring marginal Gaussian
calibration.

Relative to the scalar Mat\'ern setting of \citet{HiranoYajima2013}, our
results allow multiple latent fields and noninvertible transformations, cover
broader spectral classes, and include an explicit finite-sample bound for
exponential correlation. In that case the whole MSPE hierarchy is bounded
between $\Delta_F(v_n)$ and $\mathcal D(F)v_n$, so asymptotic efficiency
requires only $v_n\to0$, including designs that accumulate locally around the
prediction target.

The counterexamples clarify the scope of the theory. An invertible Lipschitz
transformation can retain a nonvanishing nonlinear-prediction advantage when
its regularity is insufficient relative to latent smoothness, and a
non-Gaussian process with Mat\'ern covariance can admit a nonlinear predictor
with a strictly faster MSPE rate than any linear predictor. The spatial-probit
example further shows that the Gaussian-functional representation alone is not
enough when the transformation lacks the required Sobolev regularity.

Our numerical examples provide tractable finite-sample illustrations rather
than exhaustive coverage of the multilatent class. Throughout, the mean and
covariance parameters are treated as known and there is no observation error.
Allowing a nugget discontinuity in the common latent correlation leads to a
different asymptotic prediction problem, because the prediction error then
contains a positive irreducible component. Relative-efficiency comparisons
must therefore be formulated in terms of excess rather than total MSPE.
Whether the first-order equivalence established here survives under such a
discontinuity is left for future work.

Several of the non-Gaussian models considered here are implemented
in the R package \texttt{GeoModels} \citep{GeoModels}, which provides
kriging based on the covariance of the observed process.

\clearpage

\appendix
\section*{Organization of the appendices}

All proofs, model-specific calculations, and extensions are collected here.
Appendix~\ref{app:Hermite} first proves the finite-sample comparison, then
establishes the Hermite, weak-derivative, and spectral identities used to
prove the exact oracle MSPE and preservation of mean-square regularity.
Appendix~\ref{app:convolutions} proves the observed spectral-tail result.
Appendix~\ref{app:comparison} establishes the spectral comparison of
kriging MSPEs, following the perspective of \citet{Stein1993} but proving
directly the result used here. Appendix~\ref{app:main-proofs} follows the
order of the main statements: Theorem~\ref{thm:main},
Corollaries~\ref{cor:latent-Matern} and~\ref{cor:rough},
Proposition~\ref{prop:exponential-sandwich}, and Theorem~\ref{thm:general}.
Appendix~\ref{app:examples} verifies the main examples and collects
supplementary smooth and polynomial constructions. Appendix~\ref{app:prediction-error} proves
Proposition~\ref{prop:prediction-error-limit} and the local log-Gaussian
normalization. Appendix~\ref{app:boundaries} contains the proofs of
Propositions~\ref{prop:kink}, \ref{prop:counterexample}, and
\ref{prop:probit-boundary}, again in their order in the main text.
Appendix~\ref{app:additional-models} collects the remaining extensions.

\section{Finite-sample comparison, Gaussian identities, and regularity}
\label{app:Hermite}

Throughout the appendices, positive constants denoted by $C$ may change
between occurrences. Dependence relevant to a uniform bound is stated
explicitly. For a scalar or vector random variable $V$ and $1\le r<\infty$,
we write $\|V\|_{L^r}=(\E\|V\|^r)^{1/r}$, with absolute value in the
scalar case. A subscript $\varepsilon$ on an expectation or probability
means that only the independent random signs are averaged over, with the
remaining variables held fixed.

The proofs of Propositions~\ref{prop:finite-gains}, \ref{prop:exact}, and
\ref{prop:ms-regularity} follow their order in the main text. The auxiliary
Gaussian and spectral identities are established before the proofs that use
them.

\subsection{Finite-sample MSPE ordering and gain decomposition}
\label{app:finite-gains}

\begin{proof}[Proof of Proposition~\ref{prop:finite-gains}]
Since
$
 \YN
 =
 \E\{Y(s_\star)\mid\mathcal G_n\},
$
the residual $Y(s_\star)-\YN$ is orthogonal in $L^2$ to every
$\mathcal G_n$-measurable square-integrable random variable. Both $\YL$ and
$\YN$ are $\mathcal G_n$-measurable, so in particular
$Y(s_\star)-\YN$ is orthogonal to $\YL-\YN$. Hence
\[
 Y(s_\star)-\YL
 =
 \{Y(s_\star)-\YN\}
 +
 (\YN-\YL)
\]
is an orthogonal decomposition. Squaring and taking expectations gives
\[
 \MSPEL-\MSPEN
 =
 \E[(\YL-\YN)^2]
 =
 G_{\mathrm N,n}.
\]

For the second comparison, $\YL$ is the best linear predictor based on the
observations with the mean fixed at its true value. The kriging normal
equations therefore imply that the residual $Y(s_\star)-\YL$ is orthogonal
to every centered linear combination of the observations. Moreover,
\[
 \YR-\YL
 =
 (\boldsymbol w_{\rho,n}-\boldsymbol w_{Y,n})^{\mathsf T}
 (\boldsymbol Y_n-m_Y\boldsymbol1_n),
\]
so $\YR-\YL$ is precisely such a centered linear combination. Consequently,
\[
 Y(s_\star)-\YR
 =
 \{Y(s_\star)-\YL\}
 +
 (\YL-\YR)
\]
is again an orthogonal decomposition, and therefore
\[
 \MSPER-\MSPEL
 =
 \E[(\YR-\YL)^2]
 =
 G_{\mathrm C,n}.
\]
Using the covariance matrix of $\boldsymbol Y_n$ gives
\[
 G_{\mathrm C,n}
 =
 (\boldsymbol w_{\rho,n}-\boldsymbol w_{Y,n})^{\mathsf T}
 C_{Y,n}
 (\boldsymbol w_{\rho,n}-\boldsymbol w_{Y,n}).
\]
Since $C_{Y,n}$ is positive definite, this quantity is zero if and only if
$\boldsymbol w_{\rho,n}=\boldsymbol w_{Y,n}$, which gives the stated
strictness condition. The MSPE ordering and
\eqref{eq:two-gain-decomposition} follow immediately.
\end{proof}

\subsection{Hermite representation and Gaussian weak derivatives}
\label{app:Hermite-Sobolev}

We use the normalized Hermite basis and coefficients defined in
\eqref{eq:Hermite}--\eqref{eq:KF-series}. In particular,
$b_{\boldsymbol0}=m_Y$ and $\sum_{\ell\ge1}a_\ell=\Var\{F(\bG)\}$.
The following Gaussian-chaos identities make the derivative criteria
explicit; see also \citet{Janson1997}.

\begin{lemma}\label{lem:derivatives}
Let $F\in L^2(\gq)$ have the expansion~\eqref{eq:Hermite}.
Then $F\in W^{1,2}(\gq)$ if and only if
$\sum_{\bk}|\bk|b_{\bk}^2<\infty$, in which case
\eqref{eq:gradient-energies} holds, with
$\mathcal D_j(F)=\sum_{\bk}k_j b_{\bk}^2$ and
$\mathcal D(F)=\sum_{\ell\ge1}\ell a_\ell$. More generally, if all weak derivatives
through integer order $M$ are in $L^2(\gq)$, then
\begin{equation}\label{eq:Hermite-Sobolev}
 \sum_{\bk}(1+|\bk|)^M b_{\bk}^2<\infty.
\end{equation}
\end{lemma}

\begin{proof}
The derivative of the normalized polynomial is
$h_k'=\sqrt{k}\,h_{k-1}$. Let $e_j$ be the $j$th coordinate vector
in $\R^q$. Gaussian integration by parts gives, for each coordinate,
\[
 \langle\partial_j F,\Psi_{\boldsymbol\ell}\rangle_{L^2(\gq)}
    =\sqrt{\ell_j+1}\,b_{\boldsymbol\ell+e_j}.
\]
The identity holds first for smooth functions with a cutoff, and then for
weak derivatives by approximation. Cutoffs can be removed because the
polynomial and its Gaussian integration-by-parts adjoint are square
integrable against Gaussian measure. Parseval proves
$\norm{\partial_jF}_{L^2(\gq)}^2=\sum_{\bk}k_jb_{\bk}^2$.
Conversely, if these coefficient sums are finite, Hermite partial sums
converge together with their first derivatives in $L^2(\gq)$. Closedness
of the weak derivative, also seen locally where the Gaussian density is
bounded above and below, gives a weak gradient of the limit.

For a multi-index $\boldsymbol a$ of order $M$, iteration gives
\[
 \norm{\partial^{\boldsymbol a}F}_{L^2(\gq)}^2
   =\sum_{\bk\ge\boldsymbol a}
       \frac{\bk!}{(\bk-\boldsymbol a)!}\,b_{\bk}^2.
\]
Here multi-index inequalities and factorials are coordinatewise. Summing
with coefficients $M!/\boldsymbol a!$ uses the identity
\[
 \sum_{|\boldsymbol a|=M,\,\boldsymbol a\le\bk}
 \frac{M!}{\boldsymbol a!}\frac{\bk!}{(\bk-\boldsymbol a)!}
   =(|\bk|)_M,
\]
where $(k)_M=k(k-1)\cdots(k-M+1)$, with value zero for $k<M$.
Since $(1+k)^M\le C_M\{1+(k)_M\}$ for integer $k\ge0$,
Parseval and the finite sum of order-$M$ derivative norms prove
\eqref{eq:Hermite-Sobolev}. Conversely, finiteness of that weighted
coefficient sum makes the Hermite partial sums Cauchy together with all
derivatives through order $M$. Closedness of these weak derivatives proves
the integer-order characterization used in~\eqref{eq:gaussian-sobolev}.
\end{proof}

\subsection{Covariance-map identities and spectral representation}
\label{app:covariance-identities}

\begin{lemma}\label{lem:KF-identities}
For $F\in L^2(\gq)$, the map in~\eqref{eq:KF-definition} has the
expansion~\eqref{eq:KF-series}. It is continuous on $[-1,1]$, analytic
on $(-1,1)$, and
\begin{equation}\label{eq:KF-endpoint}
 \lim_{v\downarrow0}\frac{K_F(1)-K_F(1-v)}{v}
 =\lim_{u\uparrow1}K_F'(u)=\sum_{\ell\ge1}\ell a_\ell\in[0,\infty].
\end{equation}
If $F\in W^{1,2}(\gq)$, this value is $\E\norm{\nabla F(\bG)}^2$ and,
for $-1<u<1$,
\begin{equation}\label{eq:KF-derivative-coupling}
 K_F'(u)=\E\{\nabla F(\bG)^{\mathsf T}\nabla F(\bG_u)\}.
\end{equation}
Consequently, for $0\le v\le1$,
\begin{equation}\label{eq:KF-decrement-integral}
 K_F(1)-K_F(1-v)
 =\int_{1-v}^1 K_F'(u)\dd u
 =\frac12\E[\{F(\bG)-F(\bG_{1-v})\}^2].
\end{equation}
The last equality only requires $F\in L^2(\gq)$.
\end{lemma}

\begin{proof}
For jointly standard normal $G,G'$ with correlation $u$, the identity
\[
 \E[e^{tG-t^2/2}e^{wG'-w^2/2}]=e^{utw}
\]
gives
\begin{equation}\label{eq:Hermite-pair}
 \E\{h_k(G)h_\ell(G')\}=\ind_{\{k=\ell\}}u^k.
\end{equation}
Applying this coordinatewise to $(\bG,\bG_u)$, whose paired coordinates
are independent, proves $K_F(u)=\sum_{\bk\ne0}b_{\bk}^2u^{|\bk|}$
for finite Hermite sums. Approximation in $L^2$ and Cauchy--Schwarz
extend the identity to $F$. Grouping by total degree gives
\eqref{eq:KF-series}. Since $\sum_\ell a_\ell<\infty$, this series converges
absolutely and uniformly on $[-1,1]$, and is analytic inside that
interval. Termwise differentiation on compact subintervals gives
$a_\ell=K_F^{(\ell)}(0)/\ell!$.

For $0<v\le1$,
\begin{equation}\label{eq:product-bound}
 0\le\frac{1-(1-v)^\ell}{v}
 =\sum_{j=0}^{\ell-1}(1-v)^j\le \ell.
\end{equation}
As $v\downarrow0$, the middle expression increases to $\ell$.
Monotone convergence gives the first limit in~\eqref{eq:KF-endpoint};
the same argument applied to $\sum_\ell \ell a_\ell u^{\ell-1}$ for $0\le u<1$
gives the second. Lemma~\ref{lem:derivatives} identifies finiteness of
$\sum_\ell \ell a_\ell$ with $F\in W^{1,2}(\gq)$ and proves the gradient formula.

For the derivative identity, $h_k'=\sqrt{k}\,h_{k-1}$ and
\eqref{eq:Hermite-pair} give
\[
 \E\{\nabla F(\bG)^{\mathsf T}\nabla F(\bG_u)\}
 =\sum_{\bk\ne0}|\bk|b_{\bk}^2u^{|\bk|-1}
 =\sum_{\ell\ge1}\ell a_\ell u^{\ell-1}
\]
first for polynomials. Convergence of Hermite partial sums and their
weak gradients in $L^2$, with Cauchy--Schwarz for the paired expectations,
extends the identity to $W^{1,2}$. The derivative series is uniformly
absolutely convergent on $[0,1]$ when the energy is finite, so integration
gives the first equality in~\eqref{eq:KF-decrement-integral}. The second
follows directly from identical marginal distributions:
\[
 \tfrac12\E[\{F(\bG)-F(\bG_u)\}^2]=K_F(1)-K_F(u).
\]
Taking $u=1-v$ completes the proof.
\end{proof}

When the expected squared gradient is finite, the derivative coupling also gives the exact
remainder identity
\begin{equation}\label{eq:gradient-remainder}
 \mathcal D(F)v-\Delta_F(v)
 =\frac12\int_{1-v}^1
       \E\|\nabla F(\bG)-\nabla F(\bG_u)\|^2\dd u.
\end{equation}
It follows by subtracting the derivative integral from $v\mathcal D(F)$
and using equal Gaussian marginals. The integrand tends to zero as
$u\uparrow1$, first for polynomial gradients and then by $L^2$
approximation. This supplies a probabilistic explanation of the
leading oracle limit without a pointwise Taylor expansion of $F$.
The covariance differentiation identity is a form of Gaussian
interpolation; see \citet{Voigtlaender2021}.

\paragraph{Spectral representation of the observed covariance.}

Let $f$ be the unit-mass spectral density of the common latent correlation
in~\eqref{eq:spectral-convention}. Write $f^{*\ell}$ for its $\ell$-fold
convolution, with $f^{*1}=f$. Since $f^{*\ell}$ has unit mass and Fourier
transform $\rho^\ell$, the covariance expansion~\eqref{eq:KF-series} gives
\begin{equation}\label{eq:observed-spectral-density}
 f_Y(\lambda)=\sum_{\ell\ge1}a_\ell f^{*\ell}(\lambda).
\end{equation}
Indeed, the terms are nonnegative, and Tonelli's theorem gives
$\int_{\R^d}f_Y(\lambda)\dd\lambda=\sum_{\ell\ge1}a_\ell=K_F(1)<\infty$.
Thus the series converges in $L^1(\R^d)$. The same integrability bound
justifies termwise Fourier transformation, whose result is
$\sum_{\ell\ge1}a_\ell\rho(h)^\ell=K_F\{\rho(h)\}=C_Y(h)$.
The constant Hermite term is absent because the covariance uses the
centered field; it contributes no spectral atom. This representation
requires only $F\in L^2(\gq)$ and the existence of the latent spectral
density. No tail, radial-monotonicity, or transformation-regularity
assumption is needed.

\subsection{Conditional copies and the latent-information MSPE}
\label{app:oracle-proof}

\begin{proof}[Proof of Proposition~\ref{prop:exact}]
Write $\bZ_\star=(Z_1(s_\star),\ldots,Z_q(s_\star))^{\mathsf T}$ and
$M_n=\E(\bZ_\star\mid\mathcal H_n)$. Gaussian conditioning gives
\[
 \bZ_\star\mid\mathcal H_n\sim N_q(M_n,v_nI_q).
\]
The vector $M_n$ is a linear function of the Gaussian observations and
is therefore Gaussian. The law of total covariance gives
$\Cov(M_n)=(1-v_n)I_q$, so $M_n\sim N_q\{0,(1-v_n)I_q\}$.

On an enlarged probability space if necessary, let $U,U'$ be independent
$N_q(0,I_q)$ vectors whose pair is independent of $\mathcal H_n$, and put
\[
 X=M_n+\sqrt{v_n}U,\qquad X'=M_n+\sqrt{v_n}U'.
\]
Given $\mathcal H_n$, $X$ and $X'$ are independent and each has the
conditional distribution of $\bZ_\star$. In particular,
\[
 \E\{F(X)F(X')\mid\mathcal H_n\}
 =\bigl[\E\{F(\bZ_\star)\mid\mathcal H_n\}\bigr]^2.
\]
Expanding the squared difference and using the equal conditional second
moments therefore gives
\[
 \Var\{F(\bZ_\star)\mid\mathcal H_n\}
 =\frac12\E\!\left[\{F(X)-F(X')\}^2\mid\mathcal H_n\right].
\]
Taking expectations and applying the
law of iterated expectations,
\begin{align*}
 \MSPEO
 &=\E\!\left[\Var\{F(\bZ_\star)\mid\mathcal H_n\}\right]\\
 &=\frac12\E\!\left[
       \E\!\left\{\{F(X)-F(X')\}^2\mid\mathcal H_n\right\}\right]\\
 &=\frac12\E\!\left[\{F(X)-F(X')\}^2\right].
\end{align*}
All expectations are finite because $F\in L^2(\gq)$.
Unconditionally, $(X,X')$ is jointly Gaussian, with zero mean and
\[
 \Cov(X)=\Cov(X')=\Cov(M_n)+v_nI_q=I_q,
 \qquad \Cov(X,X')=\Cov(M_n)=(1-v_n)I_q.
\]
Thus $(X,X')\stackrel d=(\bG,\bG_{1-v_n})$ in the definition of $K_F$.
Using the identical marginal means gives
\begin{align*}
 \frac12\E\!\left[\{F(X)-F(X')\}^2\right]
 &=\Var\{F(X)\}-\Cov\{F(X),F(X')\}\\
 &=K_F(1)-K_F(1-v_n).
\end{align*}
This also covers $v_n=1$, when $M_n=0$ almost surely. Dividing by $v_n$
and letting $v_n\to0$ proves~\eqref{eq:oracle-asymptotic-mspe} whenever
$K_F'(1^-)<\infty$.
Finally, if $q=1$ and $F$ has a measurable inverse on its image, each
latent observation can be recovered measurably from the corresponding
observed value. Thus $\mathcal H_n\subseteq\mathcal G_n$; the reverse
inclusion always holds because $Y(s)=F\{Z(s)\}$. Consequently,
$\mathcal G_n=\mathcal H_n$, and $\MSPEN=\MSPEO$ for every $n$.
\end{proof}

For subsequent use, the Hermite representation~\eqref{eq:KF-series}
also gives
\begin{equation}\label{eq:oracle-Hermite}
 \MSPEO=\sum_{\ell\ge1}a_\ell\{1-(1-v_n)^\ell\}.
\end{equation}
This coefficient identity is a consequence of the covariance-map formula,
not an ingredient of its probabilistic proof.

Strict positivity of $f$ makes every latent covariance matrix on distinct
sites positive definite: its spectral quadratic form can vanish only
if a finite sum of distinct exponentials vanishes identically. The Schur
complement gives $v_n>0$. For every $\varepsilon>0$, covariance continuity
and density provide an observed site $s_i$ such that
$\E[\{Z(s_\star)-Z(s_i)\}^2]<\varepsilon$. For all $n\ge i$,
\[
 0<v_n\le\E[\{Z(s_\star)-Z(s_i)\}^2]<\varepsilon.
\]
Thus $v_n\to0$, without a geometric rate or sample-path regularity
assumption on $Y$.

\subsection{Preservation of mean-square regularity}
\label{app:regularity-proof}

\begin{proof}[Proof of Proposition~\ref{prop:ms-regularity}]
By~\eqref{eq:observed-covariance},
\[
 C_Y(h)=K_F\{\rho(h)\}.
\]
For $F\in L^2(\gq)$, the Hermite representation
\eqref{eq:KF-series} implies that $K_F$ is continuous at one. Hence
$\rho(h)\to1$ gives
\[
 \E[\{Y(s+h)-Y(s)\}^2]
 =2\bigl[K_F(1)-K_F\{\rho(h)\}\bigr]
 \longrightarrow0,
\]
which proves mean-square continuity.

For the higher-order statement, the spectral criterion for a stationary
field says that the spatial derivative of multi-index $\boldsymbol a$
exists in mean square exactly when
$\int_{\R^d}\prod_{j=1}^d\lambda_j^{2a_j}f(\lambda)\dd\lambda<\infty$.
Since
\[
 \|\lambda\|^{2m}
 =\sum_{|\boldsymbol a|=m}\frac{m!}{\boldsymbol a!}
       \prod_{j=1}^d\lambda_j^{2a_j},
\]
existence of all derivatives of total order at most $m$ is equivalent to
\[
 \int_{\R^d}\|\lambda\|^{2m}f(\lambda)\dd\lambda<\infty.
\]
Here $\boldsymbol a\in\Nzero^d$ and
$\boldsymbol a!=\prod_{j=1}^d a_j!$. Lower-order moments follow from
$\|\lambda\|^{2k}\le1+\|\lambda\|^{2m}$ for $0\le k\le m$.
By~\eqref{eq:observed-spectral-density} and Tonelli's theorem, the
spectral moments of $Y$ can be evaluated term by term.
If $U_1,U_2,\ldots$ are independent random vectors with density $f$, then
$f^{*\ell}$ is the density of $S_\ell=U_1+\cdots+U_\ell$. Evenness of $f$
makes these vectors symmetric and, since the $2m$th moment is finite,
centered. The even-moment bound for sums with independent random signs gives
\[
 \E\|S_\ell\|^{2m}
 \le C_{m,d}\,\ell^m\E\|U_1\|^{2m}.
\]
For completeness, introduce independent signs $\varepsilon_i$ with
$\PP(\varepsilon_i=1)=\PP(\varepsilon_i=-1)=1/2$, independent of the
$U_i$. Expanding the even power coordinatewise yields
\[
 \E_\varepsilon\left|\sum_i\varepsilon_i U_{ij}\right|^{2m}
 \le C_m\left(\sum_i U_{ij}^2\right)^m.
\]
Combine $\|x\|^{2m}\le d^{m-1}\sum_j|x_j|^{2m}$ with
$(\sum_i a_i)^m\le\ell^{m-1}\sum_i a_i^m$ for $a_i\ge0$, and then
average over the $U_i$. This proves the bound, including $m=1$.
Therefore
\begin{align*}
 \int_{\R^d}\|\lambda\|^{2m}f_Y(\lambda)\dd\lambda
 &=\sum_{\ell\ge1}a_\ell\E\|S_\ell\|^{2m}\\
 &\le C_{m,d}\E\|U_1\|^{2m}
       \sum_{\ell\ge1}a_\ell\ell^m<\infty,
\end{align*}
because $F\in W^{m,2}(\gq)$ by~\eqref{eq:gaussian-sobolev}. This proves mean-square differentiability of order $m$ for $Y$.
\end{proof}

For a measurable version of the stationary field, the spectral-moment
condition obtained above also implies that the sample paths are almost surely
in the usual local Sobolev space $W_{\mathrm{loc}}^{m,2}(\R^d)$,
whose weak derivatives through order $m$ are locally square integrable.
This is the stationary specialization of the
general second-order result in \citet[Corollary~1]{Scheuerer2010}. The
fractional-order results and Sobolev embeddings developed there give
stronger sufficient conditions for classical sample-path differentiability.
Independently, on every event where the latent components have continuous,
respectively $C^m$, trajectories, composition with a continuous,
respectively $C^m$, function $F$ preserves that regularity.

\section{Convolution bounds and the observed high-frequency coefficient}
\label{app:convolutions}

This appendix provides the spectral ingredient needed for the main theorem.
The convolution bounds control the spectral density of the transformed field
uniformly over Hermite degree and lead to Proposition~\ref{prop:tail}.

\subsection{A bound uniform in convolution order}
\label{app:uniform-convolution}

\begin{lemma}\label{lem:conv-bound}
If $p>d$ and $f$ is a probability density satisfying
$f(x)\le B(1+\norm x)^{-p}$, then for every $\ell\ge1$,
\begin{equation}\label{eq:conv-bound}
 f^{*\ell}(x)\le B\ell^{p+1}(1+\norm x)^{-p},\qquad x\in\R^d.
\end{equation}
Also $\norm{f^{*\ell}}_\infty\le\norm f_\infty$.
\end{lemma}

\begin{proof}
The case $\ell=1$ is immediate. For $\ell\ge2$, write the convolution as
\[
 f^{*\ell}(x)
 =\int_{(\R^d)^{\ell-1}}
   \left\{\prod_{i=1}^{\ell-1}f(x_i)\right\}
   f\!\left(x-\sum_{i=1}^{\ell-1}x_i\right)
   \dd x_1\cdots\dd x_{\ell-1},
\]
and put $x_\ell=x-\sum_{i=1}^{\ell-1}x_i$, so that
$x_1+\cdots+x_\ell=x$.
At least one summand has norm at least $\norm x/\ell$. Bound the integral
by the sum over these $\ell$ events. On event $i$, bound the factor $f(x_i)$ by
\[
 B(1+\norm x/\ell)^{-p}\le B\ell^p(1+\norm x)^{-p}.
\]
Using the remaining $\ell-1$ variables as integration coordinates and
dropping the event restriction bounds their integral by one.
Summing over $i$ gives~\eqref{eq:conv-bound}. Convolution with a probability
density is an $L^\infty$ contraction, proving the second assertion.
\end{proof}

\subsection{A convolution bound with a dimension-independent order}
\label{app:refined-convolution}

The auxiliary random vectors used here have the spectral density as their
probability density; they are random frequencies, not the latent fields.

\begin{lemma}\label{lem:positive-sum}
If nonnegative random variables $A_1,\ldots,A_\ell$ satisfy
$\mathbb P(A_i>t)\le Ct^{-\nu}$ for $t>0$, then
\begin{equation}\label{eq:positive-sum-bound}
 \mathbb P\left(\sum_{i=1}^\ell A_i>t\right)
 \le C_\nu w_\nu(\ell)t^{-\nu}.
\end{equation}
The constant depends only on $\nu$ and the common tail constant $C$,
not on $\ell$. Independence of the $A_i$ is not required.
\end{lemma}
\begin{proof}
Since the variables are nonnegative, for $t>0$,
\begin{align*}
 \left\{\sum_{i=1}^\ell A_i>t\right\}
 &=\{\max_i A_i>t\}\\
 &\quad\cup\left\{\max_i A_i\le t,\ \sum_{i=1}^\ell A_i>t\right\}.
\end{align*}
For $0<\nu<1$, split according to whether $\max_i A_i>t$. The probability
of that event is at most $C\ell t^{-\nu}$. On its complement, Markov's
inequality and
\[
 \E[A_i\ind_{\{A_i\le t\}}]\le\int_0^t\mathbb P(A_i>x)\dd x
 \le\frac{C}{1-\nu}t^{1-\nu}
\]
give the same order.
For $\nu>1$, every event $B$ of probability $p_B$ satisfies
\[
 \E[A_i\ind_B]\le\int_0^\infty\min\{p_B,Cx^{-\nu}\}\dd x
 \le C_\nu p_B^{1-1/\nu}.
\]
Take $B=\{\sum_iA_i>t\}$. If $p_B>0$, then
$tp_B\le C_\nu\ell p_B^{1-1/\nu}$, proving the weight $\ell^\nu$.
For $\nu=1$, set $b=t/(2\ell)$. On
$\{\max_iA_i\le t,\ \sum_iA_i>t\}$ the sum of the terms with $A_i>b$
exceeds $t/2$. Moreover,
\[
 \E[A_i\ind_{\{b<A_i\le t\}}]
 \le b\mathbb P(A_i>b)+\int_b^t\mathbb P(A_i>x)\dd x
 \le C\{1+\log(2\ell)\}.
\]
Markov's inequality and the maximum bound prove the logarithmic case.
\end{proof}

\begin{lemma}\label{lem:refined-convolution}
Let $f$ be a radially nonincreasing probability density on $\R^d$ satisfying
$f(x)\le B(1+\|x\|)^{-d-2\nu}$, with $\nu>0$. Then
\begin{equation}\label{eq:refined-convolution}
 f^{*\ell}(x)\le C w_\nu(\ell)(1+\|x\|)^{-d-2\nu},
 \qquad \ell\ge1,
\end{equation}
where $C$ does not depend on $\ell$ or $x$.
\end{lemma}
\begin{proof}
The argument has three steps: control the sum of squared frequencies,
control the vector sum by symmetry, and use radial monotonicity to pass
from a probability tail to a pointwise density bound. Constants below may
depend on $d,\nu$ and the envelope constant $B$, but not on $\ell$ or $x$.

\emph{Step 1: squared frequencies.}
Let $U_1,\ldots,U_\ell$ be independent random vectors with density $f$, and
put $S_\ell=\sum_iU_i$ and $A_\ell=\sum_i\|U_i\|^2$. Integrating the
upper envelope of $f$ outside the ball of radius $\sqrt t$ gives
$\PP(\|U_i\|^2>t)\le Ct^{-\nu}$ for $t\ge1$; increasing $C$ makes the
same bound valid for $0<t<1$. Lemma~\ref{lem:positive-sum} therefore yields
\[
 \PP(A_\ell>t)\le C_\nu w_\nu(\ell)t^{-\nu},\qquad t>0.
\]

\emph{Step 2: the vector sum.}
Let $\varepsilon_1,\ldots,\varepsilon_\ell$ be independent signs with
$\PP(\varepsilon_i=1)=\PP(\varepsilon_i=-1)=1/2$, independent of all $U_i$.
Symmetry of $f$ gives $S_\ell\stackrel d=\sum_i\varepsilon_iU_i$.
For the $j$th coordinate, conditioning on the $U_i$ and using
$\cosh y\le e^{y^2/2}$ gives, for every real $t$,
\[
 \E_\varepsilon\exp\left(t\sum_i\varepsilon_i U_{ij}\right)
 =\prod_i\cosh(tU_{ij})
 \le\exp\left(\frac{t^2}{2}\sum_iU_{ij}^2\right).
\]
The exponential Markov inequality, applied to both signs and optimized
over $t$, gives, for $a>0$,
\[
 \PP_\varepsilon\!\left(
    \left|\sum_i\varepsilon_iU_{ij}\right|>a
    \,\middle|\,U_1,\ldots,U_\ell\right)
 \le2\exp\!\left\{-\frac{a^2}{2\sum_iU_{ij}^2}\right\}.
\]
If the denominator is zero, the coordinate is identically zero and the
right-hand side is taken as zero. A vector with norm greater than $r$
has some coordinate of absolute value greater than $r/\sqrt d$.
Using $\sum_iU_{ij}^2\le A_\ell$, a union bound therefore gives
\[
 \PP_\varepsilon\left(\left\|\sum_i\varepsilon_iU_i\right\|>r
                  \,\middle|\,U_1,\ldots,U_\ell\right)
 \le2d\exp\{-r^2/(2dA_\ell)\}.
\]
The right-hand side is interpreted as zero when $A_\ell=0$ and $r>0$.
For $b>0$, the identity
$e^{-b/a}=\int_0^a bt^{-2}e^{-b/t}\dd t$, $a>0$, and Tonelli's theorem
show that
\[
 \E e^{-b/A_\ell}
 =\int_0^\infty\frac{b}{t^2}e^{-b/t}\PP(A_\ell>t)\dd t
 \le C_\nu w_\nu(\ell)b^{-\nu}\Gamma(\nu+1).
\]
Averaging the conditional tail bound, with $b=r^2/(2d)$, now gives
\begin{equation}\label{eq:symmetric-frequency-tail}
 \PP(\|S_\ell\|>r)\le Cw_\nu(\ell)r^{-2\nu}.
\end{equation}

\emph{Step 3: the density bound.}
Convolution preserves radial nonincrease. Indeed, at positive levels the
superlevel sets of a radially nonincreasing density are empty or centered
balls, up to null sets, so each density is an integral of their indicators. The convolution of
two such indicators is the overlap volume of the balls. With both centers
on a fixed axis, slicing along parallel lines gives intervals of fixed
lengths whose overlap is nonincreasing in the separation of their centers.
Integrating over these slices and then over the two level parameters,
by Tonelli's theorem, proves this monotonicity. Orthogonal invariance
also gives radial symmetry. Thus $f^{*\ell}$ is radially nonincreasing.

Let $v_d$ be the volume of the unit ball in $\R^d$ and fix $x$ with
$\|x\|=r>0$. For $r/2\le\|y\|\le r$ we have
$f^{*\ell}(y)\ge f^{*\ell}(x)$, and hence
\begin{align*}
 v_d(1-2^{-d})r^df^{*\ell}(x)
 &\le\int_{r/2\le\|y\|\le r}f^{*\ell}(y)\dd y\\
 &=\PP(r/2\le\|S_\ell\|\le r)\\
 &\le\PP(\|S_\ell\|\ge r/2)
 \le Cw_\nu(\ell)r^{-2\nu}.
\end{align*}
The last bound follows from~\eqref{eq:symmetric-frequency-tail}; the
boundary sphere has zero probability because $S_\ell$ has a density.
This proves~\eqref{eq:refined-convolution} for $r\ge1$. For $r<1$, use
$\|f^{*\ell}\|_\infty\le\|f\|_\infty\le B$ and $w_\nu(\ell)\ge1$,
enlarging the constant to replace $r^{-d-2\nu}$ by
$(1+r)^{-d-2\nu}$ over the two regions.
\end{proof}

\subsection{Tail coefficients and the observed spectral density}
\label{app:observed-tail-proof}

\begin{lemma}\label{lem:conv-tail}
Let $f,g$ be probability densities with polynomial upper envelopes of order
$p>d$, and suppose that
\[
 f(x)\sim c_f\norm x^{-p},\qquad
 g(x)\sim c_g\norm x^{-p},\qquad c_f,c_g>0,
\qquad \norm x\to\infty.
\]
Then $(f*g)(x)\sim(c_f+c_g)\norm x^{-p}$.
Consequently, under Assumption~\ref{ass:spectra}, with $p=d+2\nu$ and
$c$ as in~\eqref{eq:latent-tail},
\begin{equation}\label{eq:multi-conv-tail}
 f^{*\ell}(x)\sim \ell c\norm x^{-p}\quad(\norm x\to\infty)
\end{equation}
for every fixed $\ell\ge1$.
\end{lemma}

\begin{proof}
Fix $A>0$ and let $r=\norm x>2A$. In the convolution integral
$\int f(u)g(x-u)\dd u$, the sets $\{\norm u\le A\}$ and
$\{\norm{x-u}\le A\}$ are disjoint. Multiplication by $r^p$ and the tail
limits give, on these two sets, respectively,
\[
 c_g\int_{\norm u\le A}f(u)\dd u,
 \qquad c_f\int_{\norm v\le A}g(v)\dd v,
\]
as $r\to\infty$. Convergence is uniform over bounded translations because
the tail limits hold over all directions and $\norm{x-u}/r\to1$ uniformly
for $\norm u\le A$.

In the remaining region both $\norm u>A$ and $\norm{x-u}>A$. At least
one of $\norm u$ and $\norm{x-u}$ is at least $r/2$. The polynomial
upper envelopes therefore bound this remaining integral by
\[
 Cr^{-p}\left\{\int_{\norm u>A}f(u)\dd u+
                  \int_{\norm v>A}g(v)\dd v\right\}.
\]
Take $r\to\infty$ and then $A\to\infty$. The remaining normalized
integral vanishes and the two main masses tend to one. This proves the
binary assertion. Induction, using Lemma~\ref{lem:conv-bound} to retain
an upper envelope for every fixed convolution order, proves
\eqref{eq:multi-conv-tail}.
\end{proof}

\begin{proof}[Proof of Proposition~\ref{prop:tail}]
Under Assumption~\ref{ass:spectra}, $f$ is bounded, continuous, and strictly
positive. By~\eqref{eq:observed-spectral-density} and the bound
$\|f^{*\ell}\|_\infty\le\|f\|_\infty$ in Lemma~\ref{lem:conv-bound},
the series defining $f_Y$ converges uniformly. Each convolution is
continuous and strictly positive, and at least one $a_\ell$ is positive.
Hence $f_Y$ is continuous and strictly positive.

Put $p=d+2\nu$.
For $r=\norm\lambda\ge1$, Lemma~\ref{lem:refined-convolution} gives
\[
 r^p f^{*\ell}(\lambda)\le Cw_\nu(\ell).
\]
Assumption~\ref{ass:Hermite} makes this summable with weights $a_\ell$.
Lemma~\ref{lem:conv-tail} and dominated convergence imply
\[
 r^p f_Y(\lambda)\longrightarrow c\sum_{\ell\ge1}\ell a_\ell
                             =cK_F'(1^-)=c\mathcal D(F)=c_Y.
\]
The tail limit supplies the two-sided envelope at large frequencies;
continuity and positivity supply it on the remaining compact ball.
\end{proof}

\section{Spectral comparison of kriging MSPEs under dense sampling}
\label{app:comparison}

This appendix supplies the second-order prediction ingredient used after the
spectral tail has been identified. We follow the residual of a single kriging
predictor and show that its normalized spectral energy leaves every bounded
frequency set under dense sampling. A ratio of spectral densities with a limit at
infinity can then be averaged against this energy. The argument uses only
second-order structure. It is related to the spectral approach to asymptotically
optimal prediction developed by \citet{Stein1993}, but the comparison result
needed here is established directly below; see also \citet{KirchnerBolin2022}.

\subsection{The normalized kriging residual}
\label{app:kriging-residual}

Write $e_s(\lambda)=e^{\mathrm i\lambda^{\mathsf T}s}$. For a continuous,
even, strictly positive integrable spectral density $f$, $L^2(f)$ is the
complex Hilbert space with inner product
$\langle g,h\rangle_{L^2(f)}=\int_{\R^d}g(\lambda)\overline{h(\lambda)}
f(\lambda)\dd\lambda$. Define
\begin{equation}\label{eq:Qf}
 Q_{f,n}(w)=\int_{\R^d}
 \left|e_{s_\star}(\lambda)-\sum_{j=1}^nw_je_{s_j}(\lambda)\right|^2
 f(\lambda)\dd\lambda,
 \qquad
 V_{f,n}=\min_{w\in\R^n}Q_{f,n}(w).
\end{equation}
Let $w_{f,n}$ denote the unique minimizer. For spectral densities $f_i$, we abbreviate
$Q_{f_i,n}$, $V_{f_i,n}$ and $w_{f_i,n}$ to $Q_{i,n}$, $V_{i,n}$ and
$w_{i,n}$. In particular, $Q_{0,n}(w_{1,n})$ is the true MSPE under the model with spectral density
$f_0$ of the predictor constructed using $f_1$.

\begin{lemma}
\label{lem:residual-energy}
Suppose $f$ has the properties above and, for some $b>0$ and $p>d$,
\[
 f(\lambda)\ge b(1+\|\lambda\|)^{-p},\qquad\lambda\in\R^d.
\]
For the designs in Assumption~\ref{ass:design}, put
\[
 u_{f,n}(\lambda)=
 \frac{e_{s_\star}(\lambda)-\sum_{j=1}^nw_{f,n,j}e_{s_j}(\lambda)}
      {\sqrt{V_{f,n}}},
 \qquad
 \Pi_{f,n}(\dd\lambda)=|u_{f,n}(\lambda)|^2f(\lambda)\dd\lambda.
\]
Then $V_{f,n}>0$, $V_{f,n}\to0$, and $\Pi_{f,n}$ is a probability measure
satisfying
\begin{equation}\label{eq:residual-local-energy}
 \Pi_{f,n}(B_R)\longrightarrow0,
 \qquad B_R=\{\lambda:\|\lambda\|\le R\},\quad R<\infty.
\end{equation}
Consequently, for every bounded Borel function $g$ such that
$g(\lambda)\to g_\infty$ as $\|\lambda\|\to\infty$,
\begin{equation}\label{eq:residual-averaging}
 \int_{\R^d}g(\lambda)\,\Pi_{f,n}(\dd\lambda)
 \longrightarrow g_\infty.
\end{equation}
\end{lemma}

\begin{proof}
The argument is the mean-square projection argument underlying simple
kriging, written in its spectral representation. Work in the complex
Hilbert space
\[
 \mathscr H_f=
 \overline{\operatorname{span}_{\mathbb C}\{e_s:s\in D\}}^{\,L^2(f)},
 \qquad
 \mathscr H_n=\operatorname{span}_{\mathbb C}
                  \{e_{s_1},\ldots,e_{s_n}\}.
\]
Evenness of $f$ makes the covariance matrices and target vectors real, so
the orthogonal projection coefficients are the real weights $w_{f,n}$.
Strict positivity of $f$ and distinctness of the sites, including the
target, give $V_{f,n}>0$.

Let $C_f(h)=\int e^{\mathrm i\lambda^{\mathsf T}h}f(\lambda)\dd\lambda$.
Since $f$ is integrable,
\[
 \|e_s-e_t\|_{L^2(f)}^2
 =2\{C_f(0)-C_f(s-t)\}\longrightarrow0\qquad(t\to s).
\]
Density of the sites and nesting therefore imply
\begin{equation}\label{eq:subspace-approximation}
 \operatorname{dist}_{L^2(f)}(h,\mathscr H_n)\longrightarrow0,
 \qquad h\in\mathscr H_f.
\end{equation}
Taking $h=e_{s_\star}$ gives $V_{f,n}\to0$. Also,
$\|u_{f,n}\|_{L^2(f)}=1$ and $u_{f,n}\perp\mathscr H_n$, whence
\[
 |\langle u_{f,n},h\rangle_{L^2(f)}|
 \le\operatorname{dist}_{L^2(f)}(h,\mathscr H_n)\longrightarrow0.
\]
Thus $u_{f,n}\rightharpoonup0$ in $\mathscr H_f$.

Choose a real, even, smooth function $\chi$ with compact support, equal
to one on a neighbourhood of $D\cup(-D)$; compactness of $D$ permits
this choice. We use the Fourier pair
\[
 \widehat\chi(\lambda)=(2\pi)^{-d}\int_{\R^d}
       e^{-\mathrm i\lambda^{\mathsf T}s}\chi(s)\dd s,
 \qquad
 \chi(s)=\int_{\R^d}e^{\mathrm i\lambda^{\mathsf T}s}
                  \widehat\chi(\lambda)\dd\lambda.
\]
Every finite linear combination $u$ of the $e_s$, $s\in D$, satisfies
\begin{equation}\label{eq:reproduction}
 u(\lambda)=\int_{\R^d}\widehat\chi(\xi-\lambda)u(\xi)\dd\xi.
\end{equation}
Indeed, the identity holds term by term because the integral on the right for
$u=e_s$ equals $e_s(\lambda)\chi(s)=e_s(\lambda)$. To control evaluation
uniformly for $\lambda\in B_R$, choose an integer $M$ with $2M>p+d$.
Smoothness and compact support give
$|\widehat\chi(\eta)|\le C_M(1+\|\eta\|)^{-M}$. The polynomial lower
bound on $f$ and $1+\|\xi\|\le(1+R)(1+\|\xi-\lambda\|)$ then imply
\[
 \frac{|\widehat\chi(\xi-\lambda)|^2}{f(\xi)}
 \le C_{M,R}(1+\|\xi-\lambda\|)^{p-2M},
 \qquad\lambda\in B_R.
\]
The right-hand side is integrable in $\xi$, uniformly in $\lambda$, so
\begin{equation}\label{eq:residual-evaluation-bound}
 C_R:=\sup_{\lambda\in B_R}
 \int_{\R^d}\frac{|\widehat\chi(\xi-\lambda)|^2}{f(\xi)}\dd\xi
 <\infty.
\end{equation}
Weak convergence in $L^2(f)$ alone would not imply pointwise
convergence. The cutoff supplies the missing continuity: for fixed
$\lambda$, define
\[
 L_\lambda(h)=\int_{\R^d}\widehat\chi(\xi-\lambda)h(\xi)\dd\xi,
 \qquad h\in L^2(f).
\]
Multiplying and dividing the integrand by $\sqrt{f(\xi)}$, then using
Cauchy--Schwarz and~\eqref{eq:residual-evaluation-bound}, gives
\[
 |L_\lambda(h)|
 \le\|h\|_{L^2(f)}
      \left(\int_{\R^d}
       \frac{|\widehat\chi(\xi-\lambda)|^2}{f(\xi)}\dd\xi\right)^{1/2}
 \le C_R^{1/2}\|h\|_{L^2(f)},\qquad\lambda\in B_R.
\]
Thus $L_\lambda$ restricts to a continuous linear functional on
$\mathscr H_f$. The residual $u_{f,n}$ is a finite linear combination of
the $e_s$, so~\eqref{eq:reproduction} gives
$u_{f,n}(\lambda)=L_\lambda(u_{f,n})\to0$. The same bound and
$\|u_{f,n}\|_{L^2(f)}=1$ imply $|u_{f,n}(\lambda)|^2\le C_R$ on $B_R$.
Since $f$ is integrable on $B_R$, dominated convergence gives
\[
 \Pi_{f,n}(B_R)=\int_{B_R}|u_{f,n}(\lambda)|^2f(\lambda)\dd\lambda
 \longrightarrow0.
\]
Finally, given $\varepsilon>0$, choose $R$ so that
$|g-g_\infty|\le\varepsilon$ outside $B_R$. Since $\Pi_{f,n}$ has total
mass one,
\[
 \left|\int g\,\dd\Pi_{f,n}-g_\infty\right|
 \le \|g-g_\infty\|_\infty\Pi_{f,n}(B_R)+\varepsilon.
\]
Letting $n\to\infty$ and then $\varepsilon\downarrow0$ proves
\eqref{eq:residual-averaging}.
\end{proof}

\subsection{Comparison under covariance misspecification}
\label{app:kriging-comparison}

\begin{lemma}\label{lem:spectral-comparison}
Let $f_0,f_1$ be continuous, even, strictly positive integrable functions.
Suppose that for some $p>d$ and $0<b\le B<\infty$,
\begin{equation}\label{eq:comparison-envelope}
 b(1+\|\lambda\|)^{-p}\le f_0(\lambda)
       \le B(1+\|\lambda\|)^{-p},
\end{equation}
and that $f_1/f_0\to a\in(0,\infty)$ as $\|\lambda\|\to\infty$.
For the designs in Assumption~\ref{ass:design},
\begin{equation}\label{eq:comparison-conclusions}
 \frac{Q_{0,n}(w_{1,n})}{V_{0,n}}\longrightarrow1,
 \qquad
 \frac{V_{1,n}}{V_{0,n}}\longrightarrow a,
 \qquad
 \frac{V_{1,n}}{Q_{0,n}(w_{1,n})}\longrightarrow a.
\end{equation}
\end{lemma}

\begin{proof}
Continuity, positivity and the limiting ratio of the spectral densities give
\begin{equation}\label{eq:bounded-ratio}
 0<m\le f_1/f_0\le M<\infty.
\end{equation}
Thus both spectral densities satisfy the assumptions of
Lemma~\ref{lem:residual-energy}. Applying~\eqref{eq:residual-averaging}
with $f=f_0$ and $g=f_1/f_0$, and then interchanging the spectral densities, gives
\[
 A_n:=\frac{Q_{1,n}(w_{0,n})}{V_{0,n}}\longrightarrow a,
 \qquad
 B_n:=\frac{Q_{0,n}(w_{1,n})}{V_{1,n}}\longrightarrow a^{-1}.
\]
The minimizing property of the two predictors implies
\[
 V_{1,n}\le Q_{1,n}(w_{0,n}),
 \qquad
 V_{0,n}\le Q_{0,n}(w_{1,n}),
\]
and therefore
\[
 B_n^{-1}\le\frac{V_{1,n}}{V_{0,n}}\le A_n.
\]
Hence $V_{1,n}/V_{0,n}\to a$. The remaining conclusions follow from
\[
 \frac{Q_{0,n}(w_{1,n})}{V_{0,n}}
 =B_n\frac{V_{1,n}}{V_{0,n}}\longrightarrow1,
 \qquad
 \frac{V_{1,n}}{Q_{0,n}(w_{1,n})}=B_n^{-1}\longrightarrow a.
\]
\end{proof}

The residual argument uses only a limiting ratio of spectral densities and
a polynomial lower bound. No relative spectral remainder rate or restriction
$d\le3$ is required, and equivalence of Gaussian measures is not assumed.

\begin{remark}\label{rem:nonnested-designs}
Nesting is used only to establish~\eqref{eq:subspace-approximation}.
The residual argument and all the MSPE comparisons remain valid for
possibly nonnested finite designs excluding $s_\star$ whenever their
prediction spaces satisfy that approximation property. A sufficient
geometric condition is
\[
 \min_{t\in S_n}\|s-t\|\longrightarrow0\qquad\text{for every }s\in D.
\]
Covariance continuity proves approximation first for each $e_s$, then for
finite linear combinations and finally for their closure. Equivalent
spectral norms, as in~\eqref{eq:bounded-ratio}, preserve this property.
In particular, approximation of the target gives $V_{f,n}\to0$.
\end{remark}

\section{Proofs of the main prediction results}
\label{app:main-proofs}

The proofs follow the order of the statements in Section~\ref{sec:main}:
Theorem~\ref{thm:main}, Corollaries~\ref{cor:latent-Matern} and~\ref{cor:rough},
Proposition~\ref{prop:exponential-sandwich}, and Theorem~\ref{thm:general}. The ingredients are the exact oracle MSPE in
Appendix~\ref{app:Hermite}, the spectral-density bounds in
Appendix~\ref{app:convolutions}, and the kriging comparison in
Appendix~\ref{app:comparison}.

\subsection{Asymptotic MSPE equivalence}
\label{app:asymptotic-proofs}

\begin{proof}[Proof of Theorem~\ref{thm:main}]
The common latent prediction variance satisfies $v_n>0$ and $v_n\to0$.
Proposition~\ref{prop:exact} gives
\begin{equation}\label{eq:oracle-scaled-limit}
 \frac{\MSPEO}{v_n}\longrightarrow\mathcal D(F)\in(0,\infty).
\end{equation}
It remains to obtain the same limit for the largest MSPE, $\MSPER$.
Under the latent spectral density $f$, the minimizing weights are
$w_{f,n}=\boldsymbol w_{\rho,n}$ and $V_{f,n}=v_n$. Consequently,
\begin{equation}\label{eq:latent-residual-risk}
 \frac{\MSPER}{v_n}
 =\int_{\R^d}\frac{f_Y(\lambda)}{f(\lambda)}
                  \,\Pi_{f,n}(\dd\lambda)
 \longrightarrow\mathcal D(F).
\end{equation}
Here Proposition~\ref{prop:tail} makes $f_Y/f$ bounded with limit
$\mathcal D(F)$, and Lemma~\ref{lem:residual-energy} gives the convergence.

Nonconstancy of $F$ and~\eqref{eq:oracle-Hermite} imply $\MSPEO>0$.
Now divide the hierarchy
\[
 0<\MSPEO\le\MSPEN\le\MSPEL\le\MSPER
\]
by $v_n$. Its lower and upper endpoints both converge to
$\mathcal D(F)$ by~\eqref{eq:oracle-scaled-limit} and
\eqref{eq:latent-residual-risk}; the squeeze theorem therefore gives the same
limit for the two intermediate MSPEs. This proves~\eqref{eq:main-risk-scale}
and~\eqref{eq:main-ratio}. Proposition~\ref{prop:finite-gains} then gives
\eqref{eq:vanishing-gains}. Finally, since
$\mathcal G_n\subseteq\mathcal H_n$, conditional-expectation
orthogonality yields
\[
 \MSPEN-\MSPEO=\E[(\YN-\YO)^2].
\]
Dividing by $\MSPEO$ proves~\eqref{eq:oracle-information-gap}.
\end{proof}

\subsection{Mat\'ern specialization}
\label{app:matern-corollary}

\begin{proof}[Proof of Corollary~\ref{cor:latent-Matern}]
The spectral density in~\eqref{eq:matern-spectral-density} is continuous,
even, strictly positive, and radially decreasing. Moreover,
\[
 f(\lambda)
 \sim C_{\nu,d}\alpha^{-2\nu}\|\lambda\|^{-d-2\nu},
 \qquad\|\lambda\|\to\infty.
\]
Thus Assumption~\ref{ass:spectra} and~\eqref{eq:radial-monotone} hold,
with $c=C_{\nu,d}\alpha^{-2\nu}>0$. The remaining assumptions are those
of the corollary, so Theorem~\ref{thm:main} applies.
\end{proof}

\subsection{The rough-latent regime}
\label{app:rough-corollary}

\begin{proof}[Proof of Corollary~\ref{cor:rough}]
Together with~\eqref{eq:L2-F}, the weak-gradient condition implies
$F\in W^{1,2}(\gq)$. Lemma~\ref{lem:derivatives} gives
\[
 \sum_{\ell\ge1}\ell a_\ell=\E\|\nabla F(\bG)\|^2<\infty.
\]
For $0<\nu<1$, $w_\nu(\ell)=\ell$, so this is precisely
Assumption~\ref{ass:Hermite}. All other hypotheses of
Theorem~\ref{thm:main} are assumed, and its conclusions follow.
\end{proof}

\subsection{The exponential finite-sample bound}
\label{app:exponential}

\begin{proof}[Proof of Proposition~\ref{prop:exponential-sandwich}]
The exponential spectral density, under~\eqref{eq:spectral-convention}, is
\begin{equation}\label{eq:exponential-density}
 f_\alpha(\lambda)=K_d\alpha^{-1}
       (\alpha^{-2}+\|\lambda\|^2)^{-(d+1)/2},
 \qquad K_d=\frac{\Gamma\{(d+1)/2\}}{\pi^{(d+1)/2}}.
\end{equation}
Since
\begin{equation}\label{eq:exponential-product}
 \rho(h)^\ell=\exp(-\ell\|h\|/\alpha),
 \qquad f_\alpha^{*\ell}=f_{\alpha/\ell},
\end{equation}
we have, at every frequency,
\begin{equation}\label{eq:exponential-domination}
 \frac{f_{\alpha/\ell}(\lambda)}{f_\alpha(\lambda)}
 =\ell\left(
 \frac{\alpha^{-2}+\|\lambda\|^2}
      {\ell^2\alpha^{-2}+\|\lambda\|^2}
 \right)^{(d+1)/2}\le\ell.
\end{equation}
By~\eqref{eq:observed-spectral-density}, the nonnegative covariance-map
coefficients therefore give
\begin{equation}\label{eq:exponential-global-domination}
 0<f_Y(\lambda)
 =\sum_{\ell\ge1}a_\ell f_{\alpha/\ell}(\lambda)
 \le\left(\sum_{\ell\ge1}\ell a_\ell\right)f_\alpha(\lambda)
 =\mathcal D(F)f_\alpha(\lambda).
\end{equation}
Strict positivity of $f_Y$ also gives nonsingular observed covariance
matrices on distinct sites.

Integrating~\eqref{eq:exponential-global-domination} against the squared
Fourier transform of any finite coefficient vector gives the corresponding
covariance quadratic-form ordering. Thus the variance of any centered linear
combination of $Y$ values is at most $\mathcal D(F)$ times the variance of the same
linear combination of $Z_1$ values. Apply this to the prediction-error coefficients
$(1,-\boldsymbol w_{\rho,n}^{\mathsf T})$. Since the predictors use the
correct mean,
\[
 \MSPER
 \le \mathcal D(F)\Var\{Z_1(s_\star)
       -\boldsymbol w_{\rho,n}^{\mathsf T}Z_1(S_n)\}
 =\mathcal D(F)v_n.
\]
Proposition~\ref{prop:exact}, the MSPE hierarchy, and positivity of
\eqref{eq:oracle-Hermite} prove
\eqref{eq:exponential-finite-sandwich} and
\eqref{eq:exponential-finite-ratios}. If $v_n\to0$, then
$\Delta_F(v_n)/(\mathcal D(F)v_n)\to1$, so the sandwich gives the same asymptotic
MSPE scale. The projection identities give the remaining conclusions of
Theorem~\ref{thm:main}. Comparison with the best predictor based only on the nearest latent
observation gives
\[
 v_n\le1-\exp\{-2\operatorname{dist}(s_\star,S_n)/\alpha\},
\]
which proves the stated sufficient design condition.

For the final bound, suppose $K_F''(1^-)<\infty$. The nonnegative
covariance-map coefficients imply
$0\le K_F''(u)\le K_F''(1^-)$ for $0\le u\le1$. The integral remainder
formula gives
\[
 0\le \mathcal D(F)v-\Delta_F(v)
   =\int_0^v(v-t)K_F''(1-t)\dd t
   \le\tfrac12 K_F''(1^-)v^2.
\]
Combining this with~\eqref{eq:exponential-finite-sandwich} proves
\eqref{eq:exponential-finite-gap}.
\end{proof}

For completeness, each ratio in~\eqref{eq:exponential-domination} tends
to $\ell$ at infinity. Dominated convergence also gives
\begin{equation}\label{eq:exponential-tail-ratio}
 \frac{f_Y(\lambda)}{f_\alpha(\lambda)}
 =\sum_{\ell\ge1}a_\ell
      \frac{f_{\alpha/\ell}(\lambda)}{f_\alpha(\lambda)}
 \longrightarrow \mathcal D(F).
\end{equation}
This verifies~\eqref{eq:exponential-cY}, but is not needed for the
finite-sample sandwich or its asymptotic consequence.

\subsection{Removing radial monotonicity}
\label{app:general-proof}

\begin{proof}[Proof of Theorem~\ref{thm:general}]
Condition~\eqref{eq:conservative-Hermite} implies
$0<\mathcal D(F)<\infty$. Repeat the proof of
Proposition~\ref{prop:tail} using Lemma~\ref{lem:conv-bound} in place of
Lemma~\ref{lem:refined-convolution}. With $p=d+2\nu$, the bound
\[
 \|\lambda\|^p f^{*\ell}(\lambda)\le C\ell^{p+1},
 \qquad \|\lambda\|\ge1,
\]
and~\eqref{eq:conservative-Hermite} justify dominated convergence in
the spectral series. Thus $f_Y/f\to\mathcal D(F)$, and $f_Y$ is continuous
and strictly positive. The ratio is bounded by positivity and continuity
on compact sets and its limit at infinity. The single-residual comparison
\eqref{eq:latent-residual-risk} and the same oracle sandwich now prove
all the conclusions, without radial monotonicity.
\end{proof}

\section{Regularity checks for the models in the main text}
\label{app:examples}

The calculations follow the representative constructions in
Section~\ref{sec:examples}. We also give the definitions and calculations
for the spatial-scale modulation, product, and Gaussian-square fields,
which appear only in Table~\ref{tab:models} of the main text.
The Gaussian weak-derivative criterion is Lemma~\ref{lem:derivatives}.
Additional ratio-field constructions are given in
Appendix~\ref{app:additional-models}.

\subsection{Log-Gaussian fields}
\label{app:lognormal-regularity}

For the log-Gaussian transformation, every derivative is a constant times
$e^{\mu+\omega z}$, hence square integrable under Gaussian measure. Its
first derivative yields~\eqref{eq:lognormal-example}. The Gaussian
moment-generating function gives
$\E e^{\mu+\omega G}=e^{\mu+\omega^2/2}$ and
$\E e^{2\mu+\omega(G+G_u)}=e^{2\mu+\omega^2(1+u)}$.
Subtracting the squared mean gives the covariance map
\begin{equation}\label{eq:KF-lognormal}
K_F(u)=e^{2\mu+\omega^2}\bigl(e^{\omega^2u}-1\bigr),
\qquad K_F'(1^-)=\omega^2e^{2\mu+2\omega^2}.
\end{equation}

For completeness, the exact log-Gaussian MSPE used in
Section~\ref{sec:numerical} follows directly from scalar Gaussian
regression. Let $M_n=\E\{Z(s_\star)\mid\mathcal H_n\}$. Then
\begin{equation}\label{eq:latent-decomposition}
 Z(s_\star)=M_n+\sqrt{v_n}U_n,
 \qquad M_n\sim N(0,1-v_n),\quad U_n\sim N(0,1),
\end{equation}
with independent $M_n,U_n$. Invertibility gives
$\YN=\exp(\mu+\omega M_n+\omega^2v_n/2)$, and the law of total variance
therefore yields
\begin{equation}\label{eq:lognormal-exact-risk}
 \MSPEN=e^{2\mu+2\omega^2}(1-e^{-\omega^2v_n})
       \sim\omega^2e^{2\mu+2\omega^2}v_n.
\end{equation}

\subsection{Sinh--arcsinh, Tukey, and Gaussian-copula transformations}
\label{app:additional-scalars}

The sinh--arcsinh distribution was introduced by
\citet{JonesPewsey2009}; the corresponding spatial random-field
construction was studied by \citet{BlasiEtAl2022}. In the present notation,
\[
 Y(s)=\mu+\sigma\sinh\{(\asinh Z(s)+\varepsilon)/\kappa\},
 \qquad\mu,\varepsilon\in\R,\quad\sigma,\kappa>0.
\]
The associated transformation satisfies the general regularity condition
for every finite latent smoothness $\nu$.

The Tukey $g$-and-$h$ random field introduced and studied by
\citet{XuGenton2017} is
\begin{equation}\label{eq:tgh-example}
 Y(s)=\xi+\omega\tau_{g,h}\{Z(s)\},\qquad
 \tau_{g,h}(z)=
 \begin{cases}
  g^{-1}(e^{gz}-1)e^{hz^2/2},&g\ne0,\\
  z e^{hz^2/2},&g=0,
 \end{cases}
 .
\end{equation}
Here $\xi,g\in\R$, $\omega>0$, and $0\le h<1/2$.
All derivatives of the sinh--arcsinh and Tukey transformations have finite
Gaussian second moments under the stated parameter restrictions. For the
Tukey field, $h<1/2$ is essential to the present second-order formulation;
it is not imposed merely for invertibility. \citet{XuGenton2017} also
develop conditional prediction for their model. Here we apply the fixed-domain
comparison between linear and optimal nonlinear MSPEs.

\paragraph{Regularity check.}

For a fixed sinh--arcsinh transformation, it and each scalar derivative
have at most polynomial growth as $|z|\to\infty$. This follows by
repeated differentiation of $\asinh z$, whose derivatives are rational
combinations of $z$ and $(1+z^2)^{-1/2}$, and by
$e^{|\asinh z|/\kappa}\le C(1+|z|)^{1/\kappa}$.
All derivatives are bounded on compact sets and have finite Gaussian second
moments. For Tukey transformations with fixed $g$ and $0\le h<1/2$,
every derivative is bounded by
\[
 C(1+|z|)^M\exp\{hz^2/2+|g||z|\}
\]
for suitable constants depending on the derivative order. The square of
this bound is Gaussian-integrable because $h<1/2$. This includes the continuous definition at $g=0$. Since derivatives of
every fixed order are square integrable, Lemma~\ref{lem:derivatives}
proves both the refined and the conservative conditions in these cases.

Gaussian-copula random fields provide another scalar construction. Let $Q$
be a strictly increasing marginal quantile function and define
\[
 Y(s)=Q\{\Phi(Z(s))\},
 \qquad
 F(z)=Q\{\Phi(z)\}.
\]
If $Q$ is locally absolutely continuous, then
\[
 F'(z)=Q'\{\Phi(z)\}\phi(z)
\]
for almost every $z$. Since $\Phi(G)$ is uniform on $(0,1)$,
\[
 \E\{F(G)^2\}=\int_0^1 Q(u)^2\dd u,
\]
and a change of variable $u=\Phi(z)$ gives
\begin{equation}\label{eq:gaussian-copula-energy}
 \mathcal D(F)
 =\int_0^1 [Q'(u)]^2
   \phi\{\Phi^{-1}(u)\}^2\dd u.
\end{equation}
Thus
\[
 \int_0^1 Q(u)^2\dd u<\infty,
 \qquad
 \int_0^1 [Q'(u)]^2
   \phi\{\Phi^{-1}(u)\}^2\dd u<\infty
\]
is a direct sufficient criterion for $F\in W^{1,2}(\gamma_1)$. Indeed,
$\Phi$ maps each compact interval into a compact subinterval of $(0,1)$,
where it is a smooth diffeomorphism. Thus $Q\circ\Phi$ is locally
absolutely continuous and its weak derivative is the displayed chain-rule
expression. The two integrals are exactly its squared Gaussian function
and derivative norms. Consequently, Corollary~\ref{cor:rough} applies
for $0<\nu<1$ under the stated latent spectral assumptions. For smoother
latent fields, applicability of Theorem~\ref{thm:main} depends on the
corresponding higher-order Gaussian-Sobolev regularity of $Q\circ\Phi$.

\subsection{Skew-Gaussian fields and fractional regularity}
\label{app:fractional-examples}

The weak gradient of $\mu+\delta|z_1|+\tau z_2$ is
$(\delta\operatorname{sign}z_1,\tau)$ almost everywhere. Its components
are bounded, giving $\mathcal D(F)=\delta^2+\tau^2$.
The covariance map is
\begin{equation}\label{eq:skew-covariance}
K_F(u)=\delta^2 A(u)+\tau^2u,\qquad
A(u)=\frac2\pi\left\{\sqrt{1-u^2}+u\arcsin u-1\right\}.
\end{equation}
This follows from the bivariate Gaussian integral
\[
 \E(|G||G'|)=\frac2\pi\{\sqrt{1-u^2}+u\arcsin u\},
 \qquad\Cov(G,G')=u,
\]
and subtraction of $(\E|G|)^2=2/\pi$. The integral can also be found
by differentiation in $u$, using
$\E(\operatorname{sign}G\operatorname{sign}G')=2\arcsin(u)/\pi$,
with initial value $2/\pi$ at $u=0$.

\paragraph{Fractional regularity of the absolute value.}

For a standard normal variable $G$, the squared orthonormal Hermite
coefficient of $|G|$ at degree $2k$, $k\ge1$, is
\begin{equation}\label{eq:absolute-Hermite}
 b_{2k}^2=\frac2\pi\frac{\{(2k-3)!!\}^2}{(2k)!}
          \asymp k^{-5/2},
\end{equation}
where $(-1)!!=1$. To obtain this identity, integrate
$|x|H_{2k}(x)$ against the standard Gaussian density twice by parts:
the distributional second derivative of $|x|$ is $2\delta_0$, where
$\delta_0$ is the unit point mass at the origin, and
$H_{2k-2}(0)=(-1)^{k-1}(2k-3)!!$. The odd coefficients vanish.
Stirling's formula gives the displayed order. Consequently,
$|z|\in W^{s,2}(\gamma_1)$ exactly when $s<3/2$; the logarithm in
\eqref{eq:Hermite-condition} at $\nu=1$ is integrable as well. The
independent linear term in the skew-Gaussian example adds only a
first-degree component. This proves its refined inclusion for $\nu<3/2$.

\subsection{Weibull fields}
\label{app:weibull-regularity}

For the Weibull transformation, write $S=z_1^2+z_2^2$ at
$\bz=(z_1,z_2)$. At a standard Gaussian vector, $S/2$ is exponential
with mean one. Away from the origin,
\[
 \partial_jF(\bz)=\frac{\tau}{k}
          (S/2)^{1/k-1}z_j,\qquad j=1,2.
\]
Thus
\[
 \sum_{j=1}^2(\partial_jF)^2
      =\frac{2\tau^2}{k^2}(S/2)^{2/k-1}.
\]
Its expectation is $2\tau^2\Gamma(2/k)/k^2$, finite for every $k>0$.
This proves the total energy in~\eqref{eq:Weibull-energy}. These are weak derivatives across the
origin. For instance, cut off the function inside a ball of radius
$\epsilon$ and smooth the transition; the squared gradient error near
zero is of order $\epsilon^{4/k}$ and tends to zero. Gaussian tails
control the cutoff at infinity. The function itself has finite moments
of every positive order. Lemma~\ref{lem:derivatives} now verifies
$W^{1,2}(\gamma_2)$.

\subsection{Clayton-like fields and marginal quantiles}
\label{app:clayton}

\begin{proof}[Proof of Corollary~\ref{cor:clayton}]
Write a standard Gaussian vector in $\R^{q_0+2}$ as $(X,W)$, with
$X\in\R^{q_0}$, $W\in\R^2$, and independent blocks. Put
\[
 S=\norm X^2+\norm W^2,\qquad T=\frac{\norm X^2}{S}.
\]
The change of variables from the independent chi-squared block norms to
their sum and ratio gives
\begin{equation}\label{eq:clayton-radial-beta}
 S\sim\chi^2_{q_0+2},\qquad
 T\sim\operatorname{Beta}({q_0}/2,1),\qquad S\ \text{and}\ T\ \text{independent}.
\end{equation}
In particular, $\E(S^{-1})=1/q_0$ and $T^{q_0/2}$ is uniform on $(0,1)$.
Except on null sets, $F_{q_0}=T^{q_0/2}$ and direct differentiation gives
\[
 \norm{\nabla T}^2=\frac{4T(1-T)}S.
\]
Consequently,
\begin{equation}\label{eq:clayton-gradient-square}
 \norm{\nabla F_{q_0}}^2
       =\frac{q_0^2}{S}T^{q_0-1}(1-T).
\end{equation}
These are weak derivatives. Indeed,
$F_{q_0}(x,w)=(\norm x/\norm{(x,w)})^{q_0}$ is locally Lipschitz away from the
origin for every integer $q_0\ge1$, including $q_0=1$, and is bounded by one.
Its gradient norm is at most $q_0/\norm{(x,w)}$. Multiplying by a smooth
radial cutoff that vanishes on the ball of radius $\epsilon$ and equals
one outside the ball of radius $2\epsilon$ changes the squared gradient
norm by at most a constant times
\[
 \int_{\norm z<2\epsilon}\norm z^{-2}\dd z
 +\epsilon^{-2}\int_{\norm z<2\epsilon}\dd z
 =O(\epsilon^{q_0}).
\]
The Gaussian density is bounded on this ball. The function error also
vanishes in $L^2(\gamma_{q_0+2})$, and Gaussian tails permit cutoffs at
infinity. Closedness of the weak gradient therefore gives
$F_{q_0}\in W^{1,2}(\gamma_{q_0+2})$.

Taking expectations in~\eqref{eq:clayton-gradient-square} and using
\eqref{eq:clayton-radial-beta} gives
\begin{align}
 \mathcal D(F_{q_0})
 &=q_0\E[T^{q_0-1}(1-T)]\notag\\
 &=\frac{q_0^2}{2}\int_0^1 t^{3q_0/2-2}(1-t)\dd t
   =\frac{2q_0}{3(3q_0-2)}.\label{eq:clayton-energy-proof}
\end{align}
The integral is finite for $q_0\ge1$. This proves
\eqref{eq:clayton-energy}. Corollary~\ref{cor:rough} yields
\eqref{eq:clayton-uniform-risk}.

Next consider a strictly increasing, locally absolutely continuous
quantile $Q$. Since $F_{q_0}(\bG)$ is uniform, the square integrability of
$F_{q_0,Q}(\bG)$ is exactly the first condition in
\eqref{eq:clayton-quantile-conditions}. For a bounded Lipschitz function $Q$,
the weak chain rule and~\eqref{eq:clayton-gradient-square} imply
\begin{align}
 \E\norm{\nabla(Q\circ F_{q_0})}^2
 &=q_0\E\left([Q'(T^{q_0/2})]^2T^{q_0-1}(1-T)\right)\notag\\
 &=\frac{q_0^2}{2}\int_0^1
       [Q'(x^{q_0/2})]^2x^{3q_0/2-2}(1-x)\dd x\notag\\
 &=q_0\int_0^1[Q'(t)]^2t^{2-2/q_0}(1-t^{2/q_0})\dd t
   =I_{q_0}(Q).\label{eq:clayton-chain-energy}
\end{align}
The final change of variable is $t=x^{q_0/2}$.

For the stated, possibly unbounded $Q$, the same identity follows by
closedness, without assuming that $Q$ is globally Lipschitz. For $N\ge3$,
define
\[
 Q_N(t)=Q(1/2)+\int_{1/2}^t Q'(x)
       \ind_{\{1/N\le x\le1-1/N,\ Q'(x)\le N\}}\dd x.
\]
Here $Q'\ge0$ almost everywhere. Each $Q_N$ extends to a bounded
Lipschitz function on $[0,1]$, converges pointwise to $Q$ on $(0,1)$,
and satisfies
$|Q_N(t)|\le |Q(t)|+2|Q(1/2)|$.
Thus $Q_N\to Q$ in $L^2(0,1)$. The derivatives converge to $Q'$ almost
everywhere, are bounded by $Q'$, and converge in the weighted squared norm
in~\eqref{eq:clayton-quantile-energy} when $I_{q_0}(Q)<\infty$.
Applying~\eqref{eq:clayton-chain-energy} to differences of the $Q_N$ shows
that $Q_N\circ F_{q_0}$ is Cauchy in $W^{1,2}(\gamma_{q_0+2})$. Its $L^2$ limit
is $F_{q_0,Q}$, so
\[
 F_{q_0,Q}\in W^{1,2}(\gamma_{q_0+2}),\qquad
 \mathcal D(F_{q_0,Q})=I_{q_0}(Q).
\]
Strict increase makes the transformation nonconstant; hence this energy
is positive. Corollary~\ref{cor:rough} proves
\eqref{eq:clayton-quantile-risk}.

For Beta margins, write $H_{a,b}$ for the distribution function, with
density proportional to $x^{a-1}(1-x)^{b-1}$ on $(0,1)$. Integration at
the endpoints and differentiation of the inverse give
\[
 Q_{a,b}'(t)\asymp t^{1/a-1}\quad(t\downarrow0),\qquad
 Q_{a,b}'(t)\asymp(1-t)^{1/b-1}\quad(t\uparrow1).
\]
The integrand in~\eqref{eq:clayton-quantile-energy} is therefore of order
$t^{2/a-2/q_0}$ near zero and $(1-t)^{2/b-1}$ near one. The second endpoint
is integrable for every $b>0$. The first is integrable exactly when
\[
 2/a-2/q_0>-1.
\]
This holds for all $a>0$ when $q_0\ge2$ and exactly for $a<2$ when $q_0=1$;
at $q_0=1,a=2$ the divergence is logarithmic. Since the quantile is bounded,
its square integrability is automatic. This proves the Beta statements.
Failure of this energy criterion is not, by itself, a proof of predictive
inefficiency outside the corollary.
\end{proof}

\subsection{Spatial-scale modulation}
\label{app:scale-regularity}

Motivated by the spatial scale-mixture construction of
\citet{PalaciosSteel2006}, consider
\[
Y(s)=Z_1(s)\exp\{\sigma Z_2(s)\},\qquad \sigma\in\R.
\]
For this transformation,
\begin{equation}\label{eq:volatility-energies}
\mathcal D(F)=(1+\sigma^2)e^{2\sigma^2},\qquad
c_Y=c(1+\sigma^2)e^{2\sigma^2}.
\end{equation}
The covariance map and the exact latent-oracle MSPE are
\begin{align}
K_F(u)&=u\exp\{\sigma^2(1+u)\},\label{eq:KF-volatility}\\
\MSPEO&=e^{2\sigma^2}\left\{1-(1-v_n)e^{-\sigma^2v_n}\right\}.
\label{eq:volatility-oracle}
\end{align}
The derivatives of $z_1e^{\sigma z_2}$ are products of a polynomial of
degree at most one in $z_1$ and an exponential of a linear function of
$z_2$. They have finite Gaussian second moments at every fixed
integer derivative order, so the transformation satisfies the regularity
conditions for every finite latent smoothness $\nu$.
The two first-derivative contributions are
\[
 \E e^{2\sigma G_2}=e^{2\sigma^2},\qquad
 \E[\{\sigma G_1e^{\sigma G_2}\}^2]=\sigma^2e^{2\sigma^2}.
\]
For the covariance, independence gives
$\E(G_1G_{u,1})=u$ and
$\E e^{\sigma(G_2+G_{u,2})}=e^{\sigma^2(1+u)}$.
This proves the displayed covariance formula; its decrement gives the
oracle MSPE.

\subsection{Polynomial benchmarks}
\label{app:polynomial-regularity}

Two simple noninvertible benchmarks are the product field
\begin{equation}\label{eq:product-example}
Y(s)=Z_1(s)Z_2(s),\qquad \mathcal D(F)=2,\qquad K_F(u)=u^2,
\end{equation}
and the Gaussian-square field
\begin{equation}\label{eq:gamma-example}
Y(s)=\tau\sum_{j=1}^q Z_j(s)^2,\qquad
\tau>0,\qquad \mathcal D(F)=4q\tau^2.
\end{equation}
The latter has Gamma margins with shape $q/2$ and scale $2\tau$,
and covariance map $K_F(u)=2q\tau^2u^2$; the displayed representation
thus gives these particular Gamma shape values.

A polynomial of total degree $M_0$ has no Hermite terms of degree exceeding
$M_0$. This proves inclusion for every exponent in
Assumption~\ref{ass:Hermite} and~\eqref{eq:conservative-Hermite}.
The energies above follow directly from differentiation and
$\E G_j^2=1$; the covariance maps follow from independence and
$\Cov(G^2,G_u^2)=2u^2$. The exact oracle MSPEs follow from
\eqref{eq:oracle-risk}.

\section{Limiting distribution of the prediction error}
\label{app:prediction-error}

This appendix proves Proposition~\ref{prop:prediction-error-limit} and
verifies the examples and local rescaling in
Section~\ref{sec:error-calibration}. The proof uses Gaussian conditioning
and an exact $L^2$ identity for $W^{1,2}(\gamma_q)$ transformations,
rather than a separate central limit theorem.

\subsection{Gaussian conditioning and an exact error identity}
\label{app:error-limit-proof}

\begin{proof}[Proof of Proposition~\ref{prop:prediction-error-limit}]
Let $\boldsymbol G,\boldsymbol U$ be independent $N_q(0,I_q)$ vectors on
an auxiliary probability space. For $0<v<1$, set
\[
 X_v=\sqrt{1-v}\,\boldsymbol G+\sqrt v\,\boldsymbol U,
 \qquad
 T_vF=\frac{F(X_v)-\E\{F(X_v)\mid\boldsymbol G\}}{\sqrt v}.
\]
Gaussian conditioning implies that
$(\boldsymbol Z_\star,\{Y(s_\star)-\YO\}/\sqrt{v_n})$ has the same
joint distribution as $(X_{v_n},T_{v_n}F)$. Indeed, the latent conditional
mean has distribution $N_q\{0,(1-v_n)I_q\}$ and its residual is independent
$N_q(0,v_nI_q)$.

For every $H\in W^{1,2}(\gamma_q)$, the conditional-copy identity and
\eqref{eq:product-bound} give
\begin{equation}\label{eq:sobolev-error-contraction}
 \|T_vH\|_{L^2}^2
 =\frac{K_H(1)-K_H(1-v)}{v}
 \le\mathcal D(H).
\end{equation}
Put $L_F=\nabla F(\boldsymbol G)^{\mathsf T}\boldsymbol U$.
Independence gives $\|L_F\|_{L^2}^2=\mathcal D(F)$ and
$\E(L_F\mid\boldsymbol G)=0$. Gaussian integration by parts with
respect to $\boldsymbol U$ yields
\begin{align*}
 \E[(T_vF)L_F]
 &=\frac1{\sqrt v}\sum_{j=1}^q
   \E\!\left[F(X_v)\,\partial_jF(\boldsymbol G)\,U_j\right]\\
 &=\E\!\left[\nabla F(X_v)^{\mathsf T}\nabla F(\boldsymbol G)\right]
 =K_F'\!\left(\sqrt{1-v}\right).
\end{align*}
The conditional-mean term vanishes because $\E(L_F\mid\boldsymbol G)=0$;
the last equality is~\eqref{eq:KF-derivative-coupling}, since
$\Cov(X_v,\boldsymbol G)=\sqrt{1-v}\,I_q$. The integration-by-parts
identity holds first for polynomials and extends to $W^{1,2}(\gamma_q)$
by Hermite approximation at each fixed $v$: the bound
\eqref{eq:sobolev-error-contraction} and the identity
$\|L_H\|_{L^2}^2=\mathcal D(H)$ control the left-hand side, while
Cauchy--Schwarz and the standard Gaussian margins of $X_v$ and
$\boldsymbol G$ control the gradient product. Here
$L_H=\nabla H(\boldsymbol G)^{\mathsf T}\boldsymbol U$.

Expanding the squared norm now gives the exact identity
\begin{equation}\label{eq:prediction-error-L2-identity}
 \|T_vF-L_F\|_{L^2}^2
 =\frac{\Delta_F(v)}{v}+\mathcal D(F)
   -2K_F'\!\left(\sqrt{1-v}\right).
\end{equation}
By~\eqref{eq:KF-endpoint}, both $\Delta_F(v)/v$ and
$K_F'(\sqrt{1-v})$ tend to $\mathcal D(F)$ as $v\downarrow0$.
Thus $T_vF\to L_F$ in $L^2$. Since $X_v\to\boldsymbol G$ in $L^2$,
this also proves the joint limit with the latent target.
Proposition~\ref{prop:exact} gives $\MSPEO/v_n\to\mathcal D(F)$,
establishing the standardized oracle limit under the weaker assumptions
stated in the proposition.

For each observed-data predictor $A\in\{\mathrm{NL},\mathrm L,\rho\}$,
$\widehat Y_n^A$ is measurable with respect to $\mathcal H_n$.
Conditional-expectation orthogonality therefore gives
\begin{equation}\label{eq:oracle-error-transfer}
 \E[(\widehat Y_n^A-\YO)^2]
 =\operatorname{MSPE}_n^A-\MSPEO=o(v_n),
\end{equation}
where the last equality follows from the applicable MSPE-equivalence
result. Since $\operatorname{MSPE}_n^A/v_n\to\mathcal D(F)$, Slutsky's
theorem transfers the joint limit to each of these predictors.

Put $S=\|\nabla F(\boldsymbol G)\|^2/\mathcal D(F)$. Conditionally on
$\boldsymbol G$, the limit is centered normal with variance $S$, and
$\E S=1$. This proves its first two moments and
\eqref{eq:error-mixture-variance}. If $S=1$ almost surely, the limit is
standard normal. Conversely, its characteristic function is
$\E\exp(-t^2S/2)$. For any $t\ne0$, strict convexity implies
\[
 \E\exp(-t^2S/2)\ge\exp(-t^2/2),
\]
with equality only if $S$ is almost surely constant. A standard normal
limit forces equality, and its constant variance must be one. This proves
\eqref{eq:normal-error-criterion} without a fourth-moment assumption.
\end{proof}

\subsection{Examples and local log-Gaussian rescaling}
\label{app:error-limit-examples}

For the log-Gaussian field, $F'(G)=\omega e^{\mu+\omega G}$ and
$\mathcal D(F)=\omega^2e^{2\mu+2\omega^2}$ give
\eqref{eq:lognormal-error-limit}. Independence yields variance one and
fourth moment $3e^{4\omega^2}$. These are moments of the limiting law;
convergence of the fourth moments of the finite-sample errors is not
asserted. The skew-Gaussian calculation follows from the weak gradient displayed in
Section~\ref{sec:examples}, whose squared norm is
$\delta^2+\tau^2$ almost everywhere.

Finally, the joint convergence in
Proposition~\ref{prop:prediction-error-limit} and continuity of the
exponential imply, for each $A\in\{\mathrm{NL},\mathrm L,\rho\}$,
\[
 \left(Y(s_\star),
       \frac{Y(s_\star)-\widehat Y_n^A}{\sqrt{v_n}}\right)
 \xrightarrow{d}
 \left(e^{\mu+\omega G},\,
       \omega e^{\mu+\omega G}U\right).
\]
Since $\operatorname{MSPE}_n^A=O(v_n)$,
$\widehat Y_n^A-Y(s_\star)\to0$ in probability. We may replace the first
coordinate by $|\widehat Y_n^A|$ without changing its strictly positive
limit. To cover possible zero denominators at finite $n$, for every
$\eta>0$,
\[
 \PP(\widehat Y_n^A=0)
 \le\PP\{Y(s_\star)\le\eta\}
   +\PP\{|\widehat Y_n^A-Y(s_\star)|\ge\eta\}.
\]
First let $n\to\infty$, then $\eta\downarrow0$. Since $Y(s_\star)>0$
almost surely, this probability tends to zero, and assigning an arbitrary
value to the ratio on that event does not affect its limit. Division and
the continuous mapping theorem prove
\eqref{eq:lognormal-local-standardization}. The limiting standard normal
law gives marginal coverage $1-\eta$ for the corresponding intervals with
half-width $z_{1-\eta/2}\omega|\widehat Y_n^A|\sqrt{v_n}$. No convergence
of conditional predictive distributions is used or asserted.

\section{Counterexamples and the discontinuous boundary case}
\label{app:boundaries}

We prove Propositions~\ref{prop:kink}, \ref{prop:counterexample}, and
\ref{prop:probit-boundary} in the order in which they appear in
Section~\ref{sec:limits}. The first two delimit the regularity and
distributional assumptions of the positive results. The probit calculation
describes an exact oracle MSPE without asserting observed-data efficiency.

\subsection{A regularity counterexample within the latent-Gaussian class}
\label{app:kink}

We prove Proposition~\ref{prop:kink} using the fractional-regularity
calculation in Appendix~\ref{app:fractional-examples} and the
kriging-MSPE comparison in Appendix~\ref{app:comparison}.

\begin{proof}[Proof of Proposition~\ref{prop:kink}]
The map $F_\epsilon(z)=z+\epsilon|z|$ has positive slopes $1-\epsilon$
and $1+\epsilon$, hence is invertible and globally Lipschitz. Its mean is
$\epsilon\sqrt{2/\pi}$ and $\mathcal D(F_\epsilon)=1+\epsilon^2$.
Equation~\eqref{eq:absolute-Hermite} gives
$F_\epsilon\notin W^{3/2,2}(\gamma_1)$. Odd--even Gaussian symmetry makes
the cross-covariances of $G$ and $|G_u|$ vanish. With $A$ as defined in
\eqref{eq:skew-covariance}, therefore,
\[
 K_{F_\epsilon}(u)=u+\epsilon^2 A(u),\qquad
 C_Y(h)=\rho(h)+\epsilon^2 A(\rho(h)).
\]
Invertibility gives $\mathcal G_n=\mathcal H_n$. The exact oracle
identity yields
\begin{equation}\label{eq:kink-exact-MSPE}
 \MSPEN=\MSPEO=v_n+\epsilon^2\{A(1)-A(1-v_n)\}
          \sim(1+\epsilon^2)v_n.
\end{equation}

The latent correlation satisfies
\[
 \rho(h)=1-\frac{h^2}{2}+\frac{|h|^3}{3}+O(h^4).
\]
Integrating $A'(u)=2\arcsin(u)/\pi$ gives
\[
 A(1-v)=A(1)-v+\frac{4\sqrt2}{3\pi}v^{3/2}+O(v^{5/2}).
\]
Combining these expansions, the observed covariance obeys
\begin{equation}\label{eq:kink-covariance-expansion}
 C_Y(h)=C_Y(0)-\frac{1+\epsilon^2}{2}h^2
 +\left(\frac{1+\epsilon^2}{3}+\frac{2\epsilon^2}{3\pi}\right)|h|^3
 +O(h^4).
\end{equation}
To justify the passage from this local expansion to a spectral-density
tail, set
\[
 b_\epsilon=\frac{1+\epsilon^2}{3}+\frac{2\epsilon^2}{3\pi}.
\]
For $t\ge0$, put $g_\epsilon(t)=C_Y(t)$. Near zero on the positive
half-line, $1-\rho(t)=t^2a_0(t)$ with $a_0$ analytic and positive, and
the nonanalytic part of $A(1-v)$ is $v^{3/2}$ times an analytic function.
Thus $g_\epsilon$ has a convergent one-sided power series at zero, and
\eqref{eq:kink-covariance-expansion} gives
\[
 g_\epsilon'(0^+)=0,\qquad g_\epsilon'''(0^+)=6b_\epsilon.
\]
On $(0,\infty)$ the covariance is smooth, and it and its derivatives
through order four decay exponentially at infinity. Its fourth derivative
is bounded near zero by the one-sided analyticity, so
$g_\epsilon^{(4)}\in L^1(0,\infty)$.

Evenness and integrability of $C_Y$, with the convention
\eqref{eq:spectral-convention}, give
\[
 f_Y(\lambda)=\frac1\pi\int_0^\infty
                      g_\epsilon(t)\cos(\lambda t)\dd t.
\]
Four integrations by parts, using $g_\epsilon'(0^+)=0$ and the decay
at infinity, yield, for $\lambda\ne0$,
\[
 \lambda^4 f_Y(\lambda)
 =\frac{6b_\epsilon}{\pi}
  +\frac1\pi\int_0^\infty
          g_\epsilon^{(4)}(t)\cos(\lambda t)\dd t
 \longrightarrow\frac{6b_\epsilon}{\pi}
 \qquad(|\lambda|\to\infty),
\]
by the Riemann--Lebesgue lemma. Thus
\begin{equation}\label{eq:kink-spectral-tail}
 f_Y(\lambda)\sim\frac2\pi
       \left(1+\epsilon^2+\frac{2\epsilon^2}{\pi}\right)|\lambda|^{-4}.
\end{equation}
The latent spectral density is
$f(\lambda)=2/\{\pi(1+\lambda^2)^2\}$.
The representation~\eqref{eq:observed-spectral-density} and the
$L^\infty$ contraction in Lemma~\ref{lem:conv-bound} give uniform convergence
of the continuous Hermite spectral series for $A(\rho(h))$. The observed
spectral density is therefore continuous and strictly positive, since
this series has nonnegative coefficients and the term $f$ is present
separately. It follows from the tail and
positivity that $f_Y$ has a two-sided polynomial envelope.
Lemma~\ref{lem:spectral-comparison} therefore applies directly to $f,f_Y$,
without invoking the main theorem, and gives
\[
 \MSPEL/v_n\longrightarrow1+\epsilon^2+2\epsilon^2/\pi.
\]
Division by~\eqref{eq:kink-exact-MSPE} proves~\eqref{eq:kink-ratio}.
Integrability of $g_\epsilon^{(4)}$ on the positive half-line justifies
the spectral-density limit; the local covariance expansion alone is
not being used as a general Tauberian theorem.
\end{proof}

\subsection{A Poisson shot-noise counterexample}
\label{app:counterexample}

This section proves Proposition~\ref{prop:counterexample}.
The nonlinear upper bound uses the local shot-noise structure, whereas
the linear lower bound follows directly from the spectral density by
Cauchy--Schwarz and a compactly supported test function.

\begin{proof}[Proof of Proposition~\ref{prop:counterexample}]
Let $\{T_j\}$ be a stationary Poisson point process on $\R$ of intensity
$\lambda$, and define
\begin{equation}\label{eq:shotnoise}
 Y(t)=\sum_{T_j\le t}(t-T_j)e^{-a(t-T_j)},\qquad a>0.
\end{equation}
The kernel $k(u)=ue^{-au}\ind_{\{u\ge0\}}$ is continuous, bounded and
integrable. The series is almost surely finite-valued at each fixed $t$ and converges
locally uniformly: on a compact interval, points in a fixed bounded region
are finite in number, while the remote-past terms are bounded by a
summable Poisson integral of $C(1+|T_j|)e^{aT_j}$.
Thus $Y$ has continuous sample paths. Translation invariance of the Poisson
process gives stationarity. The Poisson exponential formula is finite for
every positive argument, since $k$ is bounded and integrable, so the
marginals have moments of every order. They are nonnegative and nondegenerate,
and hence non-Gaussian.

The first two moments of a Poisson integral give
\[
 \E Y(t)=\lambda\int_0^\infty ue^{-au}\dd u=\lambda/a^2,
\]
and
\begin{align*}
 C_Y(h)
 &=\lambda e^{-a|h|}\int_0^\infty u(u+|h|)e^{-2au}\dd u\\
 &=\frac{\lambda}{4a^3}(1+a|h|)e^{-a|h|}.
\end{align*}
Its spectral density, in the convention of this paper, is
\begin{equation}\label{eq:shotnoise-spectral-density}
 f_Y(\omega)=\frac{\lambda}{2\pi}(a^2+\omega^2)^{-2}.
\end{equation}

Write $h=1/N$, set $W(t)=e^{at}Y(t)$, and use the notation
$x_+=\max(x,0)$. From~\eqref{eq:shotnoise},
\[
 W(t)=\sum_j(t-T_j)_+e^{aT_j}.
\]
On every compact interval this is a convex, piecewise linear function,
with an upward slope jump of size $e^{aT_j}$ at each Poisson point. Define
an observed-data predictor by
\begin{equation}\label{eq:shotnoise-predictor}
 \widetilde Y_h=
 \max\{2W(-h)-W(-2h),\;2W(h)-W(2h)\}.
\end{equation}
Because $W(t)=e^{at}Y(t)$ and $a$ is known, these four values of $W$
are determined by the observed $Y$ values. The target is $Y(0)=W(0)$ and the
four required sites belong to $S_N$ for $N\ge2$.
By convexity, both secant extrapolations in~\eqref{eq:shotnoise-predictor}
are at most $W(0)$. If at most one Poisson point lies in $[-2h,2h]$, at
least one of the intervals $[-2h,0]$ and $[0,2h]$ contains no point in its
interior. The corresponding extrapolation is then exact, and so is their
maximum. Boundary coincidences have probability zero and can be ignored.

Let $N_h$ count the Poisson points in $[-2h,2h]$.
Terms from points outside this interval are either linear on the whole
interval or identically zero there. They contribute no extrapolation
error. Each inside point contributes at most $Ch e^{2ah}$ to either
error. For $h\le1$, therefore,
\[
 0\le Y(0)-\widetilde Y_h\le ChN_h\ind_{\{N_h\ge2\}}.
\]
Since $N_h$ is Poisson with mean $\mu_h=4\lambda h$,
\[
 \E[N_h^2\ind_{\{N_h\ge2\}}]
     =\mu_h^2+\mu_h-\mu_he^{-\mu_h}=O(h^2).
\]
Consequently,
\begin{equation}\label{eq:shotnoise-nonlinear-upper}
 \operatorname{MSPE}_N^{\mathrm{NL}}\le\E[\{Y(0)-\widetilde Y_h\}^2]=O(h^4).
\end{equation}

To obtain the linear MSPE order, the simple average of the observations
at $-h,h$ has MSPE
\[
 \tfrac32 C_Y(0)+\tfrac12 C_Y(2h)-2C_Y(h)
     =\frac{\lambda}{6}h^3+O(h^4).
\]
Thus $\operatorname{MSPE}_N^{\mathrm L}\le Ch^3$.

For the lower bound, take the compactly supported function
\[
 \psi(t)=
 \begin{cases}
  1-3t^2+2|t|^3,& |t|\le1,\\
  0,& |t|>1,
 \end{cases}
 \qquad \psi_h(t)=\psi(t/h).
\]
The function $\psi$ is $C^1$ and has a square-integrable weak second
derivative. Moreover, $\psi_h(0)=1$ and $\psi_h(s)=0$ for every
$s\in S_N$, since all observation sites lie outside $(-h,h)$.
For arbitrary real weights $w=(w_s)_{s\in S_N}$, define
\[
 r_w(\omega)=1-\sum_{s\in S_N}w_s e^{\mathrm i\omega s},
 \qquad
 Q_h(w)=\int_\R |r_w(\omega)|^2 f_Y(\omega)\dd\omega.
\]
Here $Q_h(w)$ is the true MSPE of
$m_Y+\sum_{s\in S_N}w_s\{Y(s)-m_Y\}$, with $m_Y=\lambda/a^2$.
Using the Fourier transform
$\widehat\psi_h(\omega)=(2\pi)^{-1}\int_\R
 e^{-\mathrm i\omega t}\psi_h(t)\dd t$, Fourier inversion gives
\[
 \int_\R r_w(\omega)\widehat\psi_h(\omega)\dd\omega
 =\psi_h(0)-\sum_{s\in S_N}w_s\psi_h(s)=1.
\]
This inversion is justified by the integrability of $\widehat\psi_h$,
which follows from compact support and the integrable weak second
derivative. Cauchy--Schwarz therefore implies
\[
 Q_h(w)\ge
 \left\{\int_\R
    \frac{|\widehat\psi_h(\omega)|^2}{f_Y(\omega)}\dd\omega
 \right\}^{-1}.
\]
Elementary integration gives
\[
 \int_\R(\psi'')^2\dd t=24,\qquad
 \int_\R(\psi')^2\dd t=\frac{12}{5},\qquad
 \int_\R\psi^2\dd t=\frac{26}{35}.
\]
By~\eqref{eq:shotnoise-spectral-density}, Parseval's identity and the
scaling $\psi_h(t)=\psi(t/h)$ yield
\begin{align*}
 \int_\R\frac{|\widehat\psi_h(\omega)|^2}{f_Y(\omega)}\dd\omega
 &=\frac1\lambda\int_\R
       \{a^4\psi_h^2+2a^2(\psi_h')^2+(\psi_h'')^2\}\dd t\\
 &=\frac{1}{\lambda h^3}
       \left\{24+\frac{24}{5}(ah)^2+\frac{26}{35}(ah)^4\right\}.
\end{align*}
The bound is uniform in $w$. Minimizing over all weights gives
\begin{equation}\label{eq:shotnoise-linear-lower}
 \operatorname{MSPE}_N^{\mathrm L}
 \ge\frac{\lambda h^3}
       {24+\frac{24}{5}(ah)^2+\frac{26}{35}(ah)^4}.
\end{equation}
In particular,
$\liminf_{N\to\infty}N^3\operatorname{MSPE}_N^{\mathrm L}
 \ge\lambda/24>0$.
Together with the upper bound, this gives
$\operatorname{MSPE}_N^{\mathrm L}\asymp h^3$. Combining with
\eqref{eq:shotnoise-nonlinear-upper} and $h=N^{-1}$ proves
\eqref{eq:counter-rates}.
\end{proof}

\subsection{Spatial probit: covariance and boundary behaviour}
\label{app:probit}\label{sec:infinite}

We prove Proposition~\ref{prop:probit-boundary} and give additional
covariance-increment and numerical-evaluation identities for the boundary
example in Section~\ref{sec:probit}.

\subsubsection{Covariance and latent-information MSPE}
\label{app:probit-risk}

For the decrement $\Delta_F$ defined in~\eqref{eq:Delta-definition},
the covariance-map identities
\begin{equation}\label{eq:Delta-identities}
 \MSPEO=\Delta_F(v_n),\qquad
 C_Y(0)-C_Y(h)=\Delta_F(1-\rho(h))
\end{equation}
remain valid when the endpoint slope is infinite. The domain $[0,2]$
in~\eqref{eq:Delta-definition} includes negative values of $\rho(h)$.
They identify a covariance decrement and a latent-information MSPE, not an
observed-data efficiency ratio.

\begin{proof}[Proof of Proposition~\ref{prop:probit-boundary}]
For two standard Gaussian variables $G,G_u$ with correlation $u$,
$\PP(G>a)=p_a$ and
\[
 \PP(G>a,G_u>a)=\Phi_2(-a,-a;u).
\]
Subtracting $p_a^2$ gives~\eqref{eq:probit-covariance-map}, and the value
at $u=1$ is the Bernoulli variance. Plackett's differentiation identity
\citep{Plackett1954} gives, for $-1<u<1$,
\begin{equation}\label{eq:probit-map-derivative}
 K_a'(u)=\phi_2(a,a;u)
 =\frac{1}{2\pi\sqrt{1-u^2}}
      \exp\left\{-\frac{a^2}{1+u}\right\}.
\end{equation}
The same identity follows by integrating
$\partial_u\phi_2(x,y;u)=\partial_x\partial_y\phi_2(x,y;u)$ over
$(-\infty,-a]^2$. Consequently, as $t\downarrow0$,
\[
 K_a'(1-t)\sim
 \frac{e^{-a^2/2}}{2\pi\sqrt{2}}\,t^{-1/2}.
\]
This singularity is integrable at zero. Hence, writing
$\Delta_a=\Delta_{F_a}$,
\begin{equation}\label{eq:probit-endpoint}
 \Delta_a(v)=\int_0^v K_a'(1-t)\dd t
 \sim\frac{e^{-a^2/2}}{\pi\sqrt{2}}\,v^{1/2}
 =\frac{\phi(a)}{\sqrt{\pi}}\,v^{1/2}.
\end{equation}
In particular, $K_a'(1^-)=\infty$. Since $F_a\in L^2(\gamma_1)$, the
Hermite characterization in Lemmas~\ref{lem:KF-identities} and
\ref{lem:derivatives} implies $F_a\notin W^{1,2}(\gamma_1)$.

Let $\boldsymbol Z_n=(Z(s_1),\ldots,Z(s_n))^{\mathsf T}$ and
$M_n=\boldsymbol r_n^{\mathsf T}R_n^{-1}\boldsymbol Z_n$. Gaussian
conditioning yields
\[
 Z(s_\star)\mid\mathcal H_n\sim N(M_n,v_n),\qquad
 \YO=\Phi\left(\frac{M_n-a}{\sqrt{v_n}}\right).
\]
Proposition~\ref{prop:exact} applies to the bounded transformation $F_a$
without a differentiability assumption and proves~\eqref{eq:probit-oracle}.
Since $v_n\to0$ under the latent spectral and dense-design assumptions,
\eqref{eq:probit-endpoint} proves~\eqref{eq:probit-oracle-asymptotic}.

For $a=0$, $K_0(0)=0$ and~\eqref{eq:probit-map-derivative} integrates to
\[
 K_0(u)=\frac{1}{2\pi}\arcsin u,\qquad -1\le u\le1.
\]
Therefore $K_0(1)=1/4$ and
$K_0(1)-K_0(1-v_n)=(2\pi)^{-1}\arccos(1-v_n)$.
The expansion $\arccos(1-v)\sim\sqrt{2v}$ completes the proof of
\eqref{eq:probit-symmetric}.
\end{proof}

\subsubsection{Covariance increments and the regularity boundary}
\label{app:probit-increments}

Combining~\eqref{eq:probit-endpoint} with the observed covariance gives
\begin{equation}\label{eq:probit-covariance-decrement}
 C_Y(0)-C_Y(h)
 \sim\frac{\phi(a)}{\sqrt{\pi}}\{1-\rho(h)\}^{1/2},
 \qquad \|h\|\to0.
\end{equation}
For a Mat\'ern latent correlation with smoothness $\nu$, its standard
small-lag expansion yields
\begin{equation}\label{eq:probit-increments}
 C_Y(0)-C_Y(h)\asymp
 \begin{cases}
  \|h\|^\nu, & 0<\nu<1,\\
  \|h\|\{\log(1/\|h\|)\}^{1/2}, & \nu=1,\\
  \|h\|, & \nu>1.
 \end{cases}
\end{equation}
Indeed, $1-\rho(h)$ is respectively of order $\|h\|^{2\nu}$,
$\|h\|^2\log(1/\|h\|)$ and $\|h\|^2$. These are statements about
second-order increments. They neither assert sample-path smoothness of the
binary field nor, without a separate argument, identify a spectral-density
tail. In particular the positive-theorem mechanism
$f_Y/f\to\mathcal D(F)\in(0,\infty)$ cannot be invoked here.

The indicator retains a nonzero first Hermite coefficient,
$\E[G\ind_{\{G>a\}}]=\phi(a)$. Its spectral series~\eqref{eq:observed-spectral-density} therefore
contains $\phi(a)^2f$, so a strictly positive latent spectral density gives strictly
positive observed covariance matrices as well. Thus the finite-sample
ordering in Proposition~\ref{prop:finite-gains} still applies. This
ordering and the exact oracle formula do not decide whether the additional
latent information retains an asymptotically nonnegligible advantage over binary observations.

\subsubsection{Stable evaluation of the latent-information MSPE}
\label{app:probit-evaluation}

Changing variables $u=\cos t$ in the derivative integral gives the
nonnegative, nonsingular expression
\begin{equation}\label{eq:probit-stable-oracle}
 \Delta_a(v)=\frac{1}{2\pi}
 \int_0^{\arccos(1-v)}
       \exp\left\{-\frac{a^2}{1+\cos t}\right\}\dd t,
 \qquad 0\le v\le1.
\end{equation}
The upper limit may be computed as $2\arcsin\sqrt{v/2}$ to avoid loss of
precision for small $v$. This evaluates the oracle MSPE at $v=v_n$
without subtracting nearly equal bivariate normal probabilities. At $a=0$
it reduces to $\Delta_0(v)=\pi^{-1}\arcsin\sqrt{v/2}$.

\section{Further constructions and extensions}
\label{app:additional-models}

This appendix records the extension to different latent covariances and
two additional ratio-field calculations. These are not needed for the
main proof. Except in Appendix~\ref{app:heterogeneous-latents}, latent
fields are independent, standardized, and share the correlation $\rho$.

\subsection{Independent latent fields with different covariances}
\label{app:heterogeneous-latents}

\begin{remark}\label{rem:heterogeneous-latents}
Suppose~\eqref{eq:L2-F} and Assumption~\ref{ass:design} hold, but let the
independent standardized latent fields have correlation functions
$\rho_1,\ldots,\rho_q$ and continuous, positive, even, unit-mass spectral
densities
\[
 f_j(\lambda)\sim c_j\|\lambda\|^{-d-2\nu},\qquad
 c_j>0,\quad j=1,\ldots,q,
\]
with a common $\nu>0$. Under Assumption~\ref{ass:Hermite}, assume that each
$f_j$ is radially nonincreasing. Alternatively, impose
\eqref{eq:conservative-Hermite} without radial monotonicity.
Let
\[
 v_{j,n}=\Var\{Z_j(s_\star)\mid Z_j(S_n)\},\qquad
 \mathcal D_{\boldsymbol c}(F)
 =\sum_{j=1}^q\frac{c_j}{c_1}\mathcal D_j(F),
 \quad\boldsymbol c=(c_1,\ldots,c_q).
\]
Write $\widehat Y_n^{\rho_1}$ for the observed-data kriging predictor whose
weights use $\rho_1$, and $\operatorname{MSPE}_n^{\rho_1}$ for its true
MSPE. The nonlinear, true-covariance, and oracle predictors retain their
definitions, with the present latent fields. Then
\begin{equation}\label{eq:heterogeneous-risk}
 \MSPEO\sim\MSPEN\sim\MSPEL\sim\operatorname{MSPE}_n^{\rho_1}
 \sim\mathcal D_{\boldsymbol c}(F)v_{1,n}.
\end{equation}
Thus a single latent correlation remains sufficient for first-order MSPE
efficiency. The common-correlation formulation in the main text avoids
carrying the component-specific parameters through the argument.
\end{remark}

\begin{proof}[Justification of Remark~\ref{rem:heterogeneous-latents}]
First, Lemma~\ref{lem:spectral-comparison}, applied to $f_1$ and $f_j$,
gives
\begin{equation}\label{eq:heterogeneous-variance-ratios}
 \frac{v_{j,n}}{v_{1,n}}\longrightarrow\frac{c_j}{c_1},
 \qquad j=1,\ldots,q.
\end{equation}
Indeed, continuity, positivity and the common tail order give the required
two-sided envelopes, while $f_j/f_1\to c_j/c_1$. In particular,
$v_{1,n}\to0$, and the variance ratios are uniformly bounded for large $n$
because $q$ is finite.

The conditional-copy calculation in Appendix~\ref{app:oracle-proof} now
has cross-covariance matrix
$\operatorname{diag}(1-v_{1,n},\ldots,1-v_{q,n})$. Applying the Hermite
pairing identity coordinatewise gives the exact formula
\begin{equation}\label{eq:heterogeneous-oracle}
 \MSPEO=\sum_{\bk\ne\boldsymbol0}b_{\bk}^2
 \left\{1-\prod_{j=1}^q(1-v_{j,n})^{k_j}\right\}.
\end{equation}
For numbers $0\le v_j\le1$ and nonnegative integers $k_j$,
\[
 0\le1-\prod_j(1-v_j)^{k_j}\le\sum_jk_jv_j.
\]
For every fixed multi-index, expansion of the finite product and
\eqref{eq:heterogeneous-variance-ratios} show that the braces in
\eqref{eq:heterogeneous-oracle}, divided by $v_{1,n}$, tend to
$\sum_jk_jc_j/c_1$. The displayed bound dominates these normalized terms
by $C|\bk|$ for all large $n$. Since
$\sum_{\bk}|\bk|b_{\bk}^2=\mathcal D(F)<\infty$, dominated convergence
yields
\[
 \frac{\MSPEO}{v_{1,n}}\longrightarrow
 \sum_{\bk\ne\boldsymbol0}b_{\bk}^2\sum_jk_j\frac{c_j}{c_1}
 =\mathcal D_{\boldsymbol c}(F).
\]

For the observed covariance and spectral density, the corresponding
coordinatewise identities are
\[
 C_Y(h)=\sum_{\bk\ne\boldsymbol0}b_{\bk}^2
                 \prod_j\rho_j(h)^{k_j},\qquad
 f_Y=\sum_{\bk\ne\boldsymbol0}b_{\bk}^2 f_{\bk},
\]
where $f_{\bk}$ is the convolution of $k_j$ copies of $f_j$ for each $j$;
factors with $k_j=0$ are omitted. There is at least one factor because
$\bk\ne\boldsymbol0$. Each $f_{\bk}$ has unit mass, so Tonelli's theorem
and $\sum_{\bk\ne\boldsymbol0}b_{\bk}^2=K_F(1)<\infty$ justify the
spectral series in $L^1$, just as in~\eqref{eq:observed-spectral-density}.
The finite family $f_1,\ldots,f_q$ has a common polynomial upper envelope. Lemma~\ref{lem:conv-tail} therefore gives
\[
 f_{\bk}(\lambda)\sim
 \left(\sum_jk_jc_j\right)\|\lambda\|^{-d-2\nu}
\]
for each fixed $\bk$. The uniform bounds in
Lemmas~\ref{lem:conv-bound} and~\ref{lem:refined-convolution} also apply
to these mixed convolutions. Their proofs require only a common envelope
for the independent frequency vectors; in the refined case every factor
is symmetric and radially nonincreasing. Consequently the dominating
weights are $C|\bk|^{d+2\nu+1}$ or $Cw_\nu(|\bk|)$, respectively.
The assumed Hermite summability permits the limit to pass through the
spectral series, giving
\[
 \frac{f_Y(\lambda)}{f_1(\lambda)}
 \longrightarrow\mathcal D_{\boldsymbol c}(F).
\]
The ratio is bounded: the series defines a continuous positive density,
and the limiting ratio is finite and positive. Lemma~\ref{lem:residual-energy}
then gives
$\operatorname{MSPE}_n^{\rho_1}/v_{1,n}\to\mathcal D_{\boldsymbol c}(F)$.
Finally,
$\MSPEO\le\MSPEN\le\MSPEL\le\operatorname{MSPE}_n^{\rho_1}$
and the two endpoint limits prove~\eqref{eq:heterogeneous-risk}.
\end{proof}

\subsection{An auxiliary Beta random field}
\label{app:beta-field}

The auxiliary Beta random field studied in
\citet[Section~2]{BevilacquaEtAl2024} is obtained as a ratio of two
independent Gamma fields built from squared Gaussian fields. In the present
notation, let
\[
 U(s)=\sum_{j=1}^{q_1} Z_j(s)^2,\qquad
 V(s)=\sum_{j=q_1+1}^{q_1+q_2} Z_j(s)^2,
\]
and define
\begin{equation}\label{eq:beta-field}
 Y(s)=\frac{U(s)}{U(s)+V(s)},
 \qquad q_1,q_2\in\N.
\end{equation}
The Gaussian fields are independent and share the correlation $\rho$.
The corresponding transformation is
\[
 F(\bz)=\frac{\sum_{j=1}^{q_1}z_j^2}{\sum_{j=1}^{q_1+q_2}z_j^2}
 \quad(\bz\ne\boldsymbol0),\qquad F(\boldsymbol0)=0.
\]
The marginal distribution is Beta$(q_1/2,q_2/2)$. This is the auxiliary
ratio field in that work, not the subsequent Clayton-like field obtained by
an additional marginal transformation. With $q=q_1+q_2>2$,
\begin{equation}\label{eq:beta-energy}
 \mathcal D(F)=\frac{4q_1q_2}{q(q+2)(q-2)}<\infty.
\end{equation}
Thus Corollary~\ref{cor:rough} covers this bounded noninvertible
construction under~\eqref{eq:radial-monotone} and $0<\nu<1$. At $q=2$ the energy diverges despite boundedness of the
response, so that theorem gives no conclusion. The Clayton-like construction in Section~\ref{sec:examples} applies
an additional transformation and allows prescribed marginal distributions.

\paragraph{Regularity check.}

For the bounded ratio at a fixed site, write $U=U(s)$, $V=V(s)$,
$S=U+V$, and $B=U/S$.
The independence of the two Gaussian-square sums gives independent
$S\sim\chi_q^2$ and $B\sim\operatorname{Beta}(q_1/2,q_2/2)$.
The gradient norm is
\[
 \norm{\nabla F}^2=\frac{4UV}{S^3}=\frac{4B(1-B)}S.
\]
Since
\[
 \E B(1-B)=\frac{q_1q_2}{q(q+2)},\qquad
 \E S^{-1}=\frac1{q-2},
\]
we obtain~\eqref{eq:beta-energy}. A cutoff at $S^{1/2}=0$ has squared
gradient error $O(\epsilon^{q-2})$, proving the weak derivative assertion
when $q>2$. At $q=2$, the angular factor has strictly positive expectation
and $\E S^{-1}=\infty$, so the Gaussian Sobolev energy diverges.

\subsection{Higher smoothness for the uniform Clayton-like field}
\label{app:clayton-smooth}

Corollary~\ref{cor:clayton} applies to the whole rough-latent regime,
not just to the exponential covariance. For smoother latents, a sufficient
condition for the \emph{uniform} construction is that $q_0$ be even and
that some integer $M$ satisfy
\begin{equation}\label{eq:clayton-general-condition}
 M>\nu,\qquad q_0+2>2M.
\end{equation}
Then the transformation has square-integrable weak derivatives through
order $M$, and Theorem~\ref{thm:main} applies under
\eqref{eq:radial-monotone}; the verification is given below. Without spectral
monotonicity, replacing $M>\nu$ by $M\ge d+2\nu+1$ gives the alternative
Theorem~\ref{thm:general}. Arbitrary additional marginal quantiles require
their own higher-order regularity checks at these smoothness values.

\paragraph{Regularity check.}

When $q_0$ is even, $F_{q_0}$ is smooth off the origin and homogeneous of degree
zero. Its derivatives of total order $k$ are bounded in absolute value by
$C_k\norm z^{-k}$, because their restrictions to the unit sphere are
smooth and bounded. In dimension $q=q_0+2$ these bounds are locally square
integrable through order $M$ when $q>2M$. Gaussian tails control infinity.
With the radial cutoff used above, every derivative error of order
$k\le M$ has squared Gaussian norm $O(\epsilon^{q-2k})$, tending to zero.
Thus these are Gaussian-square-integrable weak derivatives of $F_{q_0}$.
Lemma~\ref{lem:derivatives} and $M>\nu$ verify
Assumption~\ref{ass:Hermite}. If instead $M\ge d+2\nu+1$, they also
verify~\eqref{eq:conservative-Hermite} for the alternative theorem.

\printbibliography

\end{document}